\documentclass{article}
\PassOptionsToPackage{numbers, compress}{natbib}
\usepackage[final]{neurips_2022}

\usepackage[utf8]{inputenc} 
\usepackage[T1]{fontenc}    
\usepackage[hidelinks]{hyperref}       
\usepackage{url}            
\usepackage{booktabs}       
\usepackage{amsfonts}       
\usepackage{nicefrac}       
\usepackage{microtype}      
\usepackage{xcolor}         
\usepackage{algorithm}
\usepackage[noend]{algpseudocode}
\usepackage{amsmath, amssymb, amsthm}
\usepackage{mathtools}      
\usepackage{float}
\usepackage{enumitem}
\usepackage{xspace}
\usepackage{bm}
\usepackage[mathscr]{eucal} 
\usepackage{float}
\usepackage{wrapfig}
\usepackage{xfrac}

\usepackage{graphicx}
\usepackage{subfigure}

\usepackage[capitalize, nameinlink]{cleveref}
\usepackage{thm-restate}

\crefformat{equation}{(#2#1#3)}

\newcommand{\mat}[1]{\mathbf{#1}}

\newcommand{\tensor}[1]{\bm{\mathscr{#1}}}
\newcommand{\vectorize}{\textnormal{vec}}
\newcommand{\frobenius}{\textnormal{F}}
\newcommand{\MM}{\textnormal{MM}\xspace}

\newcommand{\var}{\textnormal{Var}}

\newcommand{\RowSampling}{\textnormal{\texttt{SampleRows}}\xspace}

\newcommand{\KronMatMul}{\textnormal{\texttt{KronMatMul}}\xspace}
\newcommand{\FastKroneckerRegression}{\textnormal{\texttt{FastKroneckerRegression}}\xspace}
\newcommand{\FastFactorMatrixUpdate}{\textnormal{\texttt{FastFactorMatrixUpdate}}\xspace}

\newcommand{\ktimes}{\otimes}

\newcommand{\nnz}{\textnormal{nnz}}
\newcommand{\opt}{\textnormal{opt}}
\newcommand{\rank}{\textnormal{rank}}

\newcommand{\R}{\mathbb{R}}

\newcommand{\E}{\mathbb{E}}

\DeclareMathOperator*{\argmin}{arg\,min}
\DeclareMathOperator*{\defeq}{\overset{def}{=}}
\renewcommand{\epsilon}{\varepsilon}

\DeclarePairedDelimiter{\set}{\{}{\}}
\DeclarePairedDelimiter{\parens}{(}{)}
\DeclarePairedDelimiter{\bracks}{[}{]}

\DeclarePairedDelimiter{\ceil}{\lceil}{\rceil}
\DeclarePairedDelimiter{\norm}{\lVert}{\rVert}

\newcommand{\cG}{\mathcal{G}}

 \newcommand{\cP}{\mathcal{P}}
 \newcommand{\cR}{\mathcal{R}}

\theoremstyle{plain}
\newtheorem{theorem}{Theorem}[section]
\newtheorem{lemma}[theorem]{Lemma}
\newtheorem{corollary}[theorem]{Corollary}

\theoremstyle{definition}
\newtheorem{definition}[theorem]{Definition}

\newtheorem{remark}[theorem]{Remark}

\title{Subquadratic Kronecker Regression with Applications to Tensor Decomposition}

\author{%
  Matthew Fahrbach\thanks{Authors are listed alphabetically. A preliminary version of this work that focuses on efficient sketching for Tucker decomposition appears in \url{arXiv:2107.10654}~\cite{fahrbach2021fast}.}\\
  Google Research\\
  \texttt{\href{mailto:fahrbach@google.com}{fahrbach@google.com}}
  \And
  Gang Fu\\
  Google Research\\
  \texttt{\href{mailto:thomasfu@google.com}{thomasfu@google.com}}
  \And
  Mehrdad Ghadiri\\
  Georgia Tech\\
  \texttt{\href{mailto:ghadiri@gatech.edu}{ghadiri@gatech.edu}}
}

\begin{document}

\maketitle

\begin{abstract}

Kronecker regression is a highly-structured least squares problem $\min_{\mathbf{x}} \lVert \mathbf{K}\mathbf{x} - \mathbf{b} \rVert_{2}^2$, where the design matrix $\mathbf{K} = \mathbf{A}^{(1)} \otimes \cdots \otimes \mathbf{A}^{(N)}$ is a Kronecker product of factor matrices. This regression problem arises in each step of the widely-used alternating least squares (ALS) algorithm for computing the Tucker decomposition of a tensor. We present the first \emph{subquadratic-time} algorithm for solving Kronecker regression to a $(1+\varepsilon)$-approximation that avoids the exponential term $O(\varepsilon^{-N})$ in the running time. Our techniques combine leverage score sampling and iterative methods. By extending our approach to block-design matrices where one block is a Kronecker product, we also achieve subquadratic-time algorithms for (1) Kronecker ridge regression and (2) updating the factor matrices of a Tucker decomposition in ALS, which is not a pure Kronecker regression problem, thereby improving the running time of all steps of Tucker ALS. We demonstrate the speed and accuracy of this Kronecker regression algorithm on synthetic data and real-world image tensors.
\end{abstract}

\section{Introduction}
\label{sec:introduction}
Tensor decomposition has a rich multidisciplinary history
with countless applications in data mining, machine learning,
and signal processing~\cite{kolda2009tensor,rabanser2017introduction,sidiropoulos2017tensor,jang2021fast}.
The most widely-used tensor decompositions are the CP decomposition
and the Tucker decomposition.
Similar to the singular value decomposition of a matrix,
both decompositions have natural analogs of \emph{low-rank} structure.
Unlike matrix factorization, however,
computing the rank of a tensor and the best rank-one tensor are NP-hard~\cite{hillar2013most}.
Therefore, most low-rank tensor decomposition algorithms decide on the
rank structure in advance, and then optimize the variables of the decomposition
to fit the data.
While conceptually simple, this approach is extremely effective in practice for many applications.

The \emph{alternating least squares} (ALS) algorithm is the main workhorse
for low-rank tensor decomposition, e.g., it is the first algorithm mentioned in the MATLAB Tensor Toolbox~\cite{matlab}.
For both CP and Tucker decompositions,
ALS cyclically optimizes disjoint blocks of variables while keeping all others fixed.
As the name suggests, each step solves a linear regression problem.
The \emph{core tensor} update step in ALS for Tucker decompositions is
notoriously expensive but highly structured.
In fact, the design matrix of this regression problem is the
Kronecker product of the factor matrices of the Tucker decomposition
$\mat{K} = \mat{A}^{(1)} \ktimes \cdots \ktimes \mat{A}^{(N)}$.
Our work builds on a line of Kronecker regression algorithms~\cite{diao2018sketching,diao2019optimal,marco2019least}
to give the first \emph{subquadratic-time} algorithm for solving Kronecker
regression to a $(1+\varepsilon)$-approximation while avoiding an
exponential term of $O(\varepsilon^{-N})$ in the running time.

We combine leverage score sampling, iterative methods, and a novel way of multiplying sparsified Kronecker product matrices
to fully exploit the Kronecker structure of the design matrix.
We also extend our approach to block-design matrices where one block is a
Kronecker product, achieving subquadratic-time algorithms
for (1) Kronecker ridge regression and (2) updating the factor matrix of a Tucker
decomposition in ALS, which is not a pure Kronecker regression problem.
Putting everything together, this work improves the running time of all steps of
ALS for Tucker decompositions and
runs in time that is sublinear in the size of the input tensor,
linear in the error parameter $\varepsilon^{-1}$, and subquadratic in the number
of columns of the design matrix in each step.
Our algorithms support L2 regularization in the Tucker loss function, so the
decompositions can readily be used in downstream learning tasks,
e.g., using the factor matrix rows as embeddings for clustering~\cite{zhou2014decomposition}.
Regularization also plays a critical role in the more general
tensor completion problem to prevent overfitting when data is missing
and has applications in differential privacy~\cite{chaudhuri2011differentially,balu2016differentially}.

The current-fastest Kronecker regression algorithm of~\citet{diao2019optimal}
uses leverage score sampling and achieves the following running times
for $\mat{A}^{(n)}\in\mathbb{R}^{I_n\times R_n}$ with $I_n\geq R_n$, for all $n\in [N]$,
where $R=\prod_{n=1}^N R_n$ and $\omega < 2.373$
denotes the matrix multiplication exponent~\cite{alman2021refined}:
\begin{enumerate}
    \item $\tilde{O}(\sum_{n=1}^N (\text{nnz}(\mat{A}^{(n)}) \mkern-4mu + \mkern-4mu R_n^\omega) \mkern-4mu + \mkern-4mu R^{\omega} \epsilon^{-1})$
    by sampling $\tilde{O}(R \varepsilon^{-1})$
    rows of $\mat{K}$ by their leverage scores.
    \item $\tilde{O}(\sum_{n=1}^N (\text{nnz}(\mat{A}^{(n)}) + R_n^\omega\epsilon^{-1})+ R \epsilon^{-N})$
    by sampling $\tilde{O}(R_n \epsilon^{-1})$ rows from each factor matrix $\mat{A}^{(n)}$
    and taking the Kronecker product of the sampled factor matrices.
\end{enumerate}

Note that the second approach is linear in $R$,
but the error parameter has an exponential cost
in the number of factor matrices.
In this work, we show that the running time of the first approach can be improved to subquadratic in $R$
without increasing the running time dependence on~$\epsilon$ in the dominant term,
simultaneously improving on both approaches.

\begin{theorem}
\label{thm:kronecker-regression}
For $n\in[N]$, let $\mat{A}^{(n)}\in\mathbb{R}^{I_n\times R_n}$, $I_n\geq R_n$,
and $\mat{b}\in\mathbb{R}^{I_1\cdots I_n}$.
There is a $(1+\varepsilon)$-approximation algorithm for
solving $\argmin_{\mat{x}} \norm{(\mat{A}^{(1)}\otimes \cdots \otimes \mat{A}^{(N)})\mat{x} - \mat{b}}_2^2$
that runs in time
\vspace{-0.5em}
\begin{equation}
\label{eqn:main_theorem_expression}
\tilde{O}\parens*{
    \sum_{n=1}^N \parens*{\textnormal{nnz}(\mat{A}^{(n)})+R_n^\omega N^2 \epsilon^{-2}}+ \min_{S\subseteq [N]} \textnormal{MM}\parens*{
        \prod_{n\in S} R_n, R \varepsilon^{-1}, \prod_{n\in [N]\setminus S} R_n
    }
},
\end{equation}
where $\textnormal{MM}(a,b,c)$ is the running time of multiplying an
$a \times b$ matrix with a $b \times c$ matrix.
\end{theorem}

If we do not use fast matrix multiplication (\citet{gall2018improved}~and~\citet{alman2021refined}),
the last term in \Cref{eqn:main_theorem_expression} is $\tilde{O}(R^2 \epsilon^{-1})$,
which is \emph{already an improvement} over the standard $\tilde{O}(R^3 \epsilon^{-1})$ running time.
With fast matrix multiplication, 
$\text{MM}(\prod_{n\in S} R_n, R \varepsilon^{-1}, \prod_{n\in [N]\setminus S} R_n)$
is subquadratic in $R$
for any nontrivial subset $S \not\in \{\emptyset, [N]\}$,
which is an improvement over $\tilde{O}(R^{\omega} \varepsilon^{-1}) \approx \tilde{O}(R^{2.373} \varepsilon^{-1})$.
If there exists a ``balanced'' subset $S$ such that
$\prod_{n\in S} R_n \approx \sqrt{R}$,
our running time goes as low as $\tilde{O}(R^{1.626} \varepsilon^{-1})$ using \cite{gall2018improved}.
For ease of notation,
we denote the subquadratic improvement by the constant $\theta^* > 0$,
where
$
R^{2 - \theta^*}= \min_{S\subseteq [N]} \text{MM}(\prod_{n\in S} R_n, R, \prod_{n\in [N]\setminus S} R_n)$.

Updating the core tensor in the ALS algorithm for Tucker decomposition is a pure Kronecker product regression as described in Theorem \ref{thm:kronecker-regression}, but updating the factor matrices is a regression problem of the form $\argmin_{\mat{x}}\norm{\mat{K}\mat{M}\mat{x}-\mat{b}}_2^2$, where $\mat{K}$ is a Kronecker product and $\mat{M}$ is a matrix without any particular structure.
We show that such problems can be converted to block regression problems
where one of the blocks is $\mat{K}$.
We then develop sublinear-time leverage score sampling techniques
for these block matrices,
which leads to the following theorem that accelerates all of the ALS steps.

\begin{theorem}
\label{thm:main_theorem}
There is an ALS algorithm for L2-regularized Tucker decompositions that
takes a tensor $\tensor{X} \in \R^{I_1 \times \cdots \times I_N}$
and returns $N$ factor matrices
$\mat{A}^{(n)} \in \R^{I_n \times R_n}$
and a core tensor $\tensor{G} \in \R^{R_1 \times \cdots \times R_n}$
such that each factor matrix and core update
is a $(1+\varepsilon)$-approximation to the optimum
with high probability.
The running times of each step are:
\begin{itemize}[leftmargin=2em]
    \item Factor matrix $\mat{A}^{(k)}$: $\tilde{O}(\sum_{n=1}^N (\textnormal{nnz}(\mat{A}^{(n)})+R_n^\omega N^2 \epsilon^{-2})+I_k R_{\ne k}^{2-\theta^*} \varepsilon^{-1} + I_k R\sum_{n=1}^N R_n +  R_{k}^\omega)$,
    \item Core tensor $\tensor{G}$: $\tilde{O}(\sum_{n=1}^N (\textnormal{nnz}(\mat{A}^{(n)})+R_n^\omega N^2 \epsilon^{-2})+R^{2-\theta^*} \varepsilon^{-1})$,
  \end{itemize}
  where $R = \prod_{n=1}^N R_n$, $R_{\ne k} = R / R_k$, and
  $\theta^* > 0$ is a constant derived from
  fast rectangular matrix multiplication.
\end{theorem}

For tensors of even modest order,
the superlinear term in $R$ is the bottleneck in many applications
since $R$ is exponential in the order of the tensor.
It follows that our improvements are significant in both theory and practice
as illustrated in our experiments in~\Cref{sec:experiments}.
{
\begin{table}[t]
\footnotesize
  \caption{
    Running times of \texttt{TuckerALS} 
    (\Cref{alg:alternating_least_squares})
    factor matrix and core tensor updates
    for ${\lambda = 0}$
    using different Kronecker regression methods.
    The factor matrices are denoted by $\mat{A}^{(n)} \in \R^{I_n \times R_n}$.
    The input tensor has size $I = I_1 \cdots I_N$
    and the core tensor has size $R = R_1 \cdots R_N$.
    Let $I_{\ne k} = I / I_k$ and $R_{\ne k} = R / R_k$.
    We use $\omega < 2.373$ for the matrix-multiplication exponent
    and the constant $\theta^* > 0$
    for the optimally balanced fast rectangular matrix
    multiplication as stated in~\Cref{thm:kron_mat_mul_sqrt_decomp},
    i.e.,
    $R^{2-\theta^*}  = 
    \min_{T \subseteq [N]} \text{MM}\parens{
      \prod_{n\in T} R_n, R, \prod_{n\notin T} R_n
    }$.
    Factors of $N$ are dropped for notational brevity.
  }
  \label{table:algorithm_comparison}
  \centering
  \setlength{\tabcolsep}{5pt}  
  \begin{tabular}{lcc}
    \toprule
    Algorithm & Factor matrix $\mat{A}^{(k)}$ & Core tensor $\tensor{G}$ \\
    \midrule
    Naive
    &
    $O\parens{I_k R R_{\ne k} + I_k R R_k
    + R_k^\omega + I R_{\ne k}}$
    &
    $O\parens{R^\omega + IR}$ \\
    This paper (\Cref{lemma:kron_mat_mul}) 
    &
    $O\parens{I_k R(\sum_{n=1}^N R_n) + R_k^\omega + I (\sum_{n \ne k} R_n) + I_k R_{k}^2}$
    &
    $O\parens{R^2 + I \sum_{n=1}^N R_n}$ \\
    This paper (\Cref{thm:main_theorem})
    &
    $\tilde{O}(I_k R_{\ne k}^{2-\theta^*} \varepsilon^{-1} + I_k R (\sum_{n=1}^N R_n) + R_{k}^\omega \varepsilon^{-2})$
    &
    $\tilde{O}(R^{2-\theta^*} \varepsilon^{-1})$ \\
    \citet{diao2019optimal} & --- & $\tilde{O}(R^{\omega} \varepsilon^{-2})$ \\
    \bottomrule
  \end{tabular}
\end{table}
}

\subsection{Our Contributions and Techniques}

We present several new results about approximate Kronecker regression
and the ALS algorithm for Tucker decompositions.
Below is a summary of our contributions:

\begin{enumerate}
    \item Our main technical contribution is the algorithm
    $\FastKroneckerRegression$ in~\Cref{sec:kronecker_regression}.
    This Kronecker regression
    algorithm builds on the block-sketching tools
    introduced in~\Cref{sec:row_sampling},
    and combines iterative methods with
    a fast novel Kronecker-matrix multiplication for sparse vectors and matrices and fast rectangular matrix multiplication to
    achieve a running time that is \emph{subquadratic} in the number of columns in the
    Kronecker matrix.
    A key insight is to use
    the original (non-sketched) Kronecker product
    as the preconditioner in the Richardson iterations
    when solving the sketched problem.
    This, by itself, improves the running time to quadratic. 
    Then to achieve subqudratic running time,
    we exploit the singular value decomposition of Kronecker products and present a novel method for multiplying a sparsified Kronecker product matrix (\Cref{lemma:kron_mat_mul} and \Cref{thm:kron_mat_mul_sqrt_decomp}).

  \item
  We generalize our Kronecker regression techniques
  to work for Kronecker ridge regression and
  the factor matrix updates in ALS for Tucker decomposition.
  We show that a factor matrix update is equivalent to solving an
  \emph{equality-constrained} Kronecker regression problem with a low-rank
  update to the preconditioner in the Richardson iterations.
  We can implement these new matrix-vector products nearly as fast
  by using the Woodbury matrix identity.
  Thus, we provably speed up each step of Tucker ALS, i.e., the core tensor and factor matrices.

\item
    We give a block-sketching toolkit in \Cref{sec:row_sampling}
    that states we can sketch blocks of a matrix by their leverage scores,
    i.e., their leverage scores in isolation, not with respect to the entire block matrix.
    This is one of the ways we exploit the Kronecker product structure
    of the design matrix.
    This approach can be useful for constructing spectral approximations
    and for approximately solving block regression problems.
    One corollary is that we can use the ``sketch-and-solve'' method
    for any ridge regression problem (\Cref{cor:approximate_ridge_regression}).

  \item
    We compare \FastKroneckerRegression with
    \citet[Algorithm 1]{diao2019optimal} on a synthetic Kronecker regression
    task studied in~\cite{diao2018sketching,diao2019optimal}
    and
    as a subroutine in ALS for computing the Tucker decomposition of
    image tensors~\cite{ma2021fast,nascimento2016spatial,nene1996columbia}.
    Our results show the importance of reducing the running time dependence
    on the number of columns in the Kronecker product.
\end{enumerate}

\subsection{Related Work}

\paragraph{Kronecker Regression.}
\citet{diao2018sketching} recently gave the first Kronecker regression algorithm
based on \texttt{TensorSketch}~\cite{pagh2013compressed}
that is faster than forming the Kronecker product.
\citet{diao2019optimal} improved this by removing the dependence on
$O(\text{nnz}(\mat{b}))$ from the running time, where $\mat{b} \in \R^{I_1 \cdots I_N}$ is
the response vector.
\citet*{reddy2022dynamic} recently initiated the study of \emph{dynamic} Kronecker regression,
where the factor matrices $\mat{A}^{(n)}$ undergo updates and the solution vector can be efficiently queried.
\citet*{marco2019least} studied the generalized Kronecker regression problem.
Very recently, \citet{ghadiri2023bit} analyzed the bit complexity of iterative methods with preconditioning for solving linear regression problems under fixed-point arithmetic. They show that the actual running time of such algorithms
(i.e., the number of bit operations) is at most a factor of $\log(\kappa) \cdot \log(1/\epsilon)$
more than the number of arithmetic operations, where $\kappa$ is the condition number of the design matrix and~$\epsilon$ is the error parameter.
This result applies to our paper, too;
however, for the rest of this exposition,
we discuss the number of arithmetic operations.
Finally, note that $\kappa(\mat{A} \otimes \mat{B})=\kappa(\mat{A}) \cdot \kappa(\mat{B})$.

\paragraph{Ridge Leverage Scores.}
\citet{ahmed2015fast} extended the notion of
statistical leverage scores to account for L2 regularization.
Sampling from approximate ridge leverage score distributions
has since played a key role in 
sparse low-rank matrix approximation~\cite{cohen2017input},
the Nystr\"om method~\cite{musco2017recursive},
bounding statistical risk in ridge regression~\cite{mccurdy2018ridge},
and ridge regression~\cite{chowdhury2018iterative, mccurdy2018ridge,li2019towards,kacham2022sketching}.
Fast recursive algorithms for computing
approximate leverage scores~\cite{cohen2015uniform} and
for solving overconstrained least squares~\cite{li2013iterative}
are also closely related.

\paragraph{Tensor Decomposition.}
\citet{cheng2016spals} and \citet{larsen2020practical}
used leverage score sampling to speed up ALS
for CP decomposition.\footnote{
The design matrix in each step of ALS for CP decomposition is a
Khatri--Rao product, not a Kronecker product.
CP decomposition does not suffer from a bottleneck
step like ALS for Tucker decomposition since it is a sparser decomposition,
i.e., CP decomposition does not have a core tensor---just factor matrices.
}
\citet{song2019relative}
gave a polynomial-time, relative-error approximation algorithm
for several low-rank tensor decompositions, which include CP and Tucker.
\citet{frandsen2020optimization} showed
that if the tensor has an exact Tucker decomposition,
then all local minima are globally optimal.
Randomized low-rank Tucker decompositions based on sketching have become
increasingly popular, especially in streaming applications:~\cite{malik2018low, traore2019singleshot, che2019randomized, sun2020low, jang2021fast, malik2021sampling, ma2021fast, ahmadi2021randomized}.
The more general problem of low-rank tensor completion
is also a fundamental approach for estimating the values of
missing data~\cite{acar2011scalable,liu2012tensor,jain2013low,jain2014provable,filipovic2015tucker}.
Fundamental algorithms for
tensor completion
are based on
ALS~\cite{zhou2013tensor,grasedyck2015variants,liu2020tensor},
Riemannian optimization~\cite{kressner2014low,kasai2016low,madhav2018dual},
or projected gradient methods~\cite{yu2016learning}.
Optimizing the core shape of a Tucker decomposition subject to a memory constraint
or reconstruction error guarantees has also been studied recently~\cite{ghadiri2023approximately,hashemizadeh2020adaptive,ehrlacher2021adaptive,xiao2021rank}.
\section{Preliminaries}
\label{sec:preliminaries}

\begin{wrapfigure}{R}{0.6\textwidth}
\begin{minipage}{0.6\textwidth}
\vspace{-0.82cm}
\begin{algorithm}[H]
\caption{\texttt{TuckerALS}}
\label{alg:alternating_least_squares}
\textbf{Input:} $\tensor{X} \in \R^{I_1 \times \cdots \times I_N}$,
 core shape $(R_1,R_2,\dots,R_N)$, $\lambda$

\begin{algorithmic}[1]
\State Initialize core tensor $\tensor{G} \in \R^{R_1 \times R_2 \times \dots \times R_n}$
\State Initialize factors
  $\mat{A}^{(n)} \in \R^{I_n \times R_n}$ for $n \in [N]$
    \Repeat
        \For{$n=1$ to $N$}
            \State $\mat{K} \gets \mat{A}^{(1)} \otimes \dots \otimes \mat{A}^{(n-1)} \otimes \mat{A}^{(n+1)} \otimes \dots \otimes \mat{A}^{(N)}$
            \State $\mat{B} \gets \mat{X}_{(n)}$
            \For{$i=1$ to $I_{n}$}
              \State $\mat{y}^* \hspace{-0.08cm}\gets\hspace{-0.04cm} \argmin_{\mat{y}} \norm{\mat{K}\mat{G}_{(n)}^\intercal\mat{y} - \mat{b}_{i:}^\intercal}_2^2 + \lambda\norm{\mat{y}}_2^2$
              \State Update factor row $\mat{a}^{(n)}_{i:} \gets
                {\mat{y}^*}^\intercal$
            \EndFor
        \EndFor
        \State $\mat{K} \gets \mat{A}^{(1)} \ktimes \mat{A}^{(2)} \ktimes \cdots \ktimes \mat{A}^{(N)}$ 
        \State $\mat{g}^* \gets \argmin_{\mat{g}}
             \norm{\mat{K} \mat{g} - \vectorize(\tensor{X})}_{2}^2
             + \lambda \norm{\mat{g}}_{2}^2$
        \State Update core tensor $\cG \gets \vectorize^{-1}(\mat{g}^*)$
    \Until{convergence}
    \State \textbf{return}
      $\tensor{G}, \mat{A}^{(1)}, \mat{A}^{(2)}, \dots, \mat{A}^{(N)}$ 
\end{algorithmic}
\end{algorithm}
\vspace{-1.20cm}
\end{minipage}
\end{wrapfigure}

\paragraph{Notation.}
The \emph{order} of a tensor is the number of its dimensions.
We denote scalars by normal lowercase letters $x \in \R$,
vectors by boldface lowercase letters $\mat{x} \in \R^{n}$,
matrices by boldface uppercase letters $\mat{X} \in \R^{m \times n}$,
and higher-order tensors by boldface script letters
$\tensor{X} \in \R^{I_1 \times I_2 \times \cdots \times I_N}$.
We use normal uppercase letters to denote
the size of an index set (e.g., $[N] = \{1,2,\dots,N\}$).
The $i$-th entry of a vector $\mat{x}$ is denoted by $x_i$,
the $(i,j)$-th entry of a matrix $\mat{X}$ by $x_{ij}$,
and the $(i,j,k)$-th entry of a third-order tensor $\tensor{X}$ by $x_{ijk}$.

\paragraph{Linear Algebra.}
Let $\mat{I}_{n}$ denote the $n \times n$ identity matrix and
$\mat{0}_{m \times n}$ denote the $m \times n$ zero matrix.
The transpose of $\mat{A} \in \R^{m \times n}$
is $\mat{A}^\intercal$, the Moore--Penrose inverse (also called pseudoinverse) is $\mat{A}^+$, and the spectral norm is $\norm{\mat{A}}_2$.
The singular value decomposition (SVD) of $\mat{A}$ is a factorization
of the form $\mat{U} \mat{\Sigma} \mat{V}^\intercal$,
where $\mat{U} \in \R^{m \times m}$ and $\mat{V} \in \R^{n \times n}$
are orthogonal matrices,
and $\mat{\Sigma} \in \R^{m \times n}$ is a non-negative real diagonal matrix.
The entries $\sigma_{i}(\mat{A})$ of $\mat{\Sigma}$
are the singular values of $\mat{A}$,
and the number of non-zero singular values is equal to $r = \rank(\mat{A})$.
The \emph{compact SVD} is a related decomposition 
where $\mat{\Sigma} \in \R^{r \times r}$ is a
diagonal matrix containing the non-zero singular values.
The Kronecker product of two matrices $\mat{A} \in \R^{m \times n}$
and $\mat{B} \in \R^{p \times q}$ is denoted by
$\mat{A} \ktimes \mat{B} \in \R^{(mp)\times(nq)}$.

\paragraph{Tensor Products.}
\emph{Fibers} of a tensor are the vectors we get by fixing all but one index.
If $\tensor{X}$ is a third-order tensor,
we denote the column, row, and tube fibers by
$\mat{x}_{:jk}$, $\mat{x}_{i:k}$, and $\mat{x}_{ij:}$,
respectively.
The \emph{mode-$n$ unfolding} of a tensor
$\tensor{X} \in \R^{I_1 \times I_2 \times \cdots \times I_N}$
is the matrix $\mat{X}_{(n)} \in \R^{I_n \times (I_1 \cdots I_{n-1}I_{n+1}\cdots I_N)}$
that arranges the mode-$n$ fibers of $\tensor{X}$ as columns of
$\mat{X}_{(n)}$ ordered lexicographically by index.
The \emph{vectorization} of $\tensor{X} \in \R^{I_1 \times I_2 \times \cdots \times I_N}$
is the vector $\vectorize(\tensor{X}) \in \R^{I_1 I_2 \cdots I_N}$
formed by vertically stacking the entries of $\tensor{X}$
ordered lexicographically by index.
For example, this transforms $\mat{X} \in \R^{m \times n}$ into a
tall vector $\vectorize(\mat{X})$ by stacking its columns.
We use $\vectorize^{-1}(\mat{x})$ to undo this
operation when it is clear from context what the shape of the output tensor should be.

The \emph{$n$-mode product} of tensor $\tensor{X} \in \R^{I_1\times I_2 \times \cdots \times I_N}$
and matrix $\mat{A} \in \R^{J \times I_n}$ is denoted by
$\tensor{Y} = \tensor{X} \times_{n} \mat{A}$ where
$\tensor{Y} \in \R^{I_1\times \cdots \times I_{n-1} \times J \times I_{n+1} \times \cdots \times I_N}$.
This operation multiplies each mode-$n$ fiber of
$\tensor{X}$ by the matrix $\mat{A}$.
This operation is expressed elementwise as
\[\textstyle{
    \parens*{\tensor{X} \times_{n} \mat{A}}_{i_1\dots i_{n-1} j i_{n+1} \dots i_{N}}
    =
    \sum_{i_n=1}^{I_n} x_{i_1 i_2 \dots i_N} a_{j i_n}.}
\]
The Frobenius norm $\norm{\tensor{X}}_{\frobenius}$
of a tensor $\tensor{X}$ is the square root of the sum of
the squares of its entries.

\paragraph{Tucker Decomposition.}
The \emph{Tucker decomposition} decomposes tensor
$\tensor{X} \in \R^{I_1 \times I_2 \times \cdots \times I_N}$
into a \emph{core tensor}
$\tensor{G} \in \R^{R_1 \times R_2 \times \cdots \times R_N}$
and $N$ \emph{factor matrices}
$\mat{A}^{(n)} \in \R^{I_n \times R_n}$.
Given a regularization parameter $\lambda\in\mathbb{R}_{\geq 0}$,
we compute a Tucker decomposition
by minimizing the nonconvex loss function
\begin{align*}
  L\parens{\tensor{G}, \mat{A}^{(1)}, \dots, \mat{A}^{(N)} ; \tensor{X}}
  &=
  \norm{\tensor{X} - \tensor{G} \times_{1}\mat{A}^{(1)}
     \cdots
     \times_{N} \mat{A}^{(N)}
  }_{\frobenius}^2
  + \lambda \parens*{
    \norm{\tensor{G}}_{\frobenius}^2
    +
    \sum_{n=1}^N
      \norm{\mat{A}^{(n)}}_{\frobenius}^2
  }.
\end{align*}
Entries of the reconstructed tensor
$\widehat{\tensor{X}} \defeq \tensor{G} \mkern-3mu \times_{1} \mkern-3mu \mat{A}^{(1)} \mkern-3mu
     \times_{2} \mkern-3mu \cdots
     \times_{N} \mat{A}^{(N)}$
are
\begin{align}
\label{eqn:tucker_decomposition_elementwise}
  \widehat{x}_{i_1 i_2 \dots i_N}
  =
  \sum_{r_1=1}^{R_1} 
  \cdots \sum_{r_N=1}^{R_N}
  g_{r_1 r_2 \dots r_N}
  a_{i_1 r_1}^{(1)} 
  \cdots
  a_{i_N r_N}^{(N)}.
\end{align}
Equation~\Cref{eqn:tucker_decomposition_elementwise}
demonstrates that $\widehat{\tensor{X}}$ is the sum of $R_1 \cdots R_N$ rank-1 tensors.
The tuple $(R_1, R_2, \dots, R_N)$ is the \emph{multilinear rank} of the
decomposition. The multilinear rank is typically chosen in advance
and much smaller than the dimensions of~$\tensor{X}$.

\paragraph{Alternating Least Squares.}
We present \texttt{TuckerALS}
in \Cref{alg:alternating_least_squares}
and highlight its connections to Kronecker regression.
The core tensor update (Lines~{10--12})
is a ridge regression problem
where the design matrix
$\mat{K}_{\text{core}} \in \R^{I_1 \cdots I_N \times R_1 \cdots R_N}$
is a Kronecker product of the factor matrices.
Each factor matrix update (Lines~{5--9})
also has Kronecker product structure,
but there are additional subspace constraints we must account for.
We describe these constraints in more detail in \Cref{sec:algorithm}.
\section{Row Sampling and Approximate Regression}
\label{sec:row_sampling}

Here we establish our sketching toolkit.
The \emph{$\lambda$-ridge leverage score} of the $i$-th row of 
$\mat{A} \in \R^{n \times d}$ is
\begin{equation}
\label{eqn:ridge_leverage_score_def}
    \ell_{i}^{\lambda}\parens*{\mat{A}}
    \defeq
    \mat{a}_{i:}\parens*{\mat{A}^\intercal \mat{A} + \lambda\mat{I}}^+ \mat{a}_{i:}^\intercal.
\end{equation}
The matrix of \emph{cross $\lambda$-ridge leverage scores} is
$\mat{A}(\mat{A}^\intercal \mat{A} + \lambda\mat{I})^+\mat{A}^\intercal$.
We denote its diagonal by $\bm{\mat{\ell}}^\lambda(\mat{A})$
because it contains the $\lambda$-ridge leverage scores of $\mat{A}$.
Ridge leverage scores generalize \emph{statistical leverage scores}
in that setting $\lambda = 0$ gives the
leverage scores of $\mat{A}$.
We denote the vector of statistical leverage scores by
$\bm{\mat{\ell}}(\mat{A})$.
If $\mat{A} = \mat{U} \mat{\Sigma}\mat{V}^\intercal$
is the compact SVD of $\mat{A}$, then for all $i \in [n]$, we have
\begin{align}
\label{eqn:ridge_leverage_score_svd_def}
  \textstyle{
  \ell_{i}^\lambda\parens*{\mat{A}}
  =
  \sum_{k=1}^{r}
    \frac{\sigma_{k}^2\parens*{\mat{A}}}{ \sigma_{k}^2\parens*{\mat{A}} + \lambda }
    u_{ik}^2,}
\end{align}
where $r = \rank(\mat{A})$.
It follows that every $\ell_{i}^\lambda(\mat{A}) \le 1$ since
$\mat{U}$ is an orthogonal matrix.
We direct the reader to
\citet{ahmed2015fast} or \citet{cohen2015uniform} for further details.

The main results in this paper
build on approximate leverage score sampling for block matrices.
The $\lambda$-ridge leverage scores of $\mat{A} \in \R^{n \times d}$
can be computed by appending $\sqrt{\lambda} \mat{I}_d$ to the bottom of $\mat{A}$
to get $\overline{\mat{A}} \in \R^{(n+d) \times d}$
and considering the leverage scores of $\overline{\mat{A}}$,
so we state the following results
in terms of statistical leverage scores without loss of generality.

\begin{definition}
For any $\mat{A} \in \R^{n \times d}$,
the vector $\hat{\bm{\ell}}(\mat{A}) \in \R^n$
is a \emph{$\beta$-overestimate} for the leverage score distribution of
$\mat{A}$ if, for all $i \in [n]$, it satisfies
\[
  \frac{\hat{\ell}_{i} \parens*{\mat{A}}}{
    \norm{\hat{\bm{\ell}} \parens*{\mat{A}}}_{1} }
  \ge
  \beta 
  \frac{\ell_{i} \parens*{\mat{A}}}{
    \norm{\bm{\ell} \parens*{\mat{A}}}_{1} }
  =
  \beta
  \frac{\ell_{i} \parens*{\mat{A}}}{\rank(\mat{A})}.
\]
\end{definition}

Next we describe the approximate leverage score sampling
algorithm in~\citet[Section 2.4]{woodruff2014sketching}.
The core idea here is that if we sample
$\tilde{O}(d / \beta)$ rows and reweight them appropriately,
this smaller \emph{sketched} matrix can be used instead of $\mat{A}$
to give provable guarantees for many problems.

\begin{definition}[Leverage score sampling]
\label{def:sample_rows_alg}
Let $\mat{A} \in \R^{n \times d}$ and $\mat{p} \in [0,1]^n$ be
a $\beta$-overestimate for the leverage score distribution of~$\mat{A}$ such that $\norm{\mat{p}}_1 = 1$.
$\RowSampling(\mat{A}, s, \mat{p})$ denotes the following procedure.
Initialize sketch matrix $\mat{S} = \mat{0}_{s \times n}$.
For each row $i$ of $\mat{S}$, independently and with replacement,
  select an index $j \in [n]$ with probability $p_j$
  and set $s_{ij} = 1/\sqrt{p_j s}$.
Return sketch $\mat{S}$.
\end{definition}

The main result in this section is that we can choose to sketch a single block
of a matrix by the leverage scores of that block in isolation.
This sketched submatrix can then be used with the other (non-sketched)
block to give a spectral approximation to the original matrix
or for approximate linear regression.
The notation $\mat{A} \preccurlyeq \mat{B}$ is the Loewner order
and means $\mat{B} - \mat{A}$ is positive semidefinite.

\begin{restatable}{lemma}{BlockSketchIsSpectralApprox}
\label{lemma:block_sketch_is_spectral_approx}
Let $\mat{A} = \begin{bmatrix} \mat{A}_{1} ; \mat{A}_{2} \end{bmatrix}$ be vertically
stacked
with $\mat{A}_1 \in \R^{n_1 \times d}$ and $\mat{A}_2 \in \R^{n_2 \times d}$.
Let $\mat{p} \in [0,1]^{n_1}$ be a $\beta$-overestimate for the leverage
score distribution of $\mat{A}_1$.
If $s > 144d \ln(2d / \delta)/(\beta \varepsilon^2)$,
the sketch $\mat{S}$ returned by
$\RowSampling(\mat{A}_1,s,\mat{p})$
guarantees, with probability at least $1-\delta$, that
\[
  (1-\varepsilon) \mat{A}^\intercal \mat{A}
  \preccurlyeq
  \parens*{\mat{S}\mat{A}_1}^\intercal \mat{S}\mat{A}_1
  + \mat{A}_{2}^\intercal \mat{A}_{2}
  \preccurlyeq
  (1+\varepsilon) \mat{A}^\intercal \mat{A}.
\]
\end{restatable}

\begin{restatable}[Approximate block regression]{lemma}{ApproximateBlockRegression}
\label{lemma:approximate_block_regression}
Consider the problem
$\argmin_{\mat{x} \in \R^{d}} \norm{\mat{A} \mat{x} - \mat{b}}_{2}^2$
where
$\mat{A} = \begin{bmatrix} \mat{A}_1 ; \mat{A}_2 \end{bmatrix}$
and
$\mat{b} = \begin{bmatrix} \mat{b}_1 ; \mat{b}_2 \end{bmatrix}$
are vertically stacked 
and $\mat{A}_1 \in \R^{n_1 \times d}$,
$\mat{A}_2 \in \R^{n_2 \times d}$, $\mat{b}_1 \in \R^{n_1}$,~$\mat{b}_2 \in \R^{n_2}$.
Let $\mat{p} \in [0,1]^{n_1}$ be a $\beta$-overestimate for the leverage
score distribution of $\mat{A}_1$.
Let $s \ge 1680 d \ln(40 d) / (\beta \varepsilon)$
and let $\mat{S}$ be the output of $\RowSampling(\mat{A}_1,s,\mat{p})$.
If
\[\textstyle{
  \mat{\tilde{x}}^*
  =
  \argmin_{\mat{x} \in \R^d}
  \parens*{
    \norm*{\mat{S}(\mat{A}_1 \mat{x} - \mat{b}_1)}_{2}^2
  +
  \norm*{\mat{A}_2\mat{x} - \mat{b}_2}_{2}^2
  },}
\]
then, with probability at least $9/10$, we have
\[
  \norm*{\mat{A} \mat{\tilde{x}}^* - \mat{b}}_{2}^2
  \le
  (1+\varepsilon)
  \min_{\mat{x} \in \R^{d}} \norm*{\mat{A} \mat{x} - \mat{b}}_{2}^2.
\]
\end{restatable}

We defer the proofs of these results to~\Cref{app:row_sampling}.
The key idea behind \Cref{lemma:approximate_block_regression}
is that leverage scores do not increase if rows are appended to the matrix.
This then allows us to prove a sketched submatrix version of
\citet[Lemma 8]{drineas2006fast} for
approximate matrix multiplication
and satisfy the structural conditions for
approximate least squares
in~\citet{drineas2011faster}.
One consequence is that we can ``sketch and solve'' 
ridge regression,
which was shown in~\cite[Theorem 1]{wang2017sketched}
and~\cite[Theorem 2]{avron2017sharper}.

\begin{restatable}{corollary}{SketchAndSolveRidgeRegression}
\label{cor:approximate_ridge_regression}
For any $\mat{A} \in \R^{n \times d}$,
$\mat{b} \in \R^d$, $\lambda \ge 0$,
consider
\[
  \argmin_{\mat{x} \in \R^d}
  \parens{
  \norm{\mat{A} \mat{x} - \mat{b}}_{2}^2
  + \lambda \norm{\mat{x}}_{2}^2}.
\]
Let $\mat{p} \in [0,1]^{n_1}$ be a $\beta$-overestimate for the leverage
scores of $\mat{A}$ and
$s \ge 1680 d \ln(40 d) / (\beta \varepsilon)$.
If $\mat{S}$ is the output of $\RowSampling(\mat{A}, s, \mat{p})$,
then, with probability at least $9/10$,
the sketched solution
\[
  \mat{\tilde x}^* =
  \argmin_{\mat{x} \in \R^d}
  \parens{
  \norm{\mat{S}\parens*{\mat{A} \mat{x} - \mat{b}}}_{2}^2
  + \lambda \norm{\mat{x}}_{2}^2}
\]
gives a $(1+\varepsilon)$-approximation to the original problem.
\end{restatable}

\begin{remark}
The success probability of the sketch can be boosted from $9/10$ to $1-\delta$
by sampling a factor of $O(\log(1/\delta))$ more rows.
See the discussion in~\citet[Section 2]{chen2019active}
about matrix concentration bounds for more details.
\end{remark}

\section{Kronecker Regression}
\label{sec:kronecker_regression}

Now we describe the key ingredients that allow us to design an approximate Kronecker
regression algorithm whose running time is \emph{subquadratic} in the number
of columns in the design matrix.
\begin{enumerate}[leftmargin=2em]
\itemsep0em 
  \item The leverage score distribution of a Kronecker product matrix
    $\mat{K} = \mat{A}^{(1)} \ktimes \cdots \ktimes \mat{A}^{(N)}$
    is a \emph{product distribution} of the leverage score distributions of
    its factor matrices.
    Therefore, we can sample rows of $\mat{K}$ from $\bm{\ell}(\mat{K})$
    with replacement in $\tilde{O}(N)$ time after a preprocessing step.
  \item The normal matrix $\mat{K}^\intercal \mat{K} + \lambda \mat{I}$
    in the ridge regression problem
    $\min_{\mat{x}} \norm{\mat{K}\mat{x} - \mat{b}}_2^2 + \lambda \norm{\mat{x}}_2^2$
    is a $O(1)$-spectral approximation of the sketched matrix
    $(\mat{S}\mat{K})^\intercal \mat{S}\mat{K} + \lambda \mat{I}$
    by \Cref{lemma:block_sketch_is_spectral_approx}.
    Thus we can use Richardson iteration with
    $(\mat{K}^\intercal \mat{K} + \lambda\mat{I})^+$
    as the preconditioner
    to \emph{solve the sketched instance}, which
    guarantees a $(1+\varepsilon)$-approximation.
    Using $(\mat{K}^\intercal \mat{K} + \lambda\mat{I})^+$
    as the preconditioner
    allows us to \emph{heavily exploit the Kronecker structure}
    with fast matrix-vector multiplications.
  \item At this point, \emph{Kronecker matrix-vector multiplications} are still
    the bottleneck,
    so we partition the factor matrices into two groups by their number of columns
    and use our novel way of multiplying sparsified Kronecker product matrices as well as fast rectangular matrix multiplication to get a subquadratic running time.
\end{enumerate}

This first result shows how $\lambda$-ridge leverage scores of a Kronecker
product matrix decompose according to the SVDs of its factor matrices.
All missing proofs in this section are
deferred to~\Cref{app:kronecker_regression}.

\begin{restatable}{lemma}{KroneckerCrossLeverageScores}
\label{lemma:kronecker_cross_leverage_scores}
Let 
$\mat{K} = \mat{A}^{(1)} \otimes \mat{A}^{(2)} \otimes \dots \otimes \mat{A}^{(N)}$,
where each factor matrix $\mat{A}^{(n)} \in \R^{I_n \times R_n}$.
Let $(i_1,i_2,\dots,i_N)$ be the natural row indexing of $\mat{K}$
by its factors.
Let the factor SVDs be
  $\mat{A}^{(n)}
  = \mat{U}^{(n)} \mat{\Sigma}^{(n)} {\mat{V}^{(n)}}^\intercal $.
For any $\lambda \ge 0$,
the $\lambda$-ridge leverage scores of $\mat{K}$ are
\begin{align}
\label{eqn:kronecker_ridge_leverage_score}
    \ell_{(i_1,\dots,i_N)}^\lambda \parens*{\mat{K}}
    &=
    \sum_{\mat{t} \in T}
    \frac{\prod_{n=1}^N \sigma_{t_n}^2(\mat{A}^{(n)})  }{\prod_{n=1}^N \sigma_{t_n}^2(\mat{A}^{(n)}) + \lambda}
    \parens*{
      \prod_{n=1}^N
      u_{i_n t_n}^{(n)}
    }^2,
\end{align}
where the sum is over $T = [R_1] \times [R_2] \times \dots \times [R_N]$.
For statistical leverage scores, this simplifies to
$
  \ell_{(i_1,\dots,i_N)}\parens{\mat{K}}
  =
  \prod_{n=1}^N
  \ell_{i_n} \parens{\mat{A}^{(n)}}.
$
\end{restatable}

This proof
repeatedly uses the mixed-product property for Kronecker products and the
definition of $\lambda$-ridge leverage scores in
Equation~\Cref{eqn:ridge_leverage_score_def}.

\subsection{Iterative Methods}

Now we state a result for the convergence rate of
preconditioned Richardson iteration~\cite{saad2003iterative}
using the notation
$\norm{\mat{x}}_{\mat{M}}^2 = \mat{x}^\intercal \mat{M} \mat{x}$.

\begin{lemma}[Preconditioned Richardson iteration]
\label{lemma:richardson_iteration}
Let $\mat{M}$ be any matrix such that
$\mat{A}^\intercal \mat{A} \preccurlyeq
\mat{M} \preccurlyeq \kappa \cdot \mat{A}^\intercal \mat{A}$
for some $\kappa \ge 1$.
Let $\mat{x}^{(k+1)} = \mat{x}^{(k)} - \mat{M}^{+}\parens{\mat{A}^\intercal \mat{A} \mat{x}^{(k)} - \mat{A}^\intercal \mat{b}}$.
Then,
\[
  \norm{
    \mat{x}^{(k)}
    -
    \mat{x}^*
  }_{\mat{M}}
  \le
  \parens*{1 - 1/\kappa }^k
  \norm{
    \mat{x}^{(0)} - \mat{x}^*
  }_{\mat{M}},
\]
where $\mat{x}^* = \argmin_{\mat{x} \in \R^{d}} \norm{\mat{A}\mat{x} - \mat{b}}_2^2$.
\end{lemma}

\begin{remark}
The ridge regression algorithm in
\citet{chowdhury2018iterative}
is also based on sketching and preconditioned Richardson iteration.
They consider short and wide matrices where $d \gg n$
and use the \emph{sketched normal matrix as the preconditioner}
to solve the original problem.
One of our main technical contributions is to
use the \emph{original normal matrix as the preconditioner}
to solve the sketched problem.
Reversing this
is advantageous because
computing the pseduoinverse and matrix-vector products with
the original Kronecker matrix
is substantially less expensive due to its Kronecker structure.
However,
this still motivates the need for faster Kronecker matrix-vector multiplications.
\end{remark}

\subsection{Fast Kronecker-Matrix Multiplication}

The next result is a simple but useful observation about
extracting the rightmost factor matrix from the Kronecker product
and recursively computing a new less expensive
Kronecker-matrix multiplication.

\begin{restatable}{lemma}{KronMatMulLemma}
\label{lemma:kron_mat_mul}
Let $\mat{A}^{(n)} \in \R^{I_n \times J_n}$, for $n \in [N]$,
and $\mat{B} \in \R^{J_1 \cdots J_N \times K}$.
There is an algorithm
$\KronMatMul([\mat{A}^{(1)},\dots,\mat{A}^{(N)}],\mat{B})$ that
computes
$
  \parens{\mat{A}^{(1)} \ktimes \mat{A}^{(2)} \ktimes \cdots \ktimes \mat{A}^{(N)}} \mat{B}
  \in \R^{(I_1 \cdots I_N) \times K}
$
in $O\parens{K \sum_{n = 1}^N J_1 \cdots J_n I_n \cdots I_N}$ time.
\end{restatable}

The following theorem is more sophisticated.
We write the statement in terms of
rectangular matrix multiplication time $\text{MM}\parens{a,b,c}$,
which is the time to multiply an
$a \times b$ matrix by a
$b \times c$ matrix.

\begin{restatable}{theorem}{KronMatMulSqrtDecomp}
\label{thm:kron_mat_mul_sqrt_decomp}
Let $\mat{A}^{(n)} \in\R^{I_n\times R_n}$, for $n \in [N]$,
$I = I_1 \cdots I_N$, $R = R_1 \cdots R_N$,
$\mat{b}\in\R^{I}$, $\mat{c}\in \R^{R}$,
and
$\mat{S}\in \mathbb{R}^{I \times I}$ be a diagonal
matrix with $\tilde{O}(R \varepsilon^{-1})$ nonzeros.
The vectors
\[
  \parens*{\mat{A}^{(1)} \ktimes \cdots \ktimes \mat{A}^{(N)}}^{\intercal} \mat{S} \mat{b}
  ~~~~\text{and}~~~~
  \mat{S} \parens*{\mat{A}^{(1)} \ktimes \cdots \ktimes \mat{A}^{(N)}} \mat{c}
\]
can be computed in time
$
  \tilde{O}\parens{
    \min_{T \subseteq [N]} \MM \parens{
      \prod_{n\in T} R_n, R \varepsilon^{-1}, \prod_{n\notin T} R_n
    }
  }.
$
\end{restatable}

The core idea behind \Cref{thm:kron_mat_mul_sqrt_decomp} is 
that the factor matrices can be partitioned
into two groups to achieve a good ``column-product'' balance,
i.e.,
$\min_{T \subseteq [N]}\max\set{\prod_{n \in T}R_n, \prod_{n \not\in T}R_n}$ is close to $\sqrt{R}$.
Then we use the fact that $\nnz(\mat{S}) = \tilde{O}(R \varepsilon^{-1})$
with a sparsity-aware $\KronMatMul$ to solve each part of this partition
separately,
and combine them with fast rectangular matrix multiplication.
If we achieve perfect balance,
the running time is $\tilde{O}(R^{1.626} \varepsilon^{-1})$ using results of~\citet{gall2018improved},
which are explained in detail in \citet[Appendix C]{van2019dynamic}.
If one of these two factor matrix groups has at most $0.9$
of the ``column-product mass,''
the running time is $\tilde{O}(R^{1.9} \varepsilon^{-1})$.

\subsection{Main Algorithm}

We are now ready to present our main algorithm for solving approximate
Kronecker regression.

\begin{algorithm}[t]
\caption{\FastKroneckerRegression}
\label{alg:fast_kronecker_regression}
  \textbf{Input:}
  Factor matrices $\mat{A}^{(n)} \in \R^{I_n \times R_n}$,
  response vector $\mat{b} \in \R^{I_1 \cdots I_N}$,
  L2 regularization strength $\lambda$, error $\varepsilon$, failure probability $\delta$

\begin{algorithmic}[1]
  \State Set $R \gets R_1 R_2 \cdots R_N$
  \For{$n=1$ to $N$}
    \State Compute a spectral approximation $\mat{\tilde{A}}^{(n)}$ with $\tilde{O}(R_n N^2 \epsilon^{-2})$ rows by \cref{lemma:block_sketch_is_spectral_approx} such that 
    \begin{align}
     {\mat{A}^{(n)^\intercal}} {\mat{A}}^{(n)}
     \preccurlyeq
     {{\mat{\tilde{A}}}^{(n)^\intercal}} {\mat{\tilde{A}}}^{(n)}
     \preccurlyeq
     (1+\log(1+\epsilon/4)/N) {\mat{A}^{(n)^\intercal}} {\mat{A}}^{(n)}
    \end{align}
    \State Compute ${\mat{\tilde{A}}^{(n)^\intercal}} {\mat{\tilde{A}}}^{(n)}$
    and the SVD of ${\mat{\tilde{A}}^{(n)^\intercal}} \mat{\tilde{A}}^{(n)} = \mat{V}^{(n)} ({{\mat{\Sigma}^{(n)^\intercal}}} \mat{\Sigma}^{(n)}) {{\mat{V}^{(n)^\intercal}}}$
    \State Compute $(1+\log(1+\epsilon/2)/N)$-approximate leverage scores $\bm{\ell}(\mat{A}^{(n)})$ using \Cref{lemma:fast-ls-comp}
    by applying a random
    Johnson--Lindenstrauss projection
  \EndFor
  \State Initialize product distribution data structure
         $\cP$ to sample indices from
         $(\bm{\ell}(\mat{A}^{(1)}), \cdots, \bm{\ell}(\mat{A}^{(N)}))$
  \State Set $\mat{D} \gets ({\mat{\Sigma}^{(1)}}^\intercal \mat{\Sigma}^{(1)}
                  \ktimes \cdots \ktimes {\mat{\Sigma}^{(N)}}^\intercal \mat{\Sigma}^{(N)} + \lambda\mat{I}_{R})^+$
  \State Let $\mat{M}^+ = (\mat{V}^{(1)} \ktimes \cdots \ktimes \mat{V}^{(N)}) \mat{D}
        (\mat{V}^{(1)} \ktimes \cdots \ktimes \mat{V}^{(N)})^\intercal$ \label{line:preconditioner}
  \State Set $s \gets \ceil{1680 R \ln(40R) \ln(1/\delta) / \varepsilon}$
  \State Set $\mat{S} \gets \RowSampling(\mat{K}, s, \cP)$
  \State Let $\mat{\tilde K} = \mat{S}\mat{K}$ and $\mat{\tilde b} = \mat{S}\mat{b}$
  \State Initialize $\mat{x} \gets \mat{0}_{R}$
  \Repeat
    \State $\mat{x}
        \gets \mat{x} - (1-\sqrt{\varepsilon})\mat{M}^+\parens{\mat{\tilde K}^\intercal \mat{\tilde K}\mat{x} + \lambda \mat{x} - \mat{\tilde K}^\intercal\mat{\tilde b}}$
        using fast Kronecker-matrix multiplication\label{line:richardson_step}
  \Until{convergence} \\
  \Return $\mat{x}$
\end{algorithmic}
\end{algorithm}

\begin{restatable}{theorem}{FastKroneckerRegressionTheorem}
\label{thm:fast_kronecker_regression}
For any Kronecker product matrix
$
    \mat{K}
  = \mat{A}^{(1)} \ktimes \cdots \ktimes \mat{A}^{(N)}
  \in \R^{I_1 \cdots I_N \times R_1 \cdots R_N},
$
$\mat{b} \in \R^{I_1 \cdots I_N}$,
$\lambda \ge 0$,
$\varepsilon \in (0,1/4]$, and $\delta > 0$,
\FastKroneckerRegression returns
$\mat{x}^* \in \R^{R_1\cdots R_N}$
in
\[
\textstyle{
\tilde{O}\parens*{
    \sum_{n=1}^N \parens*{\textnormal{nnz}(\mat{A}^{(n)})+R_n^\omega N^2 \epsilon^{-2}}+ \min_{S\subseteq [N]} \textnormal{MM}\parens*{
        \prod_{n\in S} R_n, R \varepsilon^{-1}, \prod_{n\in [N]\setminus S} R_n
    }
},}
\] 
time
such that, with probability at least $1-\delta$,
\[
  \norm{\mat{K}\mat{x}^* -\mat{b}}_2^2 + \lambda \norm{\mat{x}}_{2}^2
  \le
  (1+\varepsilon) \min_{\mat{x}}\norm{\mat{K}\mat{x}-\mat{b}}_2^2 + \lambda \norm{\mat{x}}_{2}^2.
\]
\end{restatable}

We defer the proof to~\Cref{app:main_algorithm}
and sketch how the ideas in \Cref{alg:fast_kronecker_regression} come together.
First, we do not compute the pseudoinverse $\mat{\tilde{K}}^+$
but instead use iterative Richardson iteration (\Cref{lemma:richardson_iteration}),
which allows us avoid a $\tilde{O}(R^\omega \varepsilon^{-1})$ running time.
This technique by itself, however, only allows us to reduce the running time to $\tilde{O}(R^2 \epsilon^{-1})$
since all of the matrix-vector products
(e.g., $\mat{\tilde K}^\intercal \mat{\tilde{b}}$, $\mat{\tilde K}\mat{x}$, and multiplication against $\mat{M}^+$)
naively take $\Omega(R^2)$ time.
To achieve subquadratic time, we need three more ideas:
(1) compute an approximate SVD of each Gram matrix ${\mat{A}^{(n)}}^\intercal \mat{A}^{(n)}$
in order to construct the decomposed preconditioner $\mat{M}^+$;
(2) use fast Kronecker-vector multiplication (e.g., \Cref{lemma:kron_mat_mul})
to exploit the Kronecker structure of the decomposed preconditioner;
(3) noting that \cref{lemma:kron_mat_mul} for the Kronecker-vector products
$\mat{\tilde K}^\intercal \mat{\tilde b}$ and $\mat{\tilde K}^\intercal (\mat{\tilde K}\mat{x})$
is insufficient because the intermediate vectors can be large,
we develop a novel multiplication algorithm in \Cref{thm:kron_mat_mul_sqrt_decomp}
that fully exploits the sparsity, Kronecker structure,
and fast rectangular matrix multiplication of~\citet{gall2018improved}.

\section{Applications to Low-Rank Tucker Decomposition}
\label{sec:algorithm}

Now we apply our fast Kronecker regression algorithm
to $\texttt{TuckerALS}$ and prove \Cref{thm:main_theorem}.
We list the running times of different factor matrix and core update algorithms
in Table~\ref{table:algorithm_comparison},
and we analyze these subroutines in~\Cref{subsec:different_algorithms_running_times}.

\paragraph{Core Tensor Update.}
The core update running time
in \Cref{thm:main_theorem}
is a direct consequence of our algorithm for fast Kronecker regression in
\Cref{thm:fast_kronecker_regression}.
The only difference is that we avoid recomputing the SVD
and Gram matrix of each factor since these are
computed at the end of each factor matrix update and stored for future use.

\paragraph{Factor Matrix Update.}
The factor matrix updates require more work
because of the $\mat{G}_{(n)}^\intercal \mat{y}$ term in
Line~8 of \texttt{TuckerALS}.
To overcome this, we substitute variables and recast
each factor update as an equality-constrained
Kronecker regression problem with an appended low-rank block
to account for the L2 regularization of the original variables.
To support this new low-rank block, we use the \emph{Woodbury matrix identity}
to extend the technique of using Richardson iterations with
fast Kronecker matrix-vector multiplication for solving sketched regression instances.

The next result formalizes this substitution
and
reduces the problem to block Kronecker regression with
a subspace constraint.
This result relies on the fact that the least squares solution to $\norm{\mat{M}\mat{x} - \mat{z}}_{2}^2$ with minimum norm is $\mat{M}^+ \mat{z}$.

\begin{restatable}{lemma}{LemmaConstrainedLeastSquares}
\label{lemma:constrained_least_squares}
Let $\mat{A} \in \R^{n \times m}$, $\mat{M} \in \R^{m \times d}$,
$\mat{b} \in \R^{n}$, and $\lambda \ge 0$.
For any ridge regression problem of the form
$\argmin_{\mat{x} \in \R^{d}} (\norm*{\mat{A}\mat{M}\mat{x} - \mat{b}}_2^2 + \lambda \norm*{\mat{x}}_2^2)$,
we can solve
$
  \mat{z}_\opt = \argmin_{\mat{N}\mat{z} = \mat{0}}
  \norm{\mat{A}\mat{z} - \mat{b}}_2^2
  + \lambda \norm{\mat{M}^{+}\mat{z}}_2^2,
$
where $\mat{N} = \mat{I}_m - \mat{M} \mat{M}^+$,
and return vector $\mat{M}^+ \mat{z}_\opt$ instead.
\end{restatable}


To solve this constrained regression problem,
we can add a scaled version of the constraint matrix $\mat{N}$
as a block to the approximate regression problem and take the projection of the resulting solution.

\begin{restatable}[Approximate equality-constrained regression]{lemma}{LemmaConstrainedReg}
\label{lemma:constraint-reg}
Let $\mat{M}\in\mathbb{R}^{n\times d}$, $\mat{N}\in\mathbb{R}^{m \times d}$, $\mat{b} \in \mathbb{R}^{n}$, and $0<\varepsilon<1/3$. To solve 
$
    \min_{\mat{N}\mat{x}=\mat{0}} \norm{\mat{M}\mat{x}-\mat{b}}_2^2
$
to a $(1+\varepsilon)$-approximation,
it suffices to solve
\[
\min_{\mat{x} \in \R^{d}} \left\Vert \begin{bmatrix}
  \mat{M} \\ \sqrt{w} \mat{N}
\end{bmatrix} \mat{x} - \begin{bmatrix}
  \mat{b} \\ \mat{0}
\end{bmatrix}\right\Vert_2^2
\]
to a $(1+\varepsilon/3)$-approximation
with $w \ge (1 + 12/ \varepsilon) \norm{\mat{M}\mat{N}^+}_2^2$.
\end{restatable}

Letting $\mat{z} = \mat{G}_{(n)}^\intercal \mat{y}$
in Line~8 of \texttt{TuckerALS}
and
modifying \FastKroneckerRegression to
support additional low-rank updates to the preconditioner,
we get the \FastFactorMatrixUpdate algorithm,
presented as~\Cref{alg:fast_factor_matrix_update} in
\Cref{subsec:fast_factor_matrix_update}.
The analysis is similar to the proofs of \Cref{thm:fast_kronecker_regression}.
The factor matrix updates benefit in the same way as before from
fast Kronecker matrix-vector products, and new low-rank block updates
are supported via the Woodbury identity.
We defer the proofs of the next two results to~\Cref{app:algorithm}.

\begin{restatable}{theorem}{TheoremFastFactorMatrixUpdate}
\label{theorem:fast_factor_matrix_update}
For any $\lambda \ge 0$, $\varepsilon \in (0,1/3)$, and $\delta > 0$,
the \FastFactorMatrixUpdate algorithm updates 
$\mat{A}_{(k)} \in \R^{I_k \times R_k}$
in $\textnormal{\texttt{TuckerALS}}$
with a $(1+\varepsilon)$-approximation,
with probability at least $1 - \delta$,
in time
\[
\textstyle
  \tilde{O}\parens*{
    I_k R_{\ne k}^2 \varepsilon^{-1} \log(1 / \delta)
    + I_k R\sum_{n=1}^N R_n + R_{k}^\omega \varepsilon^{-2}}.
\]
\end{restatable}

\begin{corollary}
\FastFactorMatrixUpdate updates $\mat{A}^{(k)} \in \R^{I_k \times R_k}$
in
$\tilde{O}\parens{
    I_k R_{\ne k}^{2-\theta^*} \varepsilon^{-1} \log(1 / \delta)
    + I_k R\sum_{n=1}^N R_n + R_{k}^\omega \varepsilon^{-2}}
$ time,
where $\theta^* > 0$ is the optimally balanced $\MM$ exponent in~\Cref{thm:kron_mat_mul_sqrt_decomp}.
\end{corollary}

\section{Experiments}
\label{sec:experiments}

\begin{figure}[t]
\centering
\subfigure{\includegraphics[width=1.0\linewidth]{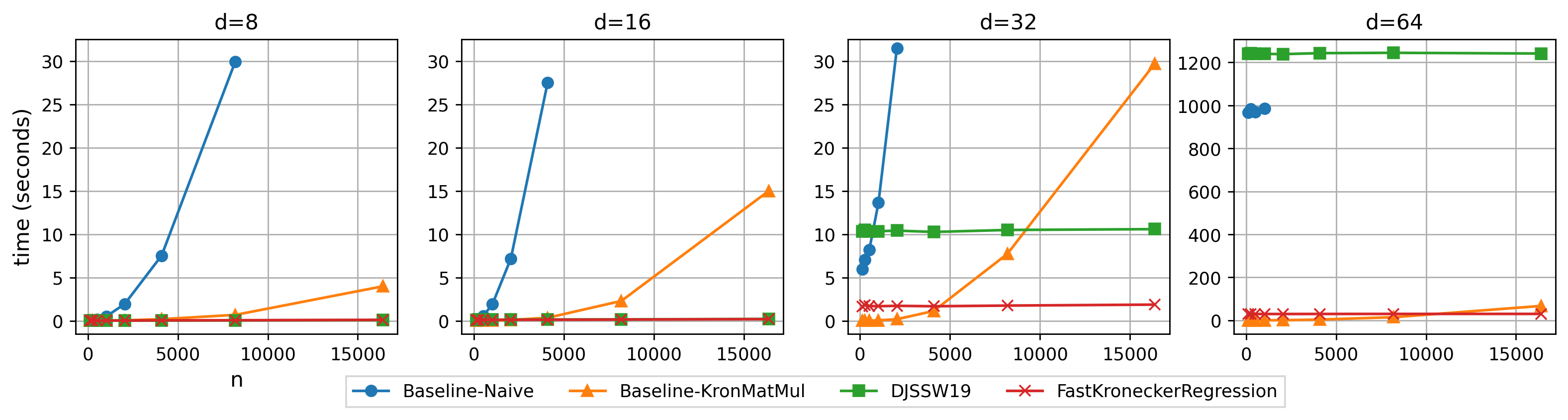}}
\vspace{-0.60cm}
\caption{
Running times of
Kronecker regression algorithms
with a design matrix of size $n^2 \times d^2$.
}
\label{fig:kronecker_regression_plot}
\vspace{-0.25cm}
\end{figure}

All experiments were run using NumPy~\cite{harris2020array} with an Intel Xeon W-2135 processor (8.25MB cache, 3.70 GHz) and 128GB of RAM.
The \FastKroneckerRegression-based ALS experiments
for low-rank Tucker decomposition on image tensors are deferred to~\Cref{app:tensor_decomposition_experiments}.
All of our code is available at \url{https://github.com/fahrbach/subquadratic-kronecker-regression}.

\paragraph{Kronecker regression.}
We build on the numerical experiments in~\cite{diao2018sketching, diao2019optimal}
for Kronecker regression that use two random factor matrices.
We generate matrices $\mat{A}^{(1)}, \mat{A}^{(2)} \in \R^{n \times d}$
where each entry is drawn i.i.d.\ from the normal distribution $\mathcal{N}(1, 0.001)$
and compare several algorithms for solving
$
  \min_{\mat{x}}\norm{(\mat{A}^{(1)} \ktimes \mat{A}^{(2)})\mat{x} - \mat{1}_{n^2} }_{2}^2 + \lambda \norm{\mat{x}}_{2}^2
$
as we increase $n,d$.
The running times are plotted in~\Cref{fig:kronecker_regression_plot}.

The algorithms we compare are:
(1) a baseline that solves the normal equation
$(\mat{K}^\intercal \mat{K} + \lambda\mat{I})^+ \mat{K}^\intercal \mat{b}$
and fully exploits the Kronecker structure of $\mat{K}^\intercal \mat{K}$
before calling $\texttt{np.linalg.pinv()}$;
(2) an enhanced baseline that combines the SVDs of $\mat{A}^{(n)}$
with \Cref{lemma:kron_mat_mul}, e.g.,
$\texttt{KronMatMul}([(\mat{U}^{(1)})^\intercal, (\mat{U}^{(2)})^\intercal], \mat{b})$,
using only Kronecker-vector products;
(3) the sketching algorithm of~\citet[Algorithm 1]{diao2019optimal};
and (4) our \FastKroneckerRegression algorithm in \Cref{alg:fast_kronecker_regression}.
For both sketching algorithms, we use $\varepsilon = 0.1$ and $\delta = 0.01$.
We reduce the number of row samples in both algorithms by $\alpha = 10^{-5}$
so that the algorithms are more practical and comparable to
the earlier experiments in~\cite{diao2018sketching, diao2019optimal}.
Lastly, we set $\lambda = 10^{-3}$.
We discuss additional parameter choice details and the full
results in~\Cref{app:kronecker_regression_experiments}.

The running times in~\Cref{fig:kronecker_regression_plot} demonstrate
several different behaviors.
The naive baseline quickly becomes impractical for moderately large values of $n$ or $d$.
\texttt{KronMatMul} is competitive for $n \le 10^4$, especially since it is an
exact method.
The runtimes of the sketching algorithms are nearly-independent of $n$.
\citet{diao2019optimal} works well for small $d$, but deteriorates tremendously as $d$ grows
because it computes
$((\mat{S}\mat{K})^\intercal \mat{S}\mat{K} + \lambda\mat{I})^+ \in \R^{d^2 \times d^2}$
and cannot exploit the Kronecker structure of $\mat{K}$, which takes $O(d^6)$ time.
\FastKroneckerRegression, on the other hand, runs in $O(d^4)$ time because it
uses quadratic-time Kronecker-vector products in each Richardson iteration step (\Cref{line:richardson_step}).

\vspace{-0.10cm}
\begin{table}[H]
    \caption{
    Kronecker regression losses for $d=64$.
    OPT denotes the loss of the \KronMatMul algorithm,
    \texttt{DJSSW19} is~\citet[Algorithm 1]{diao2019optimal},
    and \Cref{alg:fast_kronecker_regression} is~\FastKroneckerRegression.
    We also record the relative error of each algorithm and the number
    of rows sampled from $\mat{A}^{(1)} \ktimes \mat{A}^{(2)}$.
    }
    \label{tab:sketchinig_approximation}
    \centering
    \small
    \begin{tabular}{cccccccccc}
        \toprule
        $n$ & OPT & \Cref{alg:fast_kronecker_regression} & Approx & \texttt{DJSSW19} & Approx & Rows sampled ($\%$) \\
        \midrule
1024 & 0.031 & 0.032 & 1.051 & 0.035 & 1.138 & 0.0370 \\
2048 & 0.123 & 0.126 & 1.026 & 1.577 & 12.792 & 0.0093 \\
4096 & 0.507 & 0.520 & 1.026 & 275.566 & 543.776 & 0.0023 \\
8192 & 2.073 & 2.136 & 1.030 & 333.430 & 160.809 & 0.0006 \\
16384 & 8.238 & 8.608 & 1.045 & 546391.728 & 66329.791 & 0.0001 \\
        \bottomrule
    \end{tabular}
\end{table}
\vspace{-0.40cm}

These experiments also show that combining sketching with iterative methods can
give better \emph{sketch efficiency}.
\Cref{tab:sketchinig_approximation} compares the loss of
\cite[Algorithm 1]{diao2019optimal} and \FastKroneckerRegression to an exact baseline
OPT for $d=64$.
Both algorithms use the exact same sketch $\mat{S}\mat{K}$ for each value of $n$.
Our algorithm uses the original
$(\mat{K}^\intercal \mat{K} + \lambda \mat{I})^+$ as a preconditioner
to solve the sketched problem,
whereas \citet[Algorithm 1]{diao2019optimal} computes
$((\mat{S}\mat{K})^\intercal \mat{S}\mat{K} + \lambda\mat{I})^+ (\mat{S}\mat{K})^\intercal \mat{S}\mat{b}$ exactly and becomes numerically unstable
for $n \ge 2048$ when $d \in \{16, 32, 64\}$.
This raises the question about how to combine sketched information with
the original data to achieve more efficient algorithms,
even when solving sketched instances.
We leave this question of sketch efficiency as an interesting future work.

\bibliographystyle{plainnat}
\bibliography{main}

\newpage
\appendix

\section{Missing Analysis from \Cref{sec:row_sampling}}
\label{app:row_sampling}

Here we show how to use leverage scores of
the design matrix $\mat{A} \in \R^{n \times d}$
to create a smaller least squares problem
whose solution vector gives a $(1+\varepsilon)$-approximation to the
original regression problem.
Our proof relies on several sketching
and leverage score sampling results in randomized numerical linear
algebra~\cite{drineas2006fast,drineas2011faster,woodruff2014sketching}.
These prerequisite results are well-known, but scattered through the
literature.
They are the building blocks for proving our approximate
block-regression results in
\Cref{lemma:approximate_block_regression}
and
\Cref{cor:approximate_ridge_regression}.

\subsection{Approximate Least Squares}
\label{subsec:fast_least_squares}
We follow the outline of~\citet{larsen2022sketching}
(originally written in~\cite[Appendix B]{larsen2020practical}).
Consider the overdetermined least squares problem
defined by a matrix $\mat{A} \in \R^{n \times d}$ and
response vector $\mat{b} \in \R^{n}$,
where $n \ge d$ and $\rank(\mat{A}) = d$.
Define the optimal sum of squared residuals to be
\begin{equation}
\label{eqn:least_squares_problem}
  \cR^2
  =
  \min_{\mat{x} \in \R^{d}}\norm*{\mat{A}\mat{x} - \mat{b}}_{2}^2.
\end{equation}
Assume for now $\mat{A}$ is full rank.
Let the compact SVD of the design matrix be
$\mat{A} = \mat{U}_{\mat{A}} \mat{\Sigma}_{\mat{A}} \mat{V}_{\mat{A}}^\intercal$.
By definition, $\mat{U}_{\mat{A}} \in \R^{n \times d}$ is an orthonormal
basis for the column space of $\mat{A}$.
Let $\mat{U}_{\mat{A}}^\perp \in \R^{n \times (n-d)}$ be an orthonormal
basis for the $(n-d)$-dimensional subspace that is
orthogonal to the column space of~$\mat{A}$.
For notational simplicity,
let $\mat{b}^\perp = \mat{U}_{\mat{A}}^\perp {\mat{U}_{\mat{A}}^\perp}^\intercal \mat{b}$
denote the projection of $\mat{b}$ onto the orthogonal subspace $\mat{U}_{\mat{A}}^\perp$.
The vector $\mat{b}^\intercal$
is important because its norm is equal to the norm of the residual vector.
To see this, observe that $\mat{x}$ can be chosen so that $\mat{A}\mat{x}$
perfectly matches the part of $\mat{b}$ in the column space of $\mat{A}$, but
cannot (by definition) match anything in the range of $\mat{U}_{\mat{A}}^\perp$:
\begin{equation}
  \cR^2
  =
  \min_{\mat{x} \in \R^{d}}\norm*{\mat{A}\mat{x} - \mat{b}}_{2}^2
  =
  \norm*{\mat{U}_{\mat{A}}^\perp {\mat{U}_{\mat{A}}^\perp}^\intercal \mat{b}}_{2}^{2}
  =
  \norm*{\mat{b}^\perp}_{2}^{2}.
\end{equation}
We denote the solution to the least squares problem by $\mat{x}_{\opt}$,
hence we have $\mat{b} = \mat{A}\mat{x}_{\opt} + \mat{b}^\perp$.

Now we build on a structural result of \citet{drineas2011faster}
that establishes sufficient conditions on any sketching
matrix $\mat{S} \in \R^{s \times n}$
such that the solution $\mat{\tilde x}_{\opt}$ to the approximate
least squares problem
\begin{align}
\label{eqn:approximate_least_squares_problem}
  \mat{\tilde x}_{\opt}
  =
  \argmin_{\mat{x} \in \R^d}\norm*{\mat{S} \parens*{\mat{A}\mat{x} - \mat{b}}}_{2}^2
\end{align}
gives a relative-error approximation to the original least squares problem.
The two conditions we require of matrix $\mat{S}$ are:
\begin{align}
\label{eqn:structural_condition_1}
  &\sigma^2_{\min}\parens*{\mat{S} \mat{U}_{\mat{A}}} \ge 1 / \sqrt{2}, \text{~and~} \\
\label{eqn:structural_condition_2}
  &\norm*{\mat{U}_{\mat{A}}^\intercal \mat{S}^\intercal \mat{S} \mat{b}^\perp }_{2}^2 \le \varepsilon \cR^2 / 2,
\end{align}
for some $\varepsilon \in (0,1)$.
While the algorithms we consider in this work are randomized,
the following lemma is a deterministic statement.
Failure probabilities enter our analysis later when we show 
our sketch matrices satisfy
conditions~\Cref{eqn:structural_condition_1} and
\Cref{eqn:structural_condition_2} with sufficiently high probability.

\begin{lemma}[{\citet[Lemma 1]{drineas2011faster}}]
\label{lem:approx_least_squares}
Consider the overconstrained least squares approximation
problem in~\Cref{eqn:least_squares_problem}, and
let the matrix $\mat{U}_{\mat{A}} \in \R^{n \times d}$ contain the top $d$
left singular vectors of $\mat{A}$.
Assume the matrix $\mat{S}$ satisfies conditions
\Cref{eqn:structural_condition_1} and
\Cref{eqn:structural_condition_2} for some $\varepsilon \in (0, 1)$.
Then, the solution $\mat{\tilde x}_{\opt}$ to the approximate least squares
problem~\Cref{eqn:approximate_least_squares_problem} satisfies:
\begin{align}
  \norm*{\mat{A} \mat{\tilde x}_{\opt} - \mat{b} }_{2}^2
  &\le
  \parens*{1 + \varepsilon} \cR^2, \text{~and~}
  \\
  \norm*{\mat{\tilde x}_{\opt} - \mat{x}_{\opt}}_{2}^2
  &\le
  \frac{1}{\sigma^2_{\min}\parens*{\mat{A}}}
  \varepsilon \cR^2.
\end{align}
\end{lemma}

\begin{proof}
Let us first rewrite the sketched least squares problem induced by $\mat{S}$ as
\begin{align}
    \label{eqn:sketched_ls_proof_1}
    \min_{\mat{x} \in \R^{d}} \norm*{\mat{S} \mat{A}\mat{x} - \mat{S}\mat{b}}_{2}^2
    &=
    \min_{\mat{y} \in \R^{d}} \norm*{\mat{S} \mat{A}\parens*{\mat{x}_{\opt} + \mat{y}} - \mat{S}\parens*{\mat{A}\mat{x}_{\opt} + \mat{b}^{\perp}}}_{2}^2 \\
    &=
    \min_{\mat{y} \in \R^{d}} \norm*{\mat{S} \mat{A}\mat{y} - \mat{S}\mat{b}^{\perp}}_{2}^2 \notag \\
    \label{eqn:sketched_ls_proof_2}
    &=
    \min_{\mat{z} \in \R^{d}} \norm*{\mat{S} \mat{U}_{\mat{A}}\mat{z} - \mat{S}\mat{b}^{\perp}}_{2}^2.
\end{align}
Equation~\Cref{eqn:sketched_ls_proof_1} is true
because $\mat{b} = \mat{A}\mat{x}_{\opt} + \mat{b}^\perp$,
and \Cref{eqn:sketched_ls_proof_2} follows because the columns of $\mat{A}$ span
the same subspace as the columns of $\mat{U}_{\mat{A}}$.
Now, let $\mat{z}_{\opt} \in \R^{d}$ be such that
$\mat{U}_{\mat{A}} \mat{z}_{\opt} = \mat{A}(\mat{\tilde x}_{\opt} - \mat{x}_{\opt})$
and note that $\mat{z}_{\opt}$ minimizes \Cref{eqn:sketched_ls_proof_2}.
This fact follows from
\begin{align*}
    \norm*{\mat{S}\mat{A}(\mat{\tilde x}_{\opt} - \mat{x}_{\opt}) - \mat{S}\mat{b}^\perp}_{2}^2
    =
    \norm*{
    \mat{S}\mat{A}\mat{\tilde x}_{\opt}
    - \mat{S}\parens*{\mat{b} - \mat{b}^\perp} - \mat{S}\mat{b}^\perp }_{2}^2
    =
    \norm*{\mat{S}\mat{A}\mat{\tilde x}_{\opt} - \mat{S}\mat{b}}_{2}^2.
\end{align*}
Thus, by the normal equations, we have
\[
    \parens*{\mat{S}\mat{U}_{\mat{A}}}^\intercal \mat{S}\mat{U}_{\mat{A}} \mat{z}_{\opt}
    =
    \parens*{\mat{S}\mat{U}_{\mat{A}}}^\intercal \mat{S} \mat{b}^\perp.
\]
Taking the norm of both sides and observing that under condition
\Cref{eqn:structural_condition_1}
we have
$\sigma_{i}((\mat{S}\mat{U}_{\mat{A}})^\intercal \mat{S}\mat{U}_{\mat{A}})
=
\sigma_{i}^2 (\mat{S}\mat{U}_{\mat{A}}) \ge 1/\sqrt{2}$, for all $i \in [d]$,
it follows that
\begin{equation}
  \label{eqn:sketched_ls_proof_3}
    \norm*{\mat{z}_{\opt}}_{2}^2 / 2
    \le
    \norm*{\parens*{\mat{S}\mat{U}_{\mat{A}}}^\intercal \mat{S}\mat{U}_{\mat{A}} \mat{z}_{\opt}}_{2}^2
    =
    \norm*{\parens*{\mat{S}\mat{U}_{\mat{A}}}^\intercal \mat{S}\mat{b}^\perp }_{2}^2.
\end{equation}
Using condition \Cref{eqn:structural_condition_2}, we observe that
\begin{equation}
  \label{eqn:sketched_ls_proof_4}
    \norm*{\mat{z}_{\opt}}_{2}^2
    \le
    2 \norm*{\parens*{\mat{S}\mat{U}_\mat{A}}^\intercal \mat{S}\mat{b}^\perp}_{2}^2
    \le
    \varepsilon \cR^2.
\end{equation}

To establish the first claim of the lemma, let us rewrite the squared
norm of the residual vector as
\begin{align}
    \norm*{\mat{A} \mat{\tilde x}_{\opt} - \mat{b}}_{2}^2
    &=
    \norm*{\mat{A} \mat{\tilde x}_{\opt} - \mat{A}\mat{x}_{\opt} + \mat{A}\mat{x}_{\opt} - \mat{b}}_{2}^2 \notag \\
    &=
    \label{eqn:sketched_ls_proof_5}
    \norm*{\mat{A} \mat{\tilde x}_{\opt} - \mat{A}\mat{x}_{\opt}}_{2}^2
    + \norm*{\mat{A}\mat{x}_{\opt} - \mat{b}}_{2}^2 \\
    &=
    \label{eqn:sketched_ls_proof_6}
    \norm*{\mat{U}_{\mat{A}} \mat{z}_{\opt}}_{2}^2
    + \cR^2 \\
    &\le
    \label{eqn:sketched_ls_proof_7}
    \parens*{1 + \varepsilon} \cR^2,
\end{align}
where~\Cref{eqn:sketched_ls_proof_5} follows from the Pythagorean theorem
since $\mat{b}-\mat{A}\mat{x}_{\opt} = \mat{b}^\perp$, which is orthogonal
to $\mat{A}$, and consequently $\mat{A}(\mat{x}_{\opt} - \mat{\tilde x}_{\opt})$;
\Cref{eqn:sketched_ls_proof_6} follows from the definition of $\mat{z}_{\opt}$
and $\cR^2$;
and \Cref{eqn:sketched_ls_proof_7} follows from
\Cref{eqn:sketched_ls_proof_4} and the orthogonality of $\mat{U}_{\mat{A}}$.

To establish the second claim of the lemma, recall that
$\mat{A}(\mat{x}_{\opt} - \mat{\tilde x}_{\opt}) = \mat{U}_{\mat{A}} \mat{z}_{\opt}$.
Taking the norm of both sides of this expression, we have
\begin{align}
    \label{eqn:sketched_ls_proof_8}
    \norm*{\mat{x}_{\opt} - \mat{\tilde x}_{\opt}}_{2}^2
    &\le 
    \frac{\norm*{\mat{U}_{\mat{A}} \mat{z}_{\opt}}_{2}^2}{\sigma_{\min}^2\parens*{\mat{A}}} \\
    \label{eqn:sketched_ls_proof_9}
    &\le
    \frac{\varepsilon \cR^2 }{\sigma_{\min}^2\parens*{\mat{A}}},
\end{align}
where \Cref{eqn:sketched_ls_proof_8} follows since $\sigma_{\min}(\mat{A})$
is the smallest singular
value of $\mat{A}$ and $\rank(\mat{A})=d$;
and \Cref{eqn:sketched_ls_proof_9} follows from
\Cref{eqn:sketched_ls_proof_4} and the orthogonality of $\mat{U}_{\mat{A}}$.
\end{proof}

Next we present two results that are useful for proving
our sketches $\mat{S}$ satisfy the structural conditions in
Equations~\Cref{eqn:structural_condition_1} and \Cref{eqn:structural_condition_2}.
The first result states $\mat{S}\mat{U}_{\mat{A}}$ is a subspace embedding
for the column space of $\mat{U}_{\mat{A}}$.
This result can be thought of as an approximate isometry
and is noticeably stronger than the desired condition
$\sigma_{\min}^2(\mat{S}\mat{U}_{\mat{A}}) \ge 1/\sqrt{2}$.

\begin{theorem}[{\citet[Theorem 17]{woodruff2014sketching}}]
\label{thm:structural_tool_1}
Consider $\mat{A} \in \R^{n \times d}$ and its compact SVD 
$\mat{A} = \mat{U}_{\mat{A}} \mat{\Sigma}_{\mat{A}} \mat{V}_{\mat{A}}^\intercal$.
Let $\mat{p} \in [0,1]^n$
be a $\beta$-overestimate for the
leverage score distribution of $\mat{A}$.
Let $s > 144d \ln(2d / \delta)/(\beta \varepsilon^2)$.
Let the matrix $\mat{S} \in \R^{s \times n}$
be the output of $\RowSampling(\mat{A}, s, \mat{p})$ (\Cref{def:sample_rows_alg}).
Then, with probability at least $1-\delta$, simultaneously for all $i$,
we have
\[
  1 - \varepsilon
  \le
  \sigma_{i}^2 \parens*{\mat{S} \mat{U}_{\mat{A}}}
  \le
  1 + \varepsilon.
\]
\end{theorem}

For the second structural condition, we use the following result about
squared-distance sampling for approximate matrix multiplication
in~\cite{drineas2006fast}.
In our analysis of block leverage score sampling
(e.g., ridge regression),
it is possible (and beneficial) that $\beta > 1$ and that rows are sometimes
not sampled.
We modify the original theorem statement and provide a proof to show
that the result is unaffected.

\begin{theorem}[{\citet[Lemma 8]{drineas2006fast}}]
\label{thm:structural_tool_2}
Let $\mat{A} \in \R^{n \times m}$, $\mat{B} \in \R^{n \times p}$,
and $s$ denote the number of samples.
Let the vector $\mat{p} \in [0,1]^{n}$ contain probabilities such that,
for all $i \in [n]$, we have
\[
  p_{i} \ge \beta \frac{\norm{\mat{a}_{i:}}_{2}^2}{ \norm*{\mat{A}}_{\frobenius}^2 },
\]
for some constant $\beta > 0$.
We require that $\norm{\mat{p}}_1 \le 1$, but it is possible that
$\mat{p}$ does not contain all of the probability mass (i.e., $\norm{\mat{p}}_1 < 1$).
Sample $s$ row indices $(\xi^{(1)},\xi^{(2)},\dots,\xi^{(s)})$ from $\mat{p}$,
independently and with replacement, and form the approximate product
\[
  \frac{1}{s} \sum_{t=1}^s \frac{1}{p_{\xi^{(t)}}}
    \mat{a}_{\xi^{(t)}:}^\intercal \mat{b}_{\xi^{(t)}:}
  =
  \parens*{\mat{S}\mat{A}}^\intercal \mat{S}\mat{B},
\]
where $\mat{S} \in \R^{s \times n}$
is the sampling and rescaling matrix whose $t$-th row is defined by
the entries
\[
  s_{tk} =
  \begin{cases}
    \frac{1}{\sqrt{s p_{k}}} & \text{if $k = \xi_{t}$}, \\
    0 & \text{otherwise}.
  \end{cases}
\]
Disregard trials that occur with the remaining probability $1 - \norm{\mat{p}}_1$.
Then, we have
\[
  \E\bracks*{\norm*{\mat{A}^\intercal \mat{B}
  - \parens*{\mat{S}\mat{A}}^\intercal\mat{S}\mat{B}}_{\frobenius}^2}
  \le
  \frac{1}{\beta s}
  \norm*{\mat{A}}_{\frobenius}^2
  \norm*{\mat{B}}_{\frobenius}^2.
\]
\end{theorem}

\begin{proof}
First we analyze the entry of 
$(\mat{S}\mat{A})^\intercal \mat{S}\mat{B}$ at index $(i,j)$.
Viewing the approximate product as a sum of outer products,
we can write this entry in terms of scalar random variables
$X_t$, for $t \in [s]$, as follows:
\[
  X_t =
  \begin{cases}
    \frac{a_{\xi^{(t)}i} b_{\xi^{(t)}j}}{s p_{\xi^{(t)}}}
    & \text{with probability $p_i$ for each $i \in [n]$}, \\
    0 & \text{otherwise with probability $1 - \norm{p}_1$}
  \end{cases}
  ~\implies~~
  \bracks*{(\mat{S}\mat{A})^\intercal \mat{S}\mat{B}}_{ij}
  = \sum_{t=1}^s X_t.
\]
The expected values of $X_t$ and $X_t^2$ for all values of $t$ are
\begin{align*}
  \E\bracks*{X_t} &= \sum_{k=1}^n p_k \frac{a_{ki} b_{kj}}{s p_k}
                   = \frac{1}{s} \parens*{\mat{A}^\intercal \mat{B}}_{ij}, \text{~and~} \\
  \E\bracks*{X_t^2} &=
        \sum_{k=1}^n p_k \parens*{\frac{a_{ki} b_{kj}}{s p_k}}^2
        =
        \frac{1}{s^2}\sum_{k=1}^n \frac{\parens*{a_{ki} b_{kj}}^2}{p_k}.
\end{align*}
Therefore, 
$\E[\bracks{(\mat{S}\mat{A})^\intercal \mat{S}\mat{B}}_{ij}]
  = \sum_{t=1}^s \E[X_t] = \parens{\mat{A}^\intercal \mat{B}}_{ij}$,
which means the estimator is unbiased.
Furthermore, since the estimated matrix entry is the sum of $s$
i.i.d.\ random variables,
its variance is
\begin{align*}
  \var\parens*{\bracks*{(\mat{S}\mat{A})^\intercal \mat{S}\mat{B}}_{ij}}
  &=
  \sum_{t=1}^s \var\parens*{X_t} \\
  &=
  \sum_{t=1}^s \parens*{\E\bracks*{X_t^2} - \E\bracks*{X_t}^2} \\
  &=
  \sum_{t=1}^s \frac{1}{s^2} \sum_{k=1}^n \parens*{\frac{\parens*{a_{ki} b_{kj}}^2}{p_k} 
        - \parens*{ \mat{A}^\intercal\mat{B} }_{ij}^2 } \\
  &=
  \frac{1}{s} \sum_{k=1}^n \parens*{\frac{\parens*{a_{ki} b_{kj}}^2}{p_k} 
        - \parens*{ \mat{A}^\intercal\mat{B} }_{ij}^2 }.
\end{align*}
Now we apply this result to the expectation we want to bound:
\begin{align*}
  \E\bracks*{\norm*{\mat{A}^\intercal \mat{B}
  - \parens*{\mat{S}\mat{A}}^\intercal\mat{S}\mat{B}}_{\frobenius}^2}
  &=
  \sum_{i=1}^m \sum_{j=1}^p
  \E\bracks*{
    \parens*{ \bracks*{\parens*{\mat{S}\mat{A}}^\intercal\mat{S}\mat{B}}_{ij}
    - \parens*{\mat{A}^\intercal \mat{B}}_{ij}}^2} \\
  &=
  \sum_{i=1}^m \sum_{j=1}^p
  \E\bracks*{
    \parens*{ \bracks*{\parens*{\mat{S}\mat{A}}^\intercal\mat{S}\mat{B}}_{ij}
    - \E\bracks*{ \bracks*{\parens*{\mat{S}\mat{A}}^\intercal\mat{S}\mat{B}}_{ij} }}^2} \\
  &=
  \sum_{i=1}^m \sum_{j=1}^p
  \var\parens*{
    \bracks*{\parens*{\mat{S}\mat{A}}^\intercal\mat{S}\mat{B}}_{ij}} \\
  &=
  \frac{1}{s} \sum_{i=1}^m \sum_{j=1}^p
  \sum_{k=1}^n \parens*{\frac{\parens*{a_{ki} b_{kj}}^2}{p_k} 
        - \parens*{ \mat{A}^\intercal\mat{B} }_{ij}^2 } \\
  &=
  \frac{1}{s} 
  \sum_{k=1}^n \frac{\parens{\sum_{i=1}^m a_{ki}^2 }\parens{\sum_{j=1}^p b_{kj}^2 }}{p_k} 
        - \frac{n}{s}\sum_{i=1}^m\sum_{j=1}^p \parens*{ \mat{A}^\intercal\mat{B} }_{ij}^2 \\
  &=
  \frac{1}{s} 
  \sum_{k=1}^n
  \frac{\norm{\mat{a}_{k:}}_{2}^2 \norm{\mat{b}_{k:}}_{2}^2}{p_k}
  - \frac{n}{s} \norm*{\mat{A}^\intercal \mat{B}}_{\frobenius}^2 \\
  &\le
  \frac{1}{s} 
  \sum_{k=1}^n
  \frac{\norm{\mat{a}_{k:}}_{2}^2 \norm{\mat{b}_{k:}}_{2}^2}{p_k}.
\end{align*}
The last inequality uses the fact that the Frobenius norm of any matrix
is nonnegative.
Finally, by using the $\beta$-overestimate assumption on the sampling
probabilities, we have
\begin{align*}
  \E\bracks*{\norm*{\mat{A}^\intercal \mat{B}
  - \parens*{\mat{S}\mat{A}}^\intercal\mat{S}\mat{B}}_{\frobenius}^2}
  &\le 
  \frac{1}{s} 
  \sum_{k=1}^n
  \frac{\norm{\mat{a}_{k:}}_{2}^2 \norm{\mat{b}_{k:}}_{2}^2}{p_k} \\
  &\le
  \frac{1}{s} 
  \sum_{k=1}^n
  \parens*{ \norm*{\mat{A}}_{\frobenius}^2  \frac{\norm{\mat{a}_{k:}}_{2}^2 \norm{\mat{b}_{k:}}_{2}^2}{\beta \norm{\mat{a}_{k:}}_{2}^2} } \\
  &=
  \frac{1}{s \beta}
  \norm*{\mat{A}}_{\frobenius}^2
  \sum_{k=1}^n \norm*{\mat{b}_{k:}}_{2}^2 \\
  &=
  \frac{1}{s \beta}
  \norm*{\mat{A}}_{\frobenius}^2
  \norm*{\mat{B}}_{\frobenius}^2,
\end{align*}
which is the desired upper bound.
\end{proof}

\subsection{Generalizing to Submatrix Sketching}

Now that our main tools are in place, we extend the analysis of
approximate least squares to work with sketched submatrices
of the vertically stacked block design matrix.

\BlockSketchIsSpectralApprox*

\begin{proof}
Write the compact SVD of $\mat{A_1}$ as
$\mat{A}_1 = \mat{U}_{\mat{A}_1} \mat{\Sigma}_{\mat{A}_1} \mat{V}_{\mat{A}_1}^\intercal$.
\Cref{thm:structural_tool_1} guarantees that with probability at least $1-\delta$,
\[
  1 - \varepsilon \le \sigma_{i}^2\parens*{\mat{S}\mat{U}_{\mat{A}_1}} \le 1 + \varepsilon.
\]
Therefore, we have
\[
  (1-\varepsilon) \mat{I}_{d}
  \preccurlyeq
  \parens*{
  \mat{S}
  \mat{U}_{\mat{A}_1}
  }^\intercal
  \mat{S}
  \mat{U}_{\mat{A}_1}
  \preccurlyeq
  (1+\varepsilon) \mat{I}_{d}.
\]
It follows that
\begin{align*}
  (1-\varepsilon)
  \mat{A}_1^\intercal \mat{A}_1
  &=
  (1-\varepsilon)
  \mat{V}_{\mat{A}_1}
  \mat{\Sigma}_{\mat{A}_1}^\intercal
  \mat{I}_d
  \mat{\Sigma}_{\mat{A}_1}
  \mat{V}_{\mat{A}_1}^\intercal \\
  &\preccurlyeq
  \mat{V}_{\mat{A}_1}
  \mat{\Sigma}_{\mat{A}_1}^\intercal
  \mat{U}_{\mat{A}_1}^\intercal
  \mat{S}^\intercal
  \mat{S}
  \mat{U}_{\mat{A}_1}
  \mat{\Sigma}_{\mat{A}_1}
  \mat{V}_{\mat{A}_1}^\intercal \\
  &=
  \parens*{\mat{S}\mat{A}_1}^\intercal
  \mat{S}\mat{A}_1.
\end{align*}
Similarly, we have
$\parens*{\mat{S}\mat{A}_1}^\intercal
  \mat{S}\mat{A}_1
  \preccurlyeq
  (1+\varepsilon)
  \mat{A}_1^\intercal \mat{A}_1.
$
Writing $\mat{A}^\intercal \mat{A} = \mat{A}_{1}^\intercal \mat{A}_{1} +
\mat{A}_{2}^\intercal \mat{A}_{2}$ as the sum of outer products,
we have
\begin{align*}
  \parens*{1-\varepsilon} \mat{A}^\intercal \mat{A}
  &\preccurlyeq
  \parens*{1-\varepsilon} \mat{A}_{1}^\intercal \mat{A}_{1} +
  \mat{A}_{2}^\intercal \mat{A}_{2} \\
  &\preccurlyeq
  \parens*{\mat{S}\mat{A}_1}^\intercal
  \mat{S}\mat{A}_1
  +
  \mat{A}_{2}^\intercal \mat{A}_{2} \\
  &\preccurlyeq
  \parens*{1+\varepsilon} \mat{A}_{1}^\intercal \mat{A}_{1} +
  \mat{A}_{2}^\intercal \mat{A}_{2} \\
  &\preccurlyeq
  \parens*{1+\varepsilon} \mat{A}^\intercal \mat{A},
\end{align*}
which completes the proof.
\end{proof}

\ApproximateBlockRegression*
\begin{proof}
Let $\delta = 1/10$ be the desired failure probability.
Consider the augmented sketch matrix
\begin{align}
  \mat{S}' = \begin{bmatrix}
    \mat{S} & \mat{0}_{s \times n_2} \\
    \mat{0}_{n_2 \times n_1} & \mat{I}_{n_2}
  \end{bmatrix}.
\end{align}
It follows that
\begin{align}
\label{eqn:unstacked_sketch}
  \mat{S}'\mat{A} = \begin{bmatrix}
    \mat{S}\mat{A}_1 \\
    \mat{A}_2
  \end{bmatrix}.
\end{align}
Let the compact SVD of $\mat{A}$ be
$\mat{A} = \mat{U} \mat{\Sigma} \mat{V}^\intercal$.
We prove that each of the structural conditions
about $\mat{S}'$ in \Cref{eqn:structural_condition_1} and
\Cref{eqn:structural_condition_2} fail with probability at most
$\delta/2$.
Then we use a union bound and apply \Cref{lem:approx_least_squares}.

\textbf{Satisfying structural condition 1.}
It follows from \Cref{eqn:unstacked_sketch}
that
\[
  \parens*{\mat{S}'\mat{A}}^\intercal \mat{S}'\mat{A}
  =
  \parens*{\mat{S}\mat{A}_1}^\intercal \mat{S}\mat{A}_{1}
  +
  \mat{A}_2^\intercal \mat{A}_2.
\]
Using \Cref{lemma:block_sketch_is_spectral_approx}, we know
\begin{align}
\label{eqn:augmented_spectral_approx}
  \parens*{1-\varepsilon} \mat{A}^\intercal \mat{A}
  &\preccurlyeq
  \parens*{\mat{S}'\mat{A}}^\intercal \mat{S}'\mat{A}
  \preccurlyeq
  \parens*{1+\varepsilon} \mat{A}^\intercal \mat{A}.
\end{align}
Since
$\mat{A}^\intercal \mat{A} = \mat{V} \mat{\Sigma}^\intercal
     \mat{I}_{d} \mat{\Sigma} \mat{V}^\intercal
$
and
$
    \parens*{\mat{S}'\mat{A}^\intercal} \mat{S}'\mat{A}
=
  \mat{V} \mat{\Sigma}^\intercal \mat{U}^\intercal \mat{S}'^\intercal \mat{S}' \mat{U} \mat{\Sigma} \mat{V}^\intercal$,
it follows from \Cref{eqn:augmented_spectral_approx} that
\[
  (1-\varepsilon) \mat{I}_d
  \preccurlyeq
  \parens*{\mat{S}' \mat{U}}^\intercal
  \mat{S}' \mat{U}
  \preccurlyeq
  (1+\varepsilon) \mat{I}_d
\]
since $\mat{\Sigma}$ and $\mat{V}^\intercal$ are positive definite.
Therefore, the first structural condition \Cref{eqn:structural_condition_1}
is true with probability at least $1 - \delta/2$ as long as
$1 - \varepsilon \ge 1/\sqrt{2}$.
This means the number of samples needs to be at least
\[
  s > \frac{144d \ln(4d/\delta)}{\beta (1 - 1/\sqrt{2})^2}
  >
  \frac{1680 d \ln(4d / \delta)}{\beta}.
\]

\textbf{Satisfying structural condition 2.}
We show \Cref{eqn:structural_condition_2} holds with probability at least
$1-\delta/2$ using a modification of
\Cref{thm:structural_tool_2} and Markov's inequality.
First observe that
\[
  \mat{U}^\intercal \mat{b}^\perp
  =
  \mat{U}^\intercal \parens*{\mat{U}^\perp {\mat{U}^\perp}^\intercal \mat{b}}
  =
  \mat{0}_{\rank(\mat{A})},
\]
where $\mat{b}^\perp$ is defined as in \Cref{subsec:fast_least_squares}.
Thus, the second structural condition can be seen as bounding how
closely this sampled product approximates the zero vector.
It follows that
\begin{align*}
  \norm*{ \mat{U}^\intercal
      \mat{S}'^\intercal \mat{S}' \mat{b}^\perp}_{2}^2
  &=
  \norm*{ \mat{U}^\intercal \mat{b}^\perp - 
    \mat{U}^\intercal
      \mat{S}'^\intercal \mat{S}' \mat{b}^\perp}_{2}^2 \\
  &=
  \norm*{ \mat{U}^\intercal \parens*{\mat{I}_{n_1 + n_2} - 
      \mat{S}'^\intercal \mat{S}' } \mat{b}^\perp}_{2}^2 \\
  &=
  \norm*{ \mat{U}^\intercal 
  \begin{bmatrix}
    \mat{I}_{n_1} - \mat{S}^\intercal \mat{S} & \mat{0} \\
    \mat{0} & \mat{0}
  \end{bmatrix}
  \mat{b}^\perp}_{2}^2 \\
  &=
  \norm*{ \mat{\tilde U}^\intercal \parens*{\mat{I}_{n_1} - \mat{S}^\intercal \mat{S}} \mat{\tilde b}^\perp }_2^2 \\
  &=
  \norm*{ \mat{\tilde U}^\intercal \mat{\tilde b}^\perp
   - \mat{\tilde U}^\intercal \mat{S}^\intercal \mat{S} \mat{\tilde b}^\perp }_2^2,
\end{align*}
where $\mat{\tilde U} \in \R^{n_1 \times d}$
and $\mat{\tilde b}^\perp \in \R^{n_1}$
denote the first
$n_1$ rows of $\mat{U}$
and $\mat{b}^\perp$, respectively.

Now we bound the probability that a row index in $\mat{\tilde U}$
is sampled when constructing $\mat{S}$,
which allows us to apply \Cref{thm:structural_tool_2}:
\begin{align}
  \Pr\parens*{\text{row $i \in [n_1]$ is sampled}}
  &\ge
  \beta \frac{\ell_{i}(\mat{A}_1)}{\norm{\bm{\ell}(\mat{A}_1)}} \label{eqn:subsketch_probability} \\
  &=
  \beta \frac{\ell_{i}(\mat{A}_1)}{\rank(\mat{A}_1)}
    \cdot
    \frac{\norm{\bm{\ell}_{\mat{A}_1}(\mat{A})}_1}{ \ell_i(\mat{A})}
    \cdot
    \frac{\ell_i(\mat{A})}{\norm{\bm{\ell}_{\mat{A}_1}(\mat{A})}_1} \\
  &\ge
  \beta \frac{\norm{\bm{\ell}_{\mat{A}_1}(\mat{A})}_1}{\rank(\mat{A}_1)}
    \cdot
    \frac{\ell_i(\mat{A})}{\norm{\bm{\ell}_{\mat{A}_1}(\mat{A})}_1}  \label{eqn:add_rows_drop_ls}\\
  &=
  \beta \frac{\norm{\mat{\tilde U}}_\frobenius^2}{\rank(\mat{A}_1)}
    \cdot
    \frac{\norm{\mat{\tilde u}_{i:}}_2^2}{\norm{\mat{\tilde U}}_\frobenius^2}  \label{eqn:tilde_U_norm}.
\end{align}
We use
$\norm{\bm{\ell}_{\mat{A}_1}(\mat{A})}_1 = \sum_{i \in [n_1]} \ell_{i}(\mat{A})$
to denote the sum of leverage scores of $\mat{A}$ corresponding to the rows of
$\mat{A}_1$.
Equation \Cref{eqn:add_rows_drop_ls}
holds because leverage scores do not increase when rows are added to the matrix,
i.e., $\ell_i(\mat{A}_1) \ge \ell_i(\mat{A})$.
Equation \Cref{eqn:tilde_U_norm} is true because the leverage scores of
$\mat{A}$ corresponding to the rows in $\mat{A}_1$ are given by the submatrix
$\mat{\tilde U}$ in the compact SVD of $\mat{A}$.
Therefore, \Cref{thm:structural_tool_2} guarantees that
\begin{align*}
  \norm*{ \mat{U}^\intercal
      \mat{S}'^\intercal \mat{S}' \mat{b}^\perp}_{2}^2
  &=
  \norm*{ \mat{\tilde U}^\intercal \mat{\tilde b}^\perp
   - \parens{\mat{S}\mat{\tilde U}}^\intercal \mat{S} \mat{\tilde b}^\perp }_2^2 \\
  &\le
  \frac{\rank\parens*{\mat{A}_1}}{\beta \norm{\mat{\tilde U}}_\frobenius^2 s }
  \cdot
  \norm{\mat{\tilde U}}_\frobenius^2
  \norm{\mat{\tilde b}^\perp}_2^2 \\
  &\le
  \frac{\rank\parens*{\mat{A}_1}}{\beta s }
  \cdot
  \norm{\mat{b}^\perp}_2^2.
\end{align*}
Since $\mat{b}^\perp$ is the residual vector,
applying Markov's inequality gives us
\begin{align}
\label{eqn:structural_2_sample_bound}
  \Pr\parens*{
     \norm*{ \mat{U}^\intercal
      \mat{S}'^\intercal \mat{S}' \mat{b}^\perp }_{2}^2
      \ge
      \frac{\varepsilon  \norm{\mat{b}^\perp}_{2}^2}{2}
  }
  \le
  \frac{\rank\parens*{\mat{A}_1}}{\beta s }
  \cdot
  \norm{\mat{b}^\perp}_2^2
  \cdot
  \frac{2}{\varepsilon \norm{\mat{b}^\perp}_2^2}
  =
  \frac{2 \cdot \rank\parens*{\mat{A}_1}}{\beta s \varepsilon}.
\end{align}

To upper bound \Cref{eqn:structural_2_sample_bound}
by a failure probability of $\delta / 2$,
the number of samples needs to be at least
\[
  s \ge \frac{4 \cdot \rank(\mat{A}_1)}{\beta \delta \varepsilon}.
\]
\textbf{Conclusion.}
Since $d \ge \rank(\mat{A}_1)$ and $\delta = 1/10$, it follows that
\begin{align*}
  \max\set*{\frac{1680 d \ln(4d/\delta)}{\beta}, \frac{4d}{\beta \delta \varepsilon}}
  \le
  \frac{1680 d \ln(40d) }{\beta \varepsilon}
  \le
  s
\end{align*}
samples are sufficient for both
structural conditions to hold at the same time with probability
at least $1 - (\delta/2 + \delta/2) = 1 - \delta$ by a union bound.
Finally, we may apply \Cref{lem:approx_least_squares} to achieve the
$(1+\varepsilon)$-approximation guarantee.
\end{proof}

\SketchAndSolveRidgeRegression*

\begin{proof}
This is an immediate consequence of our results for
approximate block regression in \Cref{lemma:approximate_block_regression}.
Consider the augmented matrices
\[
  \mat{A}' = \begin{bmatrix}
    \mat{A} \\
    \sqrt{\lambda} \mat{I}_d
  \end{bmatrix}
  ~~~\text{and}~~~
  \mat{b}' =
  \begin{bmatrix}
    \mat{b} \\
    \mat{0_{d}}
  \end{bmatrix}.
\]
For any $\mat{x} \in \R^d$, we have
\begin{align*}
  \norm*{\mat{A}' \mat{x} - \mat{b}'}_2^2
  =
  \norm*{
  \begin{bmatrix}
    \mat{A}\mat{x} - \mat{b} \\
    \sqrt{\lambda} \mat{x}
  \end{bmatrix}
  }_2^2
  =
  \norm*{
    \mat{A}\mat{x} - \mat{b}
  }_2^2
  +
  \lambda \norm*{\mat{x}}_2^2.
\end{align*}
Therefore, it suffices to approximately solve 
$\argmin_{\mat{x} \in \R^d} \norm{\mat{A}'\mat{x} - \mat{b}'}$,
so we can 
use~\Cref{lemma:approximate_block_regression} to complete the proof.
\end{proof}

\section{Missing analysis for \Cref{sec:kronecker_regression}}
\label{app:kronecker_regression}

Here we present an explicit formula for the $\lambda$-ridge leverage scores
of a Kronecker matrix in terms of the singular value decompositions of its
factor matrices.
Then we show how to achieve fast Kronecker product-matrix multiplications in
\Cref{app:fast_kronecker_matrix_multiplication}
and prove our main Kronecker regression theorem
in \Cref{app:main_algorithm}.

\KroneckerCrossLeverageScores*

\begin{proof}
For notational brevity, we prove the
claim for $\mat{K} = \mat{A} \ktimes \mat{B} \ktimes \mat{C}$.
The order-$N$ version follows by the same argument.

First, the mixed property property of Kronecker products implies that
\begin{align*}
    \mat{K}^\intercal \mat{K}
    &=
    \parens*{\mat{A}^\intercal \mat{A}} \ktimes \parens*{\mat{B}^\intercal \mat{B}} \ktimes \parens*{\mat{C}^\intercal \mat{C}}.
\end{align*}
Let $\mat{A} = \mat{U}_{\mat{A}} \mat{\Sigma}_{\mat{A}} \mat{V}_{\mat{A}}^\intercal$
be the SVD of $\mat{A}$
such that $\mat{U}_{\mat{A}} \in \R^{I_1 \times I_1}$ and $\mat{V}_{\mat{A}} \in \R^{R_1 \times R_1}$.
The orthogonality of $\mat{U}_{\mat{A}}$ implies that
\[
    \mat{A}^\intercal \mat{A}
    = \mat{V}_{\mat{A}} \mat{\Sigma}_{\mat{A}}^2 \mat{V}_{\mat{A}}^\intercal,
\]
where $\mat{\Sigma}_{\mat{A}}^2$ denotes $\mat{\Sigma}_{\mat{A}}^\intercal \mat{\Sigma}_{\mat{A}}$.
Similarly, let
$\mat{B} = \mat{U}_{\mat{B}} \mat{\Sigma}_{\mat{B}} \mat{V}_{\mat{B}}^\intercal$
and
$\mat{C} = \mat{U}_{\mat{C}} \mat{\Sigma}_{\mat{C}} \mat{V}_{\mat{C}}^\intercal$.
It follows from the mixed-product property that 
\begin{align*}
    \mat{K}^\intercal \mat{K} + \lambda \mat{I}
    &=
    \parens*{\mat{V}_{\mat{A}} \mat{\Sigma}^2_{\mat{A}} \mat{V}_{\mat{A}}^\intercal}
    \ktimes
    \parens*{\mat{V}_{\mat{B}} \mat{\Sigma}^2_{\mat{B}} \mat{V}_{\mat{B}}^\intercal}
    \ktimes
    \parens*{\mat{V}_{\mat{C}} \mat{\Sigma}^2_{\mat{C}} \mat{V}_{\mat{C}}^\intercal}
    + \lambda \mat{I} \\
    &=
    \parens*{\mat{V}_{\mat{A}} \ktimes \mat{V}_{\mat{B}} \ktimes \mat{V}_{\mat{C}}}
    \parens*{\mat{\Sigma}^2_{\mat{A}} \ktimes \mat{\Sigma}^2_{\mat{B}} \ktimes \mat{\Sigma}^2_{\mat{C}}}
    \parens*{\mat{V}_{\mat{A}}^\intercal \ktimes \mat{V}_{\mat{B}}^\intercal \ktimes \mat{V}_{\mat{C}}^\intercal}
    + \lambda \mat{I} \\
    &=
    \parens*{\mat{V}_{\mat{A}} \ktimes \mat{V}_{\mat{B}} \ktimes \mat{V}_{\mat{C}}}
    \parens*{\mat{\Sigma}^2_{\mat{A}} \ktimes \mat{\Sigma}^2_{\mat{B}} \ktimes \mat{\Sigma}^2_{\mat{C}} + \lambda \mat{I}}
    \parens*{\mat{V}_{\mat{A}}^\intercal \ktimes \mat{V}_{\mat{B}}^\intercal \ktimes \mat{V}_{\mat{C}}^\intercal}.
\end{align*}
Since $\parens{\mat{X} \mat{Y}}^+ = \mat{Y}^+ \mat{X}^+$ if $\mat{X}$ or $\mat{Y}$ is orthogonal,
we have
\begin{align*}
    \parens*{\mat{K}^\intercal \mat{K} + \lambda \mat{I}}^{+}
    &=
    \parens*{\parens*{\mat{V}_{\mat{A}} \ktimes \mat{V}_{\mat{B}} \ktimes \mat{V}_{\mat{C}}}
    \parens*{\mat{\Sigma}^2_{\mat{A}} \ktimes \mat{\Sigma}^2_{\mat{B}} \ktimes \mat{\Sigma}^2_{\mat{C}} + \lambda \mat{I}}
    \parens*{\mat{V}_{\mat{A}}^\intercal \ktimes \mat{V}_{\mat{B}}^\intercal \ktimes \mat{V}_{\mat{C}}^\intercal}}^{+} \\
    &=
    \parens*{\mat{V}_{\mat{A}}^\intercal \ktimes \mat{V}_{\mat{B}}^\intercal \ktimes \mat{V}_{\mat{C}}^\intercal}^{+}
    \parens*{\mat{\Sigma}^2_{\mat{A}} \ktimes \mat{\Sigma}^2_{\mat{B}} \ktimes \mat{\Sigma}^2_{\mat{C}} + \lambda \mat{I}}^{+}
    \parens*{\mat{V}_{\mat{A}} \ktimes \mat{V}_{\mat{B}} \ktimes \mat{V}_{\mat{C}}}^{+} \\
    &=
    \parens*{\mat{V}_{\mat{A}} \ktimes \mat{V}_{\mat{B}} \ktimes \mat{V}_{\mat{C}}}
    \parens*{\mat{\Sigma}^2_{\mat{A}} \ktimes \mat{\Sigma}^2_{\mat{B}} \ktimes \mat{\Sigma}^2_{\mat{C}} + \lambda \mat{I}}^{+}
    \parens*{\mat{V}_{\mat{A}}^\intercal \ktimes \mat{V}_{\mat{B}}^\intercal \ktimes \mat{V}_{\mat{C}}^\intercal}.
\end{align*}
Next, observe that
\begin{align*}
    \mat{K}
    &=
    \parens*{\mat{U}_{\mat{A}} \mat{\Sigma}_{\mat{A}} \mat{V}_{\mat{A}}^\intercal}
    \ktimes
    \parens*{\mat{U}_{\mat{B}} \mat{\Sigma}_{\mat{B}} \mat{V}_{\mat{B}}^\intercal}
    \ktimes
    \parens*{\mat{U}_{\mat{C}} \mat{\Sigma}_{\mat{C}} \mat{V}_{\mat{C}}^\intercal} \\
    &=
    \parens*{
    \mat{U}_{\mat{A}} \ktimes \mat{U}_{\mat{B}} \ktimes \mat{U}_{\mat{C}}
    }
    \parens*{
    \mat{\Sigma}_{\mat{A}} \ktimes \mat{\Sigma}_{\mat{B}} \ktimes \mat{\Sigma}_{\mat{C}}
    }
    \parens*{
    \mat{V}_{\mat{A}}^\intercal \ktimes \mat{V}_{\mat{B}}^\intercal \ktimes \mat{V}_{\mat{C}}^\intercal
    }.
\end{align*}
Putting everything together, the $\lambda$-ridge cross leverage scores can be expressed as
\begin{align}
\label{eqn:leverage_score_matrix_eigendecomposition}
    \mat{K}\parens*{\mat{K}^\intercal \mat{K} + \lambda \mat{I}}^{+} \mat{K}^\intercal
    &=
    \parens*{
        \mat{U}_{\mat{A}} \ktimes \mat{U}_{\mat{B}} \ktimes \mat{U}_{\mat{C}}
    }
    \mat{\Lambda}
    \parens*{
        \mat{U}_{\mat{A}} \ktimes \mat{U}_{\mat{B}} \ktimes \mat{U}_{\mat{C}}
    }^\intercal,
\end{align}
where
\[
    \mat{\Lambda} =
    \parens*{\mat{\Sigma}_{\mat{A}} \ktimes \mat{\Sigma}_{\mat{B}} \ktimes \mat{\Sigma}_{\mat{C}}}
    \parens*{\mat{\Sigma}^2_{\mat{A}} \ktimes \mat{\Sigma}^2_{\mat{B}} \ktimes \mat{\Sigma}^2_{\mat{C}} + \lambda \mat{I}}^{+}
    \parens*{\mat{\Sigma}_{\mat{A}} \ktimes \mat{\Sigma}_{\mat{B}} \ktimes \mat{\Sigma}_{\mat{C}}}.
\]
Equation~\Cref{eqn:leverage_score_matrix_eigendecomposition} is the
eigendecomposition of
$\mat{K}\parens*{\mat{K}^\intercal \mat{K} + \lambda \mat{I}}^{+} \mat{K}^\intercal$.
In particular, $\mat{\Lambda} \in \R^{I_1 I_2 I_3 \times I_1 I_2 I_3}$
is a diagonal matrix of eigenvalues, where
the $(i_1,i_2,i_3)$-th eigenvalue is
\begin{align}
\label{eqn:leverage_score_matrix_eigenvalues}
  \lambda_{(i_1,i_2,i_3)} =
    \frac{\sigma_{i_1}^2\parens*{\mat{A}} \sigma_{i_2}^2\parens*{\mat{B}} \sigma_{i_3}^2\parens*{\mat{C}}}{\sigma_{i_1}^2\parens*{\mat{A}} \sigma_{i_2}^2\parens*{\mat{B}} \sigma_{i_3}^2\parens*{\mat{C}} + \lambda}.
\end{align}

The value of
$\ell^\lambda_{(i_1,i_2,i_3),(j_1,j_2,j_3)}\parens{\mat{K}}$
follows from the definition of cross $\lambda$-ridge leverage scores
in Equation~\Cref{eqn:ridge_leverage_score_svd_def}.

Finally, the statistical leverage score property holds because
setting $\lambda = 0$ gives an expression that is the product of the
leverage scores of the factor matrices.
\end{proof}

\subsection{Fast Kronecker-Matrix Multiplication}
\label{app:fast_kronecker_matrix_multiplication}

\KronMatMulLemma*

\begin{proof}
We prove the claim by induction on $N$.
Our approach will be to show that we can extract the rightmost factor matrix
out of the Kronecker product and solve a smaller instance recursively.

If $N=1$, this is a standard instance of matrix-matrix multiplication
that takes $O(I_1 J_1 K)$ time. 
Now let
\[
  \mat{X} = \mat{A}^{(1)} \ktimes \cdots \ktimes \mat{A}^{(N)}
  \in \R^{P \times Q}
\]
and
\[
  \mat{Y} = \mat{A}^{(N+1)} \in \R^{R \times S}
\]
Let $\mat{b} \in \R^{QS}$ be an arbitrary column of
$\mat{B}$.
We compute each of these $K$ matrix-vector products separately.
Now we show how to efficiently compute
$\mat{c} = (\mat{X} \ktimes \mat{Y})\mat{b} \in \R^{PR}$.
The entry in $\mat{c}$ at the canonical index $(p,r)$
is
\[
  c_{pr}
  =
  \parens*{
    \mat{x}_{p, :} \ktimes
    \mat{y}_{r, :}
  } \mat{b}.
\]
Writing this out, we have
\begin{align*}
  c_{pr}
  &=
  \sum_{q=1}^Q
  \sum_{s=1}^S
  x_{p,q}
  y_{r,s}
  b_{qs} \\
  &=
  \sum_{q=1}^Q
  x_{p,q}
  \sum_{s=1}^S
  y_{r,s}
  b_{qs}.
\end{align*}
Therefore, for each $(q,r)$ we can precompute
\[
  z_{q,r} = \sum_{s=1}^S y_{r,s} b_{qs}.
\]
Computing all of $\mat{Z} \in \R^{Q \times R}$
takes $O(QRS)$ time.
Now that we have $\mat{Z}$, we can compute the output~$\mat{c}$ as:
\begin{align*}
  c_{pr}
  &=
  \sum_{q=1}^Q
  x_{p,q}
  \sum_{s=1}^S
  y_{r,s}
  b_{qs} \\
  &=
  \sum_{q=1}^Q
  x_{p,q}
  z_{q,r}.
\end{align*}
Therefore, we can write a natural matricized version of $\mat{c}$ as
\[
  \mat{C} = \mat{X} \mat{Z} \in \R^{P \times R}.
\]
This matrix $\mat{C}$ can be computed recursively
since $\mat{X}$ is a Kronecker product.
Translating back to the original dimensions as stated in the lemma, we have
$P = I_1 \cdots I_N$, $Q = J_1 \cdots J_N$, $R = I_{N+1}$, and $S = J_{N+1}$.
Computing $\mat{Z}$ takes time
\[
  O(QRS) = O(J_1 \cdots J_N I_{N+1} J_{N+1}) = O(J_1 \cdots J_{N+1} I_{N+1}).
\]
By induction, the recursive solve for $\mat{XZ}$ takes time 
\[
  O\parens*{I_{N+1} \sum_{n=1}^N J_1 \cdots J_{n} I_n \cdots I_N}.
\]
Adding the two running times together and accounting for all $K$ columns
of $\mat{B}$ gives us a total running time of
\[
  O\parens*{K \parens*{ J_1 \cdots J_{N+1} I_{N+1}
    + 
    I_{N+1} \sum_{n=1}^N J_1 \cdots J_{n} I_n \cdots I_N
  }}
  =
O\parens*{K 
    \sum_{n=1}^{N+1} J_1 \cdots J_{n} I_n \cdots I_{N+1}
  },
\]
which completes the proof.
\end{proof}

\begin{lemma}[\citet{zhang2013kronecker}]
\label{lemma:mixed_kron_vec_mult}
For a matrix $\mat{C}=[\mat{c}_1,\ldots,\mat{c}_q]\in \mathbb{R}^{p\times q}$, let 
\[
  \vectorize(\mat{C}) = \begin{bmatrix}
    \mat{c}_1 \\
\vdots \\
    \mat{c}_q
\end{bmatrix} \in \mathbb{R}^{pq}.
\]
  Let $\mat{A}\in\mathbb{R}^{m \times q}$ and $\mat{B}\in\mathbb{R}^{n\times p}$. Then, $\vectorize(\mat{B} \mat{C} \mat{A}^\intercal) = (\mat{A} \ktimes \mat{B})\vectorize(\mat{C})$.
\end{lemma}

\begin{theorem}
  Let $\mat{A}^{(1)} \in\mathbb{R}^{R_1 \times R_1},\dots,\mat{A}^{(N)} \in\mathbb{R}^{R_N \times R_N}$
  and $\mat{c}\in\mathbb{R}^{R_1\cdots R_N}$. Let $R=\prod_{n=1}^N R_n$.
  Then $(\mat{A}^{(1)} \ktimes \cdots \ktimes \mat{A}^{(N)}) \mat{c}$ can be computed in time $O(R \sum_{n=1}^N R_n)$.
\end{theorem}

\begin{proof}
  Let $\mat{C}$ be an $R_N \times (R_{1}\cdots R_{N-1})$ matrix such that $\mat{c} = \vectorize(\mat{C})$ (see Lemma \ref{lemma:mixed_kron_vec_mult}).
  For each $n\in[N]$, let $\mat{I}_n$ be the identity matrix of size $R_n\times R_n$. Then by Lemma \ref{lemma:mixed_kron_vec_mult},
\begin{align}
\nonumber
  \parens*{\mat{A}^{(1)} \ktimes \dots \ktimes \mat{A}^{(N)}} \mat{c} 
    &= 
  \parens*{\mat{A}^{(1)} \ktimes \dots \ktimes \mat{A}^{(N)}} \vectorize(\mat{C}) \\
    &= \nonumber
    \vectorize\parens*{\mat{I}_N \mat{A}^{(N)} \mat{C} \parens*{\mat{A}^{(1)} \ktimes \cdots \ktimes \mat{A}^{(N-1)}}^\intercal}
\\ & = \label{eq:square_kron_mult}
\parens*{\mat{A}^{(1)} \ktimes \cdots \ktimes \mat{A}^{(N-1)} \ktimes \mat{I}_N} \vectorize\parens*{\mat{A}^{(N)} \mat{C}}.
\end{align}
Since $\mat{A}^{(N)}$ is $R_N\times R_N$ and
$\mat{C}$ is $R_N \times (R_{1}\cdots R_{N-1})$, $\mat{A}^{(N)} \mat{C}$ can be computed in time $O(R_N R)$.

Now note that although Kronecker product is not commutative, $\mat{A}\ktimes \mat{B}$ and $\mat{B}\ktimes \mat{A}$ are permutation equivalent, i.e., there are permutation matrices that transform one to the other. 
Therefore, instead of computing $(\mat{A}^{(1)} \ktimes \cdots \ktimes \mat{A}^{(N-1)} \ktimes \mat{I}_N) \vectorize(\mat{A}^{(N)}\mat{C})$,
we can compute $(\mat{I}_N \ktimes \mat{A}^{(1)} \ktimes \cdots \ktimes \mat{A}^{(N-1)}) \mat{c}_1$,
where~$\mat{c}_1$ is a permutation of $\vectorize(\mat{A}^{(N)}\mat{C})$.
  We proceed with this multiplication and use a technique similar to \eqref{eq:square_kron_mult},
  which results in a cost of $O(R_{N-1} R)$.
  Continue until all the matrices in the Kronecker part are the identity.
  Then we can return a permutation of the final vector since the identity multiplied by a vector is the vector itself.
\end{proof}

\KronMatMulSqrtDecomp*

\begin{proof}
First observe that while the Kronecker product is not commutative, $\mat{A}\ktimes \mat{B}$ and $\mat{B}\ktimes \mat{A}$ are permutation equivalent, i.e., there are permutation matrices that transform one to the other.
Therefore, without loss of generality we assume the minimizer
$\argmin_{T \subseteq [N]} \MM(\prod_{n\in T} R_n, R/\epsilon, \prod_{n\notin T} R_n)$
is the set $[k]$ where $1\leq k \leq N$.
For any diagonal matrix $\mat{S}$, let $S$ be the set corresponding to the indices of the nonzero entries of $\mat{S}$
and let $\mat{I}_S$ be a diagonal matrix where an entry is equal to one if its index is in~$S$ and it is zero otherwise. 

Note that because $\mat{S}$ is an $(I_1\cdots I_N)\times (I_1\cdots I_N)$ matrix, each element of $S$ (i.e., nonzero of~$\mat{S}$)
corresponds to a tuple $(i_1,\ldots,i_N)\in [I_1] \times \cdots \times [I_N]$. Let 
\begin{align*}
&
S_1 = \{(i_1,\ldots, i_{k}): \exists i_{k+1}\in[I_{k+1}],\ldots, i_N\in[I_N] \text{ such that } (i_1,\ldots, i_N)\in S\},
\\ &
S_2 = \{(i_{k+1},\ldots, i_N): \exists i_1\in[I_1],\ldots, i_{k}\in[I_k] \text{ such that } (i_1,\ldots, i_N)\in S\}.
\end{align*}
Let $\mat{B}_S$ be an $(I_1\cdots I_k) \times (I_{k+1} \cdots I_N)$ matrix such that $\mat{S}\mat{b} = \vectorize(\mat{B}_S)$ (see Lemma \ref{lemma:mixed_kron_vec_mult}).
Then by Lemma \ref{lemma:mixed_kron_vec_mult}, we have
\begin{align*}
    \parens{\mat{A}^{(1)} \ktimes \cdots \ktimes \mat{A}^{(N)}}^{\intercal} \mat{S} \mat{b} 
    &= 
    \parens{\mat{A}^{(1)} \ktimes \cdots \ktimes \mat{A}^{(N)}}^{\intercal} \mat{I}_S \mat{S} \mat{b} \\
    &=
    \parens{\mat{A}^{(1)} \ktimes \cdots \ktimes \mat{A}^{(N)}}^{\intercal} \parens*{\mat{I}_{S_1} \ktimes \mat{I}_{S_2}} \mat{S} \mat{b} \\
    &= 
\parens{\parens{\mat{A}^{(1)} \ktimes \cdots \ktimes \mat{A}^{(k)}}^{\intercal} \mat{I}_{S_1}} \ktimes
       \parens{\parens{\mat{A}^{(k+1)} \ktimes \cdots \ktimes \mat{A}^{(N)}}^{\intercal} \mat{I}_{S_2}} \vectorize(\mat{B}_S) \\
    &=
    \vectorize \parens{\parens{\parens{\mat{A}^{(k+1)} \ktimes \cdots \ktimes \mat{A}^{(N)}}^{\intercal} \mat{I}_{S_2}}
    \mat{B}_S \parens{\mat{I}_{S_1}\parens{\mat{A}^{(1)} \ktimes \cdots \ktimes \mat{A}^{(k)}}}}.
\end{align*}
Note that the number of nonzeros in $\mat{B}_S$ is equal to the number of nonzeros in $\mat{S}$, which is $O(R/\epsilon)$.
Therefore, $((\mat{A}^{(k+1)} \ktimes \cdots \ktimes \mat{A}^{(N)})^{\intercal} \mat{I}_{S_2}) \mat{B}_S$ can be computed in
$O(R / \epsilon \prod_{n=k+1}^N R_n)$ time
because $(\mat{A}^{(k+1)} \ktimes \cdots \ktimes \mat{A}^{(N)})^{\intercal} \mat{I}_{S_2}$ has $\prod_{n=k+1}^N R_n$ rows and $\mat{B}_S$ has $O(R/\epsilon)$ nonzero entries. Moreover,
\begin{align*}
  O\parens*{\frac{R}{\epsilon} \prod_{n=k+1}^N R_n}
  &=
  O\parens*{\MM\parens*{1, \frac{R}{\epsilon}, \prod_{n=k+1}^N R_n}} \\
  &=
  O\parens*{\MM\parens*{\prod_{n=1}^k R_n, \frac{R}{\epsilon}, \prod_{n=k+1}^N R_n}}.
\end{align*}

Now, note that $|S_1|\leq |S| = O(R/\epsilon)$.
Therefore, multiplying $((\mat{A}^{(k+1)} \ktimes \cdots \ktimes \mat{A}^{(N)})^{\intercal} \mat{I}_{S_2}) \mat{B}_S$
with $(\mat{I}_{S_1}(\mat{A}^{(1)} \ktimes \cdots \ktimes \mat{A}^{(k)}))$ can be done in time
$O(\MM(\prod_{n=1}^k R_n, R/\epsilon, \prod_{n=k+1}^N R_n))$
because $((\mat{A}^{(k+1)} \ktimes \cdots \ktimes \mat{A}^{(N)})^{\intercal} \mat{I}_{S_2}) \mat{B}_S$ has $\prod_{n=k+1}^N R_n$ rows and $\mat{A}^{(1)} \ktimes \cdots \ktimes \mat{A}^{(k)}$ has $\prod_{n=1}^k R_n$ columns.

Now let $\mat{C}$ be an $(R_{k+1}\cdots R_N) \times (R_{1} \cdots R_k)$ matrix such that $\mat{c} = \vectorize(\mat{C})$
(see Lemma \ref{lemma:mixed_kron_vec_mult}). Then we have
\begin{align*}
\mat{S} (\mat{A}^{(1)} \ktimes \cdots \ktimes \mat{A}^{(N)}) \mat{c}
& = 
\mat{S} \mat{I}_S (\mat{A}^{(1)} \ktimes \cdots \ktimes \mat{A}^{(N)}) \mat{c}
\\ & = 
\mat{S} (\mat{I}_{S_1} \ktimes \mat{I}_{S_2}) (\mat{A}^{(1)} \ktimes \cdots \ktimes \mat{A}^{(N)}) \mat{c}
\\ & = 
\mat{S} (\mat{I}_{S_1} (\mat{A}^{(1)} \ktimes \cdots \ktimes \mat{A}^{(k)}))
   \ktimes (\mat{I}_{S_2} (\mat{A}^{(k+1)} \ktimes \cdots \ktimes \mat{A}^{(N)})) \vectorize(\mat{C})
\\ & =
\mat{S} \vectorize( (\mat{I}_{S_2} (\mat{A}^{(k+1)} \ktimes \cdots \ktimes \mat{A}^{(N)})) \mat{C} ((\mat{A}^{(1)} \ktimes \cdots \ktimes \mat{A}^{(k)})^\intercal \mat{I}_{S_1}))
\end{align*}
We have $|S_2|\leq |S| =O(R/\epsilon)$.
Therefore, $\mat{I}_{S_2} (\mat{A}_{k+1} \ktimes \cdots \ktimes \mat{A}_N)$ has $O(R/\epsilon)$ nonzero entries.
Moreover, $\mat{C}$ is an $(R_{k+1}\cdots R_N) \times (R_{1} \cdots R_k)$ matrix.
Hence, $(\mat{I}_{S_2} (\mat{A}^{(k+1)} \ktimes \cdots \ktimes \mat{A}^{(N)})) \mat{C}$
can be computed in time $O(\MM(\frac{R}{\epsilon}, \prod_{n=k+1}^{N} R_n, \prod_{n=1}^{k} R_n))
= O(\MM(\prod_{n=1}^{k} R_n, \frac{R}{\epsilon}, \prod_{n=k+1}^{N} R_n))$.

Observe that we do not need to compute all entries of
\[
(\mat{I}_{S_2} (\mat{A}^{(k+1)} \ktimes \cdots \ktimes \mat{A}^{(N)})) \mat{C}
((\mat{A}^{(1)} \ktimes \cdots \ktimes \mat{A}^{(k)})^\intercal \mat{I}_{S_1}).
\]
Instead, we only need to compute entries corresponding to nonzero entries of $\mat{S}$.
Computing each such entry takes $O(\prod_{n=1}^{k} R_n)$ time because
$(\mat{I}_{S_2} (\mat{A}^{(k+1)} \ktimes \cdots \ktimes \mat{A}^{(N)})) \mat{C}$
and $((\mat{A}^{(1)} \ktimes \cdots \ktimes \mat{A}^{(k)})^\intercal \mat{I}_{S_1})$ have $\prod_{n=1}^{k} R_n$ columns and rows, respectively.
Moreover, the number of nonzeros is $\mat{S}$ is $\tilde{O}(R/\epsilon)$.
Therefore, computing all entries of
$(\mat{I}_{S_2} (\mat{A}^{(k+1)} \ktimes \cdots \ktimes \mat{A}^{(N)})) \mat{C} ((\mat{A}^{(1)} \ktimes \cdots \ktimes \mat{A}^{(k)})^\intercal \mat{I}_{S_1})$ that correspond to nonzero entries of $\mat{S}$ takes time
\[
    O\parens*{\frac{R}{\epsilon} \prod_{n=1}^k R_n}
    =
    O\parens*{\MM\parens*{\prod_{n=1}^k, \frac{R}{\epsilon}, 1}}
    =
    O\parens*{\MM\parens*{\prod_{n=1}^k R_n, \frac{R}{\epsilon}, \prod_{n=k+1}^N R_n}}. \qedhere
\]
\end{proof}

\subsection{Main Algorithm}
\label{app:main_algorithm}

\begin{lemma}[Johnson--Lindenstrauss random projection \cite{johnson1984extensions,arriaga2006algorithmic}]
\label{lemma:random-projection}
Let $\mat{x}\in\mathbb{R}^d$. Assume the entries in $\mat{G}\in \mathbb{R}^{r\times d}$ are sampled independently from $\mathcal{N}(0,1)$. Then,
\[
\Pr\left((1-\epsilon)\norm{\mat{x}}_2^2 \leq \norm*{\frac{1}{\sqrt{r}}\mat{G}\mat{x}}_2^2 \leq (1+\epsilon) \norm{\mat{x}}_2^2\right)
\geq
1- 2e^{-\left(\epsilon^2-\epsilon^3\right) r /4}.
\]
\end{lemma}

\begin{lemma}
\label{lemma:fast-ls-comp}
Let $\mat{A}\in\mathbb{R}^{n\times d}$ and $0<\epsilon\leq 1/4$. Given $\mat{\tilde{A}}\in \mathbb{R}^{k\times d}$ and $\mat{N} = \mat{\tilde{A}}^\intercal \mat{\tilde{A}}\in\mathbb{R}^{d\times d}$ such that
\[
\mat{A}^\intercal \mat{A} \preccurlyeq \mat{\tilde{A}}^\intercal \mat{\tilde{A}} \preccurlyeq (1+\varepsilon/4) \mat{A}^\intercal \mat{A},
\]
with high probability, all leverage scores of $\mat{A}$ can be computed to a $(1+\varepsilon/2)$ approximation in $\tilde{O}(\textnormal{nnz}(\mat{A})+kd+d^\omega)$ time.
\end{lemma}

\begin{proof}
Let $\mat{M} = (1+\varepsilon/4) (\mat{\tilde{A}}^\intercal \mat{\tilde{A}})^{-1}$.
It follows that $(\mat{A}^\intercal \mat{A})^{-1} \preccurlyeq \mat{M} \preccurlyeq (1+\varepsilon/4) (\mat{A}^\intercal \mat{A})^{-1}$. Hence, for any $\mat{x}\in\mathbb{R}^d$, we have
\[
\mat{x}^\intercal (\mat{A}^\intercal \mat{A})^{-1} \mat{x} \preccurlyeq \mat{x}^\intercal \mat{M} \mat{x} \preccurlyeq (1+\varepsilon/4) \mat{x}^\intercal (\mat{A}^\intercal \mat{A})^{-1} \mat{x}.
\]
Now note that $\mat{M} = \mat{M} \mat{M}^{-1} \mat{M} = \frac{1}{1+\varepsilon/4} \mat{M} \mat{\tilde{A}}^\intercal \mat{\tilde{A}} \mat{M}$. Hence,
\[
\mat{x}^\intercal \mat{M} \mat{x} = \frac{1}{1+\varepsilon/4} \mat{x}^\intercal \mat{M} \mat{\tilde{A}}^\intercal \mat{\tilde{A}} \mat{M} \mat{x} = \frac{1}{1+\varepsilon/4} \norm{\mat{\tilde{A}} \mat{M} \mat{x}}_2^2.
\]
Using Lemma \ref{lemma:random-projection} with $\epsilon/20$ and $r=O(\log n)$, we can compute a random matrix $\mat{G}$ such that, with high probability, for all $\mat{a}_i$, we have
\[
\norm{\mat{\tilde{A}} \mat{M} \mat{a}_i}_2^2 \leq \frac{1}{1-\varepsilon/20} \norm{\mat{G}\mat{\tilde{A}} \mat{M} \mat{a}_i}_2^2 \leq \frac{1+\varepsilon/20}{1-\varepsilon/20} \norm{\mat{\tilde{A}} \mat{M} \mat{a}_i}_2^2 \leq (1+\varepsilon/6) \norm{\mat{\tilde{A}} \mat{M} \mat{a}_i}_2^2.
\]
Combining the above, we have
\[
\mat{a}_i^\intercal (\mat{A}^\intercal \mat{A})^{-1} \mat{a}_i \leq \mat{a}_i^\intercal \mat{M} \mat{a}_i = \frac{1}{1+\epsilon/4} \norm{\mat{\tilde{A}} \mat{M} \mat{a}_i}_2^2 \leq
\frac{1}{(1+\varepsilon/4)(1-\varepsilon/20)} \norm{\mat{G}\mat{\tilde{A}} \mat{M} \mat{a}_i}_2^2,
\]
and
\[
\frac{1}{(1+\varepsilon/4)(1-\varepsilon/20)} \norm{\mat{G}\mat{\tilde{A}} \mat{M} \mat{a}_i}_2^2
\leq
\mat{a}_i^\intercal \mat{M} \mat{a}_i
\leq
(1+\varepsilon/4) \mat{a}_i^\intercal (\mat{A}^\intercal \mat{A})^{-1} \mat{a}_i.
\]
Therefore, $\frac{1}{(1+\varepsilon/4)(1-\varepsilon/20)} \norm{\mat{G}\mat{\tilde{A}} \mat{M} \mat{a}_i}_2^2$ is a $1+\varepsilon/4\leq 1+\epsilon/2$ approximation of the leverage score of $\mat{a}_i$.

Lastly, we discuss the running time. Note that given $\mat{N}$, we can compute $\mat{M}$ in $\tilde{O}(d^\omega)$ time.
Moreover, since $\mat{G}$ has $\tilde{O}(1)$ rows, $\mat{G}\mat{\tilde{A}}$ and $\mat{G}\mat{\tilde{A}} \mat{M}$ can be computed in $\tilde{O}(kd+d^2)$ time. Finally, given $\mat{G}\mat{\tilde{A}} \mat{M}$, we can
compute $\mat{G}\mat{\tilde{A}} \mat{M} \mat{a}_i$, for all $i\in[n]$, in $\tilde{O}(\textnormal{nnz}(\mat{A}))$.
\end{proof}

\FastKroneckerRegressionTheorem*

\begin{proof}
By \cite[Lemma 8]{cohen2015uniform}, we can
compute a $(1 + \varepsilon/N)$-spectral approximation $\mat{\tilde{A}}^{(n)}$ of $\mat{A}^{(n)}$,
with $\tilde{O}(R_n N^2 \epsilon^{-2})$ rows,
in $\tilde{O}(\textnormal{nnz}(\mat{A}^{(n)})+R_n^\omega N^2 \epsilon^{-2})$ time. Given $\mat{\tilde{A}}^{(n)}$, we can compute ${\tilde{\mat{A}}^{(n)^\intercal}} \tilde{\mat{A}}^{(n)}$ in $\tilde{O}(R_n^\omega N^2 \epsilon^{-2})$ time.
Finally, given ${\tilde{\mat{A}}^{(n)^\intercal}} \tilde{\mat{A}}^{(n)}\in\mathbb{R}^{R_n\times R_n}$, we can compute its SVD in $\tilde{O}(R_n^\omega)$ time.

Given $\mat{\tilde{A}}^{(n)}$ and ${\mat{\tilde{A}}^{(n)^\intercal}} \mat{\tilde{A}}^{(n)}$, by \cref{lemma:fast-ls-comp},
it takes
$\tilde{O}(\textnormal{nnz}(\mat{A}^{(n)})+R_n^\omega N^2 \epsilon^{-2})$ time
to compute the (approximate) leverage scores $\bm{\ell}(\mat{A}^{(n)})$.
Compute the cumulative density function of each 
leverage score distribution in $O(I_n)$ time.
This allows us to sample from the product distribution
$\cP = \bm{\ell}(\mat{A}^{(1)}) \ktimes \cdots \ktimes \bm{\ell}(\mat{A}^{(N)})$
in $\tilde{O}(N)$ by sampling each coordinate independently.
Sampling from $\cP$ is equivalent to sampling from $\bm{\ell}(\mat{K})$ by
\Cref{lemma:kronecker_cross_leverage_scores}.

Note that 
\[
(1+\log(1+\epsilon/4)/N)^N \leq (e^{\log(1+\epsilon/4)/N})^N = e^{\log(1+\epsilon/4)} = 1+\epsilon/4.
\]
Therefore, because for all $n\in[N]$, we have
\[
{\mat{A}^{(n)^\intercal}} \mat{A}^{(n)} \preccurlyeq
{\mat{\tilde{A}}^{(n)^\intercal}} \mat{\tilde{A}}^{(n)} \preccurlyeq (1+\log(1+\epsilon/4)/N) {\mat{A}^{(n)^\intercal}} \mat{A}^{(n)},
\]
it follows that
\begin{align*}
({\mat{A}^{(1)^\intercal}} \mat{A}^{(1)})\ktimes\cdots\ktimes ({\mat{A}^{(N)}}^\intercal \mat{A}^{(N)}) & \preccurlyeq ({\mat{\tilde{A}}^{(1)^\intercal}} \mat{\tilde{A}}^{(1)})\ktimes\cdots\ktimes ({\mat{\tilde{A}}^{(N)^\intercal}} \mat{\tilde{A}}^{(N)})
\\ & \preccurlyeq
(1+\varepsilon/4) ({\mat{A}^{(1)^\intercal}} \mat{A}^{(1)})\ktimes\cdots\ktimes ({\mat{A}^{(N)^\intercal}} \mat{A}^{(N)}).
\end{align*}
Thus, the approximate leverage scores we get in
\cref{alg:fast_kronecker_regression} for $\mat{A}^{(1)}\ktimes\cdots\ktimes \mat{A}^{(N)}$ are within a factor of $(1+\varepsilon/4)$ of the true leverage scores.

Therefore, our preconditioner given by the SVD, $\mat{V}_{\mat{K}}
  \parens*{
    \mat{\Sigma}_{\mat{K}}^\intercal \mat{\Sigma}_{\mat{K}} + \lambda\mat{I}_R }^+
  \mat{V}_{\mat{K}}^\intercal$, is a spectral approximation of $\parens*{\mat{K}^\intercal \mat{K} + \lambda\mat{I}_R}^+$.
  More specifically,
\begin{align*}
    \mat{\tilde K}^\intercal \mat{\tilde K} + \lambda\mat{I} &
  \preccurlyeq
  \frac{1}{1-\sqrt{\varepsilon}}
  \parens*{\mat{K}^\intercal \mat{K} + \lambda\mat{I}_R} 
  \preccurlyeq
  \frac{1}{1-\sqrt{\varepsilon}} \mat{V}_{\mat{K}}
  \parens*{
    \mat{\Sigma}_{\mat{K}}^\intercal \mat{\Sigma}_{\mat{K}} + \lambda\mat{I}_R }^+
  \mat{V}_{\mat{K}}^\intercal
  \\ &
  \preccurlyeq \frac{1+\epsilon/4}{1-\sqrt{\epsilon}} \parens*{\mat{K}^\intercal \mat{K} + \lambda\mat{I}_R} 
  \preccurlyeq
  \frac{(1+\epsilon/4)(1+\sqrt{\varepsilon})}{1-\sqrt{\varepsilon}}
  \parens*{\mat{\tilde K}^\intercal \mat{\tilde K} + \lambda \mat{I}}.
\end{align*}
Therefore, by $O(\log(1/\epsilon))$ iterations of Richardson (\cref{lemma:richardson_iteration}), we converge to the desired accuracy.
Finally, note that each iteration of Richardson can be done in time
\[
\tilde{O}\parens*{
    \min_{S\subseteq [N]} \textnormal{MM}\parens*{
        \prod_{n\in S} R_n, R \varepsilon^{-1}, \prod_{n\in [N]\setminus S} R_n
    }
},
\]
using our novel sparse Kronecker-matrix multiplication procedure (\cref{thm:kron_mat_mul_sqrt_decomp}), $\KronMatMul$ (\cref{lemma:kron_mat_mul}), and the structure of the preconditioner, which is a diagonal matrix multiplied from left and right by matrices with Kronecker structure.
\end{proof}
\section{Missing Analysis from \Cref{sec:algorithm}}
\label{app:algorithm}

In this section we show how to extend the \FastKroneckerRegression
(\Cref{alg:fast_kronecker_regression})
to work for factor matrix updates in \texttt{TuckerALS}.
This analysis is presented in~\Cref{subsec:equality_constrained_least_squares_reduction} and~\Cref{subsec:fast_factor_matrix_update}.
Then we compare the running times of different core tensor and
factor matrix update algorithms in~\Cref{subsec:different_algorithms_running_times} to establish baselines.

\subsection{Factor Matrix Update Substitution: Reducing to Equality-Constrained Least Squares}
\label{subsec:equality_constrained_least_squares_reduction}

\LemmaConstrainedLeastSquares*

\begin{proof}
Let
$\mat{z} = \mat{M} \mat{x} \in \R^{m}$.
For any $\mat{x} \in \R^{d}$,
$\mat{z}$ is in the column space of $\mat{M}$
and hence
orthogonal to any vector in the left null space
of $\mat{M}$.
Therefore, we can optimize over $\mat{z} \in \R^{m}$ subject to
$\mat{N}\mat{z} = \mat{0}$ instead
because for any $\mat{x} \in \R^{d}$,
$
  \mat{N} \mat{M} \mat{x} =
  (\mat{I}_m - \mat{M} \mat{M}^+) \mat{M} \mat{x} =
  (\mat{M} - \mat{M})\mat{x} =
  \mat{0}.
$
Using this substitution, we can also replace
the term
$\lambda \norm{\mat{x}}_2^2$
by $\lambda \norm{\mat{M}^+ \mat{z}}_2^2$
because for any $\mat{z}$,
the least squares solution to
$\mat{z} = \mat{M} \mat{x}$
with minimum norm is
$\mat{M}^+ \mat{z}$~\cite{planitz19793}.
\end{proof}

\LemmaConstrainedReg*
\begin{proof}
First note that for any $w\geq 0$, we have
\begin{align}
\label{eq:constraint-reg-0}
\min_{\mat{x} \in \R^{d}} \left\Vert \begin{bmatrix}
  \mat{M} \\ \sqrt{w} \mat{N}
\end{bmatrix} \mat{x} - \begin{bmatrix}
  \mat{b} \\ \mat{0}
\end{bmatrix}\right\Vert_2^2 \leq \min_{\mat{N}\mat{x}=\mat{0}} \left\Vert \begin{bmatrix}
  \mat{M} \\ \sqrt{w} \mat{N}
\end{bmatrix} \mat{x} - \begin{bmatrix}
  \mat{b} \\ \mat{0}
\end{bmatrix}\right\Vert_2^2 = \min_{\mat{N}\mat{x}=\mat{0}} \norm{\mat{M}\mat{x}-\mat{b}}_2^2.
\end{align}
Suppose $\mat{\hat{x}}\in\mathbb{R}^d$ such that
\begin{align}
\label{eq:constraint-reg-1}
\left\Vert \begin{bmatrix}
  \mat{M} \\ \sqrt{w} \mat{N}
\end{bmatrix} \mat{\hat{x}} - \begin{bmatrix}
  \mat{b} \\ \mat{0}
\end{bmatrix}\right\Vert_2^2
    \leq
 (1+\varepsilon/3) \min_{\mat{x} \in \R^{d}} \left\Vert \begin{bmatrix}
  \mat{M} \\ \sqrt{w} \mat{N}
\end{bmatrix} \mat{x} - \begin{bmatrix}
  \mat{b} \\ \mat{0}
\end{bmatrix}\right\Vert_2^2.
\end{align}
Let $\mat{z} = (\mat{I} - \mat{N}^+ \mat{N})\mat{\hat{x}}$.
It follows that $\mat{N}\mat{z} = \mat{0}$ because $\mat{N}=\mat{N}\mat{N}^+ \mat{N}$.
Therefore,
\[
\left\Vert \begin{bmatrix}
  \mat{M} \\ \sqrt{w} \mat{N}
\end{bmatrix} \mat{z} - \begin{bmatrix}
  \mat{b} \\ \mat{0}
\end{bmatrix}\right\Vert_2^2 = \left\Vert \begin{bmatrix}
  \mat{M} \\ \sqrt{w} \mat{N}
\end{bmatrix} (\mat{I} - \mat{N}^+ \mat{N})\mat{\hat{x}} - \begin{bmatrix}
  \mat{b} \\ \mat{0}
\end{bmatrix}\right\Vert_2^2 = \norm{\mat{M}(\mat{I}-\mat{N}^+ \mat{N}) \mat{\hat{x}} - \mat{b}}_2^2.
\]

By the triangle inequality,
\begin{align*}
\norm{\mat{M} (\mat{I} - \mat{N}^+ \mat{N}) \mat{\hat{x}} - \mat{b}}_2 \leq \norm{\mat{M} \mat{\hat{x}} - \mat{b}}_2 + \norm{\mat{M} \mat{N}^+ \mat{N} \mat{\hat{x}}}_2.
\end{align*}
Therefore,
\begin{align*}
\norm{\mat{M} (\mat{I} - \mat{N}^+ \mat{N}) \mat{\hat{x}} - \mat{b}}_2^2 \leq \norm{\mat{M} \mat{\hat{x}} - \mat{b}}_2^2 + \norm{\mat{M} \mat{N}^+ \mat{N} \mat{\hat{x}}}_2^2 + 2\norm{\mat{M} \mat{\hat{x}} - \mat{b}}_2 \norm{\mat{M} \mat{N}^+ \mat{N} \mat{\hat{x}}}_2.
\end{align*}

Now we have two cases: 
\begin{itemize}
    \item Case 1: $2\norm{\mat{M} \mat{N}^+ \mat{N} \mat{\hat{x}}}_2 \leq \frac{\varepsilon}{3} \norm{\mat{M} \mat{\hat{x}} - \mat{b}}_2$,
    \item Case 2: $2\norm{\mat{M} \mat{N}^+ \mat{N} \mat{\hat{x}}}_2 > \frac{\varepsilon}{3} \norm{\mat{M} \mat{\hat{x}} - \mat{b}}_2$.
\end{itemize}
Note that by the consistency of operator norms, we have
\[
\norm{\mat{M} \mat{N}^+ \mat{N} \mat{\hat{x}}}_2 \leq \norm{\mat{M} \mat{N}^+}_2 \norm{\mat{N} \mat{\hat{x}}}_2.
\]
Therefore, in the first case we have
\begin{align*}
\norm{\mat{M} (\mat{I} - \mat{N}^+ \mat{N}) \mat{\hat{x}} - \mat{b}}_2^2 & \leq \parens*{1+\frac{\varepsilon}{3}}\norm{\mat{M} \mat{\hat{x}} - \mat{b}}_2^2 + \norm{\mat{M} \mat{N}^+ \mat{N} \mat{\hat{x}}}_2^2 \\ & \leq \parens*{1+\frac{\varepsilon}{3}}\left(\norm{\mat{M} \mat{\hat{x}} - \mat{b}}_2^2 + w \norm{\mat{N} \mat{\hat{x}}}_2^2\right),
\end{align*}
where the last inequality follows from our choice of $w$.

In the second case we have
\begin{align*}
\norm{\mat{M} (\mat{I} - \mat{N}^+ \mat{N}) \mat{\hat{x}} - \mat{b}}_2^2 
&\leq \norm{\mat{M} \mat{\hat{x}} - \mat{b}}_2^2 + \parens*{1+\frac{12}{\epsilon}} \norm{\mat{M} \mat{N}^+ \mat{N} \mat{\hat{x}}}_2^2 \\
&\leq \norm{\mat{M} \mat{\hat{x}} - \mat{b}}_2^2 + w\norm{\mat{N} \mat{\hat{x}}}_2^2.
\end{align*}
Therefore, in both cases 
\begin{align}
\label{eq:constraint-reg-2}
\norm{\mat{M} (\mat{I} - \mat{N}^+ \mat{N}) \mat{\hat{x}} - \mat{b}}_2^2 \leq \parens*{1+\frac{\varepsilon}{3}}(\norm{\mat{M} \mat{\hat{x}} - \mat{b}}_2^2 + w \norm{\mat{N} \mat{\hat{x}}}_2^2).
\end{align}
Moreover, $\varepsilon<1/3$, so then $(1+\varepsilon/3)^2\leq 1+\varepsilon$. Therefore by Equations~\eqref{eq:constraint-reg-0}, \eqref{eq:constraint-reg-1} and \eqref{eq:constraint-reg-2},
we have
\begin{align*}
\norm{\mat{M}\mat{z}-\mat{b}}_2^2 & = \norm{\mat{M} (\mat{I} - \mat{N}^+ \mat{N}) \mat{\hat{x}} - \mat{b}}_2^2 \leq (1+\varepsilon)\min_{\mat{x} \in \R^d} \left\Vert \begin{bmatrix}
  \mat{M} \\ \sqrt{w} \mat{N}
\end{bmatrix} \mat{x} - \begin{bmatrix}
  \mat{b} \\ \mat{0}
\end{bmatrix}\right\Vert_2^2 \\ & \leq (1+\varepsilon)\min_{\mat{N}\mat{x}=0} \norm{\mat{M}\mat{x}-\mat{b}}_2^2.
\end{align*}
Finally, note that $\mat{N}\mat{z}=\mat{0}$ and that 
$\mat{z}$ is a feasible solution.
\end{proof}

\paragraph{Application to factor matrix updates.}
Now we explain how the reduction to equality-constrained least squares in~\Cref{lemma:constrained_least_squares}
applies to factor matrix updates in \texttt{TuckerALS}.
For the factor matrix updates, we solve a regression problem of the form:
\begin{align}
\label{eq:factor_matrix_reg}
  \argmin_{\mat{y} \in \R^{R_n}} \norm{\parens{\mat{A}^{(1)} \otimes \dots \otimes \mat{A}^{(n-1)} \otimes \mat{A}^{(n+1)} \otimes \dots \otimes \mat{A}^{(N)}} \mat{G}_{(n)}^\intercal \mat{y} - \mat{b}_{i:}^\intercal}_2^2  + \lambda \norm{\mat{y}}_2^2, 
\end{align}
where $\mat{b}_{i:}$ is the $i$-th row of the mode-$n$ unfolding of tensor $\tensor{X}$.
Note that for any $\mat{y}$, $\mat{G}_{(n)}^\intercal \mat{y}$ is a vector in the column space of $\mat{G}_{(n)}^\intercal$. Thus, $\mat{G}_{(n)}^\intercal \mat{y}$ is orthogonal to any vector in the left null space of $\mat{G}_{(n)}^\intercal$. Let $\mat{N}$ be a matrix in which the rows are a basis for the left null space of $\mat{G}_{(n)}^\intercal$. Then, solving the following is equivalent to solving \eqref{eq:factor_matrix_reg}:
\begin{align}
\label{eq:factor_matrix_reg_2}
  \min_{\mat{N} \mat{z} = \mat{0}} & \norm{\parens{\mat{A}^{(1)} \otimes \dots \otimes \mat{A}^{(n-1)} \otimes \mat{A}^{(n+1)} \otimes \dots \otimes \mat{A}^{(N)}} \mat{z} - \mat{b}_{i:}^\intercal}_2^2   + \lambda \norm{(\mat{G}_{(n)}^\intercal)^+ \mat{z}}_2^2,
\end{align}
where $(\mat{G}_{(n)}^\intercal)^+$ is the pseudoinverse of $\mat{G}_{(n)}^\intercal$.

To explain the $\norm{(\mat{G}_{(n)}^\intercal)^+ \mat{z}}_2^2$ term, consider a vector $\mat{y}$ that under the transformation $\mat{G}_{(n)}^\intercal$ goes to $\mat{z}$, i.e., $\mat{G}_{(n)}^\intercal \mat{y} = \mat{z}$. The set of solutions to this linear system is $(\mat{G}_{(n)}^\intercal)^+ \mat{z} + (\mat{I}- (\mat{G}_{(n)}^\intercal)^+ \mat{G}_{(n)}^\intercal)\mat{w}$, for all $\mat{w}$ by~\cite{james1978generalised}.
Moreover, $(\mat{G}_{(n)}^\intercal)^+ \mat{z}$ is orthogonal to $(\mat{I}- (\mat{G}_{(n)}^\intercal)^+ \mat{G}_{(n)}^\intercal)\mat{w}$ because
\begin{align*}
  ((\mat{G}_{(n)}^\intercal)^+ \mat{z})^\intercal (\mat{I}- (\mat{G}_{(n)}^\intercal)^+ \mat{G}_{(n)}^\intercal)\mat{w} 
  &= ((\mat{G}_{(n)}^\intercal)^+ \mat{G}_{(n)}^\intercal(\mat{G}_{(n)}^\intercal)^+ \mat{z})^\intercal (\mat{I}- (\mat{G}_{(n)}^\intercal)^+ \mat{G}_{(n)}^\intercal)\mat{w}
\\
  &= ((\mat{G}_{(n)}^\intercal)^+ \mat{z})^\intercal (\mat{G}_{(n)}^\intercal)^+ \mat{G}_{(n)}^\intercal (\mat{I}- (\mat{G}_{(n)}^\intercal)^+ \mat{G}_{(n)}^\intercal)\mat{w} \\
  &= ((\mat{G}_{(n)}^\intercal)^+ \mat{z})^\intercal (\mat{G}_{(n)}^\intercal)^+ (\mat{G}_{(n)}^\intercal- \mat{G}_{(n)}^\intercal(\mat{G}_{(n)}^\intercal)^+ \mat{G}_{(n)}^\intercal)\mat{w} \\
  &= 0.
\end{align*}
Therefore, by the Pythagorean theorem, we have
\begin{align*}
  \norm{(\mat{G}_{(n)}^\intercal)^+ \mat{z} + (\mat{I}- (\mat{G}_{(n)}^\intercal)^+ \mat{G}_{(n)}^\intercal)\mat{w}}_2^2 
  = \norm{(\mat{G}_{(n)}^\intercal)^+ \mat{z}}_2^2 + \norm{(\mat{I}- (\mat{G}_{(n)}^\intercal)^+ \mat{G}_{(n)}^\intercal)\mat{w}}_2^2,
\end{align*}
so
$\norm{(\mat{G}_{(n)}^\intercal)^+ \mat{z} + (\mat{I}- (\mat{G}_{(n)}^\intercal)^+ \mat{G}_{(n)}^\intercal)\mat{w}}_2^2$ is minimized when $\mat{w}=\mat{0}$ \cite{planitz19793}.
Thus, for all $y$ such that $\mat{G}_{(n)}^\intercal \mat{y} = \mat{z}$,
it follows that $(\mat{G}_{(n)}^\intercal)^+ \mat{z}$ minimizes $\norm{\mat{y}}_2^2$
in \eqref{eq:factor_matrix_reg},
hence we can replace $\mat{y}$ by $(\mat{G}_{(n)}^\intercal)^+ \mat{z}$.

Note that $\mat{N}=\mat{I} - \mat{G}_{(n)}^\intercal (\mat{G}_{(n)}^\intercal)^{+}$ works because for any $\mat{y}$, by definition of pseudoinverse:
\[
  (\mat{I} - \mat{G}_{(n)}^\intercal (\mat{G}_{(n)}^\intercal)^{+}) \mat{G}_{(n)}^\intercal \mat{y} =(\mat{G}_{(n)}^\intercal - \mat{G}_{(n)}^\intercal (\mat{G}_{(n)}^\intercal)^{+} \mat{G}_{(n)}^\intercal) \mat{y} = 0.
\] 
More generally, a vector $\mat{z}$ is in the image of $\mat{G}_{(n)}^\intercal$
if and only if $(\mat{I} - \mat{G}_{(n)}^\intercal (\mat{G}_{(n)}^\intercal)^{+}) \mat{z} = \mat{0}$.
This is an alternate formulation of~\Cref{lemma:constrained_least_squares}
and leads to the algorithm~\FastFactorMatrixUpdate below.

\subsection{Fast Factor Matrix Update Algorithm}
\label{subsec:fast_factor_matrix_update}

\begin{algorithm}
\caption{\FastFactorMatrixUpdate}
\label{alg:fast_factor_matrix_update}
  \textbf{Input:}
  Tensor $\tensor{X} \in \R^{I_1 \times \cdots \times I_N}$,
  factors $\mat{A}^{(n)} \in \R^{I_n \times R_n}$,
  core $\tensor{G} \in \R^{R_1 \times \cdots \times R_N}$,
  index $n$,
  $\lambda$, error $\varepsilon$, probability $\delta$

\begin{algorithmic}[1]
  \State Set $R_{\ne n} \gets R_1 \cdots R_{n-1} R_{n+1} \cdots R_N$
  \State $\mat{K} = \mat{A}^{(1)} \ktimes \cdots \ktimes \mat{A}^{(n-1)} \ktimes \mat{A}^{(n+1)} \ktimes \cdots \ktimes \mat{A}^{(N)}$
  \State Set $\mat{B} \gets \mat{X}_{(n)}$
  \State Initialize product distribution data structure
         $\cP$ to sample indices from
         $(\bm{\ell}(\mat{A}^{(1)}), \cdots, \bm{\ell}(\mat{A}^{(N)}))$
  \State Let $\mat{N} = \mat{I}_{R_{\ne n}} - \mat{G}_{(n)}^\intercal 
          {(\mat{G}_{(n)}^\intercal)}^+$
  \State Let
  $w \ge \parens*{1 + \frac{12}{\varepsilon}}
  \left\Vert \begin{bmatrix}
  \mat{K} \\ \sqrt{\lambda} (\mat{G}_{(n)}^\intercal)^+
\end{bmatrix} \mat{N}^+ \right\Vert_2^2$ as in \cref{lemma:constraint-reg}, and construct the operator
  \[
    \mat{M}^+ =
    \parens*{
    \begin{bmatrix}
      \mat{K} \\
      \sqrt{\lambda} (\mat{G}_{(n)}^\intercal)^+ \\
      \sqrt{w} \mat{N}
    \end{bmatrix}^\intercal
    \begin{bmatrix}
      \mat{K} \\
      \sqrt{\lambda} (\mat{G}_{(n)}^\intercal)^+ \\
      \sqrt{w} \mat{N}
    \end{bmatrix}
    }^+
  \]
  using the Woodbury identity in Equation~\Cref{eqn:woodbury_expression}
  \State Set $s \gets \ceil{1680 R_{\ne n} \ln(40R_{\ne n}) \ln(I_n / \delta) / \varepsilon}$
  \For{$i=1$ to $I_n$}
    \State Set $\mat{S} \gets \RowSampling(\mat{K}, s, \cP)$
    \State Set $\mat{\tilde K} \gets \mat{S}\mat{K}$ and $\mat{\tilde b} \gets \mat{S}\mat{b}_{i:}^\intercal$
    \State Initialize $\mat{z} \gets \mat{0}_{R_{\ne n}}$
    \Repeat
      \State Update $\mat{z}
          \gets \mat{z} - (1-\sqrt{\varepsilon})\mat{M}^+\parens{(\mat{\tilde K}^\intercal \mat{\tilde K} + w\mat{I})\mat{z} + \lambda \mat{G}_{(n)}^+ {(\mat{G}_{(n)}^\intercal)}^+ \mat{z}
          - w \mat{G}_{(n)}^+ \mat{G}_{(n)}\mat{z} - \mat{\tilde K}^\intercal\mat{\tilde b}}$ using fast Kronecker-matrix multiplication \label{line:richardson_step_2}
    \Until{convergence}
  \State Update factor matrix row $\mat{a}^{(n)}_{i:} \gets \mat{z}^\intercal \mat{G}_{(n)}^+$
  \EndFor
\end{algorithmic}
\end{algorithm}

\TheoremFastFactorMatrixUpdate*

\begin{proof}
Each factor row matrix update in \texttt{TuckerALS} (\Cref{alg:alternating_least_squares})
has the form
\[
  \mat{a}_{i:}^{(n)} \gets \argmin_{\mat{y} \in \R^{R_n}} \parens*{
    \norm{\mat{K}\mat{G}_{(n)}^\intercal\mat{y} - \mat{b}_{i:}^\intercal}_2^2 + \lambda\norm{\mat{y}}_2^2}.
\]
Use \Cref{lemma:constrained_least_squares} to reduce
the factor matrix updates to solving the equality-constrained
Kronecker regression problem
\begin{align*}
  \mat{z}_{\opt}
  =
  \argmin_{\mat{N}\mat{z}=\mat{0}}
  \parens*{
   \norm{\mat{K} \mat{z} - {\mat{b}_{i:}}^\intercal}_{2}^2
  +
  \lambda \norm{{(\mat{G}_{(n)}^\intercal)}^+ \mat{z}}_{2}^2
  }.
\end{align*}
The correctness of this algorithm is analogous to the argument in
the proof of \Cref{thm:fast_kronecker_regression}, but now we 
have more sophisticated blocks in the data matrix and
need to account for them.

We solve each row update independently.
The construction of the sketched submatrix $\mat{\tilde K}$ 
guarantees that $(1-\sqrt{\varepsilon})^{-1}\mat{M}$
is a 3-spectral approximation to the sketched normal matrix
\[
  \mat{\tilde M} \defeq \mat{\tilde K}^\intercal \mat{\tilde K} + \lambda\mat{G}^+_{(n)}{(\mat{G}_{(n)}^\intercal)}^+ + w\mat{N}^\intercal \mat{N},
\]
with probability at least $1 - \delta/I_n$,
by \Cref{lemma:block_sketch_is_spectral_approx}.
Thus, we can use the (non-sketched) matrix $\mat{M}^+$
as a preconditioner and exploit its Kronecker structure
since this iterative method converges in $\tilde{O}(1)$ steps
by~\Cref{lemma:richardson_iteration}.

It remains to show the main difference with~\Cref{thm:fast_kronecker_regression}:
the time complexity of one Richardson iteration (Line~\ref{line:richardson_step_2}).
We show in \Cref{lemma:fast_factor_richardson_step} how
$\mat{M}^+ \mat{x}$ can be computed in time
\[
\tilde{O}\parens{R_{\ne n}^2 \varepsilon^{-1}\log(I_n/\delta) + R \sum_{k=1}^N R_{k} + R_{n}^\omega}
\]
using the Woodbury matrix identity since $\mat{M}^+$ is a rank-$R_n$ update
to $(\mat{K}^\intercal \mat{K} + w\mat{I})^+$.

The solution of each sketch is a 
$(1+\varepsilon)$-approximation to the optimal factor row by
\Cref{lemma:approximate_block_regression}, and
the success guarantee follows from a union bound over all $I_n$ rows.
\end{proof}

\begin{lemma}
\label{lemma:fast_factor_richardson_step}
\textbf{Line~13} in \FastFactorMatrixUpdate
takes
$\tilde{O}\parens{R_{\ne n}^2 \varepsilon^{-1}\log(I_n/\delta) + R \sum_{k=1}^N R_{k}}$ time
after preprocessing.
\end{lemma}

\begin{proof}
Recall that $\mat{N} = \mat{I} - \mat{G}_{(n)}^\intercal (\mat{G}_{(n)}^\intercal)^+$
and consider the following equality.
Letting
\begin{align*}
  \mat{M} \defeq
  \begin{bmatrix}
\mat{K} \\
    \sqrt{\lambda} \parens{\mat{G}_{(n)}^\intercal}^+ \\
    \sqrt{w} \mat{N}
  \end{bmatrix}^\intercal 
  \begin{bmatrix}
    \mat{K} \\
    \sqrt{\lambda} \parens{\mat{G}_{(n)}^\intercal}^+ \\
    \sqrt{w} \mat{N}
  \end{bmatrix}, 
\end{align*}
we have
\begin{align*}
  \mat{M}
  &= 
  \mat{K}^\intercal \mat{K} + \lambda \mat{G}_{(n)}^+ (\mat{G}_{(n)}^\intercal)^+ + w \mat{N}^\intercal \mat{N} \\
  &=
  \mat{K}^\intercal \mat{K} + \lambda \mat{G}_{(n)}^+ (\mat{G}_{(n)}^\intercal)^+ + w (\mat{I} - \mat{G}_{(n)}^\intercal (\mat{G}_{(n)}^\intercal)^+)^\intercal (\mat{I} - \mat{G}_{(n)}^\intercal (\mat{G}_{(n)}^\intercal)^+) \\
  &=
  \mat{K}^\intercal \mat{K} + \lambda \mat{G}_{(n)}^+ (\mat{G}_{(n)}^\intercal)^+ + w(\mat{I} - \mat{G}_{(n)}^+ \mat{G}_{(n)} - \mat{G}_{(n)}^\intercal (\mat{G}_{(n)}^\intercal)^+ + \mat{G}_{(n)}^+ \mat{G}_{(n)} \mat{G}_{(n)}^\intercal (\mat{G}_{(n)}^\intercal)^+) \\
  &=
  (\mat{K}^\intercal \mat{K} + w \mat{I}) + \lambda \mat{G}_{(n)}^+ (\mat{G}_{(n)}^\intercal)^+ \\ & - w(\mat{G}_{(n)}^+ \mat{G}_{(n)} + \mat{G}_{(n)}^\intercal (\mat{G}_{(n)}^\intercal)^+ - \mat{G}_{(n)}^+ \mat{G}_{(n)} \mat{G}_{(n)}^\intercal (\mat{G}_{(n)}^\intercal)^+).
\end{align*}
For any matrix, we have the pseudoinverse identity
$\mat{A}^+ \mat{A} \mat{A}^\intercal = \mat{A}^\intercal$,
so it follows that
\[
  \mat{G}_{(n)}^+ \mat{G}_{(n)} \mat{G}_{(n)}^\intercal (\mat{G}_{(n)}^\intercal)^+
  =
  \mat{G}_{(n)}^\intercal (\mat{G}_{(n)}^\intercal)^+.
\]
Therefore,
\begin{align*}
  \mat{M}
  &=
(\mat{K}^\intercal \mat{K} + w \mat{I}) + \lambda \mat{G}_{(n)}^+ (\mat{G}_{(n)}^\intercal)^+ - w(\mat{G}_{(n)}^+ \mat{G}_{(n)} + \mat{G}_{(n)}^\intercal (\mat{G}_{(n)}^\intercal)^+ - \mat{G}_{(n)}^\intercal (\mat{G}_{(n)}^\intercal)^+) \\
  &=
(\mat{K}^\intercal \mat{K} + w \mat{I}) + \lambda \mat{G}_{(n)}^+ (\mat{G}_{(n)}^\intercal)^+ - w(\mat{G}_{(n)}^+ \mat{G}_{(n)}) \\
  &=
  (\mat{K}^\intercal \mat{K} + w \mat{I}) + \mat{G}_{(n)}^+ (\lambda (\mat{G}_{(n)}^\intercal)^+ - w \mat{G}_{(n)}).
\end{align*}
Applying the Woodbury matrix identity, we have
  \begin{align}
    \mat{M}^+
    &=
    (\mat{K}^\intercal \mat{K} + w \mat{I})^{-1} 
    \label{eqn:woodbury_expression} \\ \notag
    & + 
(\mat{K}^\intercal \mat{K} + w \mat{I})^{-1} \mat{G}_{(n)}^+ (\mat{I} + (\lambda (\mat{G}_{(n)}^\intercal)^+ - w \mat{G}_{(n)}) (\mat{K}^\intercal \mat{K} + w \mat{I})^{-1} \mat{G}_{(n)}^+)^{-1} \\ & \cdot (\lambda (\mat{G}_{(n)}^\intercal)^+ - w \mat{G}_{(n)}) (\mat{K}^\intercal \mat{K} + w \mat{I})^{-1}. \notag
\end{align}
First note that
\begin{align*}
  \mat{I} + (\lambda (\mat{G}_{(n)}^\intercal)^+ - w \mat{G}_{(n)}) (\mat{K}^\intercal \mat{K} + w \mat{I})^{-1} \mat{G}_{(n)}^+
  \in \R^{R_n \times R_n}.
\end{align*}
The time complexity of computing the factored SVD of $\mat{K}^\intercal \mat{K}$
is
$O(\sum_{k \ne n} (I_{k} R_k^2 + R_k^\omega))$.
In \texttt{TuckerALS}, these are computed at the end of each factor matrix
update and therefore do not need to be computed in this step.
After this,
multiplying a vector by
$(\mat{K}^\intercal \mat{K} + w \mat{I})^{-1}$
can be done in time $O(R_{\ne n} \sum_{m\ne n} R_{m})$
by \Cref{lemma:kron_mat_mul}.
Therefore, computing
\[
  \mat{I} + (\lambda (\mat{G}_{(n)}^\intercal)^+ - w \mat{G}_{(n)}) (\mat{K}^\intercal \mat{K} + w \mat{I})^{-1} \mat{G}_{(n)}^+
\]
takes $O(R \sum_{k=1}^N R_k)$ time.
Computing the inverse of this matrix takes $O(R_n^\omega)$ time.
Moreover, this inverse can be used for all Richardson iteration steps and
row updates.
Finally, observe that multiply any vector with 
$\mat{G}_{(n)}^+$ or
$(\lambda (\mat{G}_{(n)}^\intercal)^+ - w \mat{G}_{(n)})$ takes $O(R)$ time.
Therefore, to evaluate $\mat{M}^+ \mat{z}$ for any $\mat{z}$, 
we use Equation~\Cref{eqn:woodbury_expression}
and repeatedly apply matrix-vector multiplcations from right to left.
The total running time per evaluation after preprocessing is
\[
  O\parens*{R \sum_{k=1}^N R_{k}}.
\]

Now we show that the vector
\[
  \mat{\tilde K}^\intercal \mat{\tilde K} \mat{z} + \lambda \mat{G}_{(n)}^+ (\mat{G}_{(n)}^\intercal)^+ \mat{z} + w \mat{N}^\intercal \mat{N} \mat{z}
  - \mat{\tilde K}^\intercal \mat{\tilde b}
\]
can be computed fast enough.
By the same argument above, this is equivalent to
\[
(\mat{\tilde K}^\intercal \mat{\tilde K} + w \mat{I})\mat{z}
+ \lambda \mat{G}_{(n)}^+ (\mat{G}_{(n)}^\intercal)^+ \mat{z} - w\mat{G}_{(n)}^+ \mat{G}_{(n)} \mat{z}
- \mat{\tilde K}^\intercal \mat{\tilde b}.
\]
We can compute $\mat{\tilde K}^\intercal \mat{\tilde b}$ in
$\tilde{O}(R_{\ne n}^2 \varepsilon^{-1} \log(I_n /\delta))$ time,
$\lambda \mat{G}_{(n)}^+ (\mat{G}_{(n)}^\intercal)^+ \mat{z}$
and $w\mat{G}_{(n)}^+ \mat{G}_{(n)} \mat{z}$
take $O(R)$ time,
and
$(\mat{\tilde K}^\intercal \mat{\tilde K} + w \mat{I})\mat{z}$ takes
$\tilde{O}(R + R_{\ne n}^2 \varepsilon^{-1} \log(I_n /\delta))$ time.
Summing all of these running times completes the proof.
\end{proof}

\subsection{Comparison of different factor matrix and core tensor update
algorithms.}
\label{subsec:different_algorithms_running_times}

In this subsection, we prove the running times for the different algorithms in~\Cref{table:algorithm_comparison}.

\paragraph{Naive factor matrix update.}

Consider the $n$-th factor matrix $\mat{A}^{(n)} \in \R^{I_n \times R_n}$.
First let
\[
  \mat{K} = \mat{G}_{(n)} \parens{
    \mat{A}^{(1)} \otimes \dots \otimes \mat{A}^{(n-1)} \otimes
    \mat{A}^{(n+1)} \otimes \dots \otimes \mat{A}^{(N)}}^\intercal
\]
and
\[
  \mat{B} = \mat{X}_{(n)}.
\]
For $i=1$ to $I_n$ we want to solve:
\[
  \mat{a}^{(n)}_{i:} \gets
          \argmin_{\mat{y} \in \R^{1 \times R_n}}
              \norm{\mat{y} \mat{K} - \mat{b}_{i:}}_{2}^2
          + \lambda \norm{\mat{y}}_{2}^2
\]
We use the normal equation for ridge regression:
\begin{align*}
  \mat{y}_{\opt}^\intercal &=
  \parens*{\mat{K} \mat{K}^\intercal + \lambda\mat{I}}^{+}
    \mat{K} \mat{b}_{i:}^\intercal.
\end{align*}
Let's start by computing $\mat{K}\mat{K}^\intercal + \lambda \mat{I}$ efficiently:
\begin{align}
\nonumber
  & \mat{K}\mat{K}^\intercal \\ \nonumber & =
    \mat{G}_{(n)} \parens{
    \mat{A}^{(1)} \otimes \dots \otimes \mat{A}^{(n-1)} \otimes
    \mat{A}^{(n+1)} \otimes \dots \otimes \mat{A}^{(N)}}^\intercal
    \\ & \hspace{0.5cm} \cdot
    \parens{
    \mat{A}^{(1)} \otimes \dots \otimes \mat{A}^{(n-1)} \otimes
    \mat{A}^{(n+1)} \otimes \dots \otimes \mat{A}^{(N)}}
    \mat{G}_{(n)}^\intercal \notag \\
    & =
    \mat{G}_{(n)}
    \parens*{
      {\mat{A}^{(1)}}^\intercal \mat{A}^{(1)}
      \otimes
      \dots
      \otimes
      {\mat{A}^{(n-1)}}^\intercal \mat{A}^{(n-1)}
      \otimes
      {\mat{A}^{(n+1)}}^\intercal \mat{A}^{(n+1)}
      \otimes
      \dots
      \otimes
      {\mat{A}^{(N)}}^\intercal \mat{A}^{(N)}
    }
    \mat{G}_{(n)}^\intercal. \label{eqn:factor_update_gram_matrix}
\end{align}
The internal terms of the form
${\mat{A}^{(k)}}^\intercal \mat{A}^{(k)}$ can be computed in
$O(R_k^2 I_k)$ time using $O(R_k^2)$ space.
Therefore, we can compute
\[
{\mat{A}^{(1)}}^\intercal \mat{A}^{(1)}
      \otimes
      \dots
      \otimes
      {\mat{A}^{(n-1)}}^\intercal \mat{A}^{(n-1)}
      \otimes
      {\mat{A}^{(n+1)}}^\intercal \mat{A}^{(n+1)}
      \otimes
      \dots
      \otimes
      {\mat{A}^{(N)}}^\intercal \mat{A}^{(N)}
\]
in $O(R_{\ne n}^2 N)$ time and $O(R_{\ne n}^2)$ space,
where $R_{\ne n} = \prod_{k \ne n} R_{k}$.
Further, we can compute $\mat{K}\mat{K}^\intercal \in \R^{R_n \times R_n}$
in time
\[
  O(R_{\ne n}^2 R_n + R_n^2 R_{\ne n})
\]
by multiplying the two rightmost matrices, and then the remaining two.
Therefore, we can compute $(\mat{K}\mat{K}^\intercal + \lambda\mat{I})^+$
in $O(R_n^\omega)$ time.

Now we need to work on the term $\mat{K} \mat{b}_{i:}^\intercal \in \R^{r_n}$.
We have the following:
\begin{align}
  \label{eqn:factor_matrix_second_kronmatmul}
  \mat{K} \mat{b}_{i:}^\intercal
  &=
  \mat{G}_{(n)} \parens{
    \mat{A}^{(1)} \otimes \dots \otimes \mat{A}^{(n-1)} \otimes
    \mat{A}^{(n+1)} \otimes \dots \otimes \mat{A}^{(N)}}^\intercal
  \mat{b}_{i:}^\intercal.
\end{align}
We need to compute the internal Kronecker product in
$O(I_{\ne n} R_{\ne n} N)$ time and $O(I_{\ne n} R_{\ne n})$ space.
The vector multiplication with $\mat{b}_{i:}^\intercal$ takes
$O(R_{\ne n} I_{\ne n})$ time (and this happens $I_n$ times).
These can actually be combined so that the total max memory consumed is $O(R_{\ne n})$
by computing each entry of the matrix-vector product individually.
Then the final matrix multiplication takes
$O(R_n R_{\ne n})$ time.

Putting everything together, we can compute 
$(\mat{K}\mat{K}^\intercal + \lambda\mat{I})^+ \mat{K} \mat{b}_{i:}^\intercal$
with a final matrix-vector multiplication in $O(R_n^2)$ time.
Thus, the overall running time of this factor matrix is:
\[
  O\parens*{\parens*{\sum_{k \ne n}R_{k}^2 I_k} + R_{\ne n}^2 N + R_{\ne n}^2 R_n + R_n^2 R_{\ne n}
  + R_n^\omega + I_{\ne n} R_{\ne n} N + R_{\ne n} I}
\]
and the space complexity is
$
  O(R_{\ne n}^2 + R_{\ne n}) = O(R_{\ne n}^2).
$

\paragraph{Naive core tensor update.}

We use the normal equation to solve:
\[
    \tensor{G} \gets \argmin_{\tensor{G}'}
             \norm{\parens{\mat{A}^{(1)} \ktimes \mat{A}^{(2)} \ktimes \cdots \ktimes \mat{A}^{(N)}} \vectorize(\tensor{G}') - \vectorize(\tensor{X})}_{2}^2
             + \lambda \norm{\vectorize{\parens{\tensor{G}'}}}_{2}^2.
\]
Let
$\mat{K} = \mat{A}^{(1)} \ktimes \mat{A}^{(2)} \ktimes \cdots \ktimes \mat{A}^{(N)}$.
We need to efficiently compute
\[
  \mat{g}_{\opt} = (\mat{K}^\intercal \mat{K} + \lambda \mat{I})^+ \mat{K}^\intercal \vectorize(\tensor{X}).
\]
Observe that
\begin{align*}
  \mat{K}^\intercal \mat{K}
  &=
  {\mat{A}^{(1)}}^\intercal \mat{A}^{(1)}
  \otimes
  {\mat{A}^{(2)}}^\intercal \mat{A}^{(2)}
  \otimes
  \dots
  \otimes
  {\mat{A}^{(N)}}^\intercal \mat{A}^{(N)}.
\end{align*}
Thus, we can compute $\mat{K}^\intercal\mat{K}$
in time $O((\sum_{n=1}^N R_{n}^2 I_n) + R^2 N)$ where $R = \prod_{n=1}^N R_n$.
It follows that we can then compute
$(\mat{K}^\intercal \mat{K} + \lambda \mat{I})^+$ in time $O(R^\omega)$
and $O(R^2)$ space.

Now we need to compute $\mat{K}^\intercal \mat{x} \in \R^{R}$.
Computing $\mat{K}^\intercal$ explicitly requires
$O(IR)$ space and $O(IRN)$ time.
Computing the matrix-vector product one entry at a time requires
$O(IRN)$ time and only $O(R)$ space.
Putting everything together, the overall running time is:
\[
  O\parens*{\parens*{\sum_{n=1}^N R_{n}^2 I_n} + R^2 N + R^\omega + IRN}.
\]
The space complexity is $O(R^2)$.

\paragraph{Factor matrix update with \texttt{KronMatMul}.}
We use \texttt{KronMatMul} (\Cref{lemma:kron_mat_mul})
in \Cref{eqn:factor_update_gram_matrix} as
follows:
\[
  \texttt{KronMatMul}\parens*{[
    {\mat{A}^{(1)}}^\intercal \mat{A}^{(1)},
    \dots,
    {\mat{A}^{(n-1)}}^\intercal \mat{A}^{(n-1)},
    {\mat{A}^{(n+1)}}^\intercal \mat{A}^{(n+1)},
    \dots
    {\mat{A}^{(N)}}^\intercal \mat{A}^{(N)}
  ], \mat{G}_{(n)}^\intercal
  }.
\]
The resulting matrix has size $R_{\ne n} \times R_n$.
The running time for this step is
$O(R_{\ne n} R_n \sum_{k \ne n} R_k) = O(R \sum_{k \ne n} R_k)$
and the space complexity is $O(R)$.
Previously this step was $O(R_{\ne n}^2N + R_{\ne n}^2 R_n)$.

Now for the expensive matrix-vector product:
\[
  \texttt{KronMatMul}\parens*{
    [
      {\mat{A}^{(1)}}^\intercal, \dots, {\mat{A}^{(n-1)}}^\intercal,
    {\mat{A}^{(n+1)}}^\intercal, \dots, {\mat{A}^{(N)}}^\intercal],
  \mat{b}_{i:}^\intercal
  }
\]
The result is a vector of size $R_{\ne n}$.
The running time of this subroutine is $O(I_{\ne n} \sum_{k \ne n} R_k)$,
and it uses $O(I_{\ne n})$ space.
It follows that we can compute $\mat{K}\mat{b}_{i:}^\intercal$ by multiplying
this vector with $\mat{G}_{(n)}$,
which takes $O(R_n R_{\ne n}) = O(R)$ time.
Previously this took $O(I_{\ne n} R_{\ne n}N + R_{\ne n}I_{\ne n} + R)$ time.

Thus, the updated time complexity is
\begin{align*}
  O\parens*{R\parens*{\sum_{k\neq n} R_n}
  + R_n^\omega + I_{\ne n} \parens*{\sum_{k \ne n} R_k} + R_{n}^2 I_n}.
\end{align*}

\paragraph{Core update with \texttt{KronMatMul}.}
The running time for $\mat{K}^\intercal \mat{x} \in \R^{I_1 \cdots I_N}$ is
$O((R_1 + \dots + R_N) I)$ and space is $O(I)$.
Therefore, the overall running time is
$
  O\parens{R^2 N + R^\omega + I \sum_{n=1}^N R_n}.
$
The overall space complexity is $O(R^2)$.
\section{Supplementary Material for Experiments}
\label{app:experiments}

\subsection{Kronecker Regression}
\label{app:kronecker_regression_experiments}

Here we provide the full experimental results for the Kronecker regression
numerical experiments in~\Cref{sec:experiments}.
In the tables below,
\texttt{Baseline} is the naive baseline algorithm that computes the
normal equation exactly (i.e., fully constructs $\mat{K}$ in the $\mat{K}^\intercal \mat{b}$ term),
\texttt{KronMatMul} is an exact baseline that uses the fast Kronecker-vector product in \Cref{lemma:kron_mat_mul},
\texttt{DJSSW19} is the sketching algorithm of~\citet{diao2019optimal},
and \Cref{alg:fast_kronecker_regression} is our \FastKroneckerRegression algorithm.

We note that the number of sketches used by both sketching algorithms is reduced by an
$\alpha = 10^{-5}$ factor so that the number of total row samples is of the same
order of magnitude as the fixed number of samples in the experiments of~\cite{diao2018sketching,diao2019optimal}. Our choice of $\lambda = 10^{-3}$ for the
L2 regularization strength does not affect the running times
or significantly impact the relative errors.

\paragraph{Running times.}
The running times of all the algorithms are presented in~\Cref{tab:appendix_runtimes}.
We denote entries where the algorithm ran out of memory by a dash.

\begin{table}[h]
    \caption{Running times of Kronecker regression algorithms with an $n^2 \times d^2$ design matrix (seconds).}
    \label{tab:appendix_runtimes}
    \centering
    \begin{tabular}{rrrrrrrrrrr}
        \toprule
        $d$ & $n$ & \texttt{Baseline} & \texttt{KronMatMul} & \texttt{DJSSW19} & \Cref{alg:fast_kronecker_regression}  &  \\
        \midrule
8 & 128 & 0.0106 & 0.0008 & 0.0056 & 0.0074 \\
 & 256 & 0.0349 & 0.0014 & 0.0051 & 0.0155 \\
 & 512 & 0.1343 & 0.0037 & 0.0064 & 0.0172 \\
 & 1024 & 0.5009 & 0.0115 & 0.0087 & 0.0192 \\
 & 2048 & 1.9403 & 0.0466 & 0.014 & 0.0229 \\
 & 4096 & 7.5253 & 0.1855 & 0.0236 & 0.042 \\
 & 8192 & 29.9479 & 0.676 & 0.0404 & 0.0578 \\
 & 16384 & -- & 3.9717 & 0.0776 & 0.1019 \\
        \midrule
16 & 128 & 0.1416 & 0.001 & 0.1154 & 0.0942 \\
& 256 & 0.1992 & 0.0022 & 0.1212 & 0.0988 \\
& 512 & 0.5485 & 0.0062 & 0.1216 & 0.0927 \\
& 1024 & 1.8994 & 0.0222 & 0.1184 & 0.0982 \\
& 2048 & 7.1677 & 0.0916 & 0.1286 & 0.1112 \\
& 4096 & 27.5237 & 0.3399 & 0.1365 & 0.114 \\
& 8192 & -- & 2.2858 & 0.1599 & 0.1367 \\
& 16384 & -- & 15.0028 & 0.2036 & 0.1869 \\
\midrule
32 & 128 & 5.9593 & 0.0019 & 10.3478 & 1.6544 \\
& 256 & 7.0358 & 0.0039 & 10.4735 & 1.8021 \\
& 512 & 8.1757 & 0.0118 & 10.3156 & 1.7045 \\
& 1024 & 13.6273 & 0.0474 & 10.3379 & 1.6802 \\
& 2048 & 31.4996 & 0.1739 & 10.4088 & 1.7193 \\
& 4096 & -- & 1.129 & 10.2585 & 1.6732 \\
& 8192 & -- & 7.7335 & 10.4856 & 1.7505 \\
& 16384 & -- & 29.7196 & 10.5815 & 1.8613 \\
\midrule
64 & 128 & 966.1175 & 0.0058 & 1240.6471 & 29.6798 \\
& 256 & 982.5172 & 0.0102 & 1243.1648 & 29.6453 \\
& 512 & 970.6558 & 0.0262 & 1240.7365 & 29.3725 \\
& 1024 & 985.4097 & 0.09 & 1239.8602 & 29.5546 \\
& 2048 & -- & 0.6687 & 1238.6558 & 29.4367 \\
& 4096 & -- & 3.6422 & 1243.2635 & 29.7029 \\
& 8192 & -- & 14.5137 & 1244.8347 & 29.8926 \\
& 16384 & -- & 67.0037 & 1241.5569 & 29.9143 \\
        \bottomrule
    \end{tabular}
\end{table}

\paragraph{Losses.}
We list the losses of all the Kronecker regression algorithms in~\Cref{tab:appendex_losses}.
Entries where the algorithm ran out of memory by a dash are denoted by a dash.
We note that \texttt{DJSSW19} begins to have numerical stability problems for $n \ge 2048$
for $d \in \{16, 32, 64\}$, though it is solving the same sketched ridge regression problem
as~\Cref{alg:fast_kronecker_regression}.

\begin{table}[h]
    \caption{Losses of the Kronecker regression algorithms with a random $n^2 \times d^2$ design matrix.}
    \label{tab:appendex_losses}
    \centering
    \begin{tabular}{rrrrrrrrrrr}
        \toprule
        $d$ & $n$ & \texttt{Baseline} & \texttt{KronMatMul} & \texttt{DJSSW19} & \Cref{alg:fast_kronecker_regression}  &  \\
        \midrule
8 & 128 & 0.0033 & 0.0033 & 0.0035 & 0.0035 \\
& 256 & 0.0151 & 0.0151 & 0.0161 & 0.0161 \\
& 512 & 0.0606 & 0.0606 & 0.0635 & 0.0635 \\
& 1024 & 0.2531 & 0.2531 & 0.2636 & 0.2636 \\
& 2048 & 1.032 & 1.032 & 1.0599 & 1.0598 \\
& 4096 & 4.2118 & 4.2118 & 4.4741 & 4.4525 \\
& 8192 & 16.9013 & 16.9009 & 17.6642 & 17.5533 \\
& 16384 & -- & 66.6618 & 96.9055 & 72.6006 \\
        \midrule
16 & 128 & 0.0019 & 0.0019 & 0.0021 & 0.0021 \\
& 256 & 0.0077 & 0.0077 & 0.0081 & 0.0081 \\
& 512 & 0.0313 & 0.0313 & 0.0324 & 0.0324 \\
& 1024 & 0.1306 & 0.1306 & 0.1328 & 0.1328 \\
& 2048 & 0.5244 & 0.5244 & 0.5457 & 0.5375 \\
& 4096 & 2.0917 & 2.0915 & 2.2278 & 2.1381 \\
& 8192 & -- & 8.2601 & 9.1014 & 8.5251 \\
& 16384 & -- & 33.3611 & 42.8783 & 34.4764 \\
\midrule
32 & 128 & 0.0007 & 0.0007 & 0.0009 & 0.0009 \\
& 256 & 0.0037 & 0.0037 & 0.004 & 0.004 \\
& 512 & 0.0163 & 0.0163 & 0.0172 & 0.0172 \\
& 1024 & 0.0671 & 0.0671 & 0.0693 & 0.0692 \\
& 2048 & 0.2661 & 0.266 & 0.388 & 0.2701 \\
& 4096 & -- & 1.0434 & 2.9236 & 1.0589 \\
& 8192 & -- & 4.1777 & 4.7172 & 4.2537 \\
& 16384 & -- & 16.874 & 24149.8906 & 17.4612 \\
\midrule
64 & 128 & 0.0002 & 0.0002 & 0.0004 & 0.0004 \\
& 256 & 0.0015 & 0.0015 & 0.0019 & 0.0019 \\
& 512 & 0.0078 & 0.0078 & 0.0087 & 0.0087 \\
& 1024 & 0.0307 & 0.0307 & 0.035 & 0.0323 \\
& 2048 & -- & 0.1233 & 1.5768 & 0.1264 \\
& 4096 & -- & 0.5068 & 275.5661 & 0.5198 \\
& 8192 & -- & 2.0735 & 333.4299 & 2.1358 \\
& 16384 & -- & 8.2375 & 546391.7285 & 8.6082 \\
        \bottomrule
    \end{tabular}
\end{table}

\subsection{Low-rank Tucker Decomposition of Image Tensors}
\label{app:tensor_decomposition_experiments}

Here we compare different Kronecker regression algorithms
in the core update of the alternating least squares (ALS) algorithm for Tucker decompositions.
For the sketching-based algorithms, we increase the number of row samples to study
how this affects the quality of the tensor decomposition.
We record the quality of the tensor decomposition using the relative reconstruction
error
$\norm{\widehat{\tensor{X}} - \tensor{X}}^2_{\frobenius} / \norm{\tensor{X}}^2_{\frobenius}$.
The number of row samples used is $m \in \{1024, 4096, 16384\}$,
as in the experiments of~\cite{diao2018sketching, diao2019optimal}.

We compare against higher-order orthogonal iteration (HOOI) and ALS as baseline algorithms.
We use the Tensorly~\cite{kossaifi2019tensorly} implementation of HOOI,
which is an industry standard.
We do not use L2 regularization so that we can compare against HOOI.
We compare our Kronecker regression algorithm with \cite[Algorithm 1]{diao2019optimal},
denoted by \texttt{DJSSW19}.
The running times reported are the mean iteration times, where an iteration
includes all factor matrix updates and the core tensor update.
Trials that ran out of memory or failed to converge are denoted by a dash.
All algorithms are run for five iterations.

\paragraph{Cardiac MRI.}
This dataset is a $256 \times 256 \times 14 \times 20$ tensor whose elements are
MRI measurements indexed by $(x,y,z,t)$
where $(x,y,z)$ is a point in space and $t$ corresponds to time.

\begin{table}[H]
    \caption{Relative reconstruction errors for cardiac MRI tensor with different multilinear ranks.}
    \label{tab:mri_rre}
    \centering
    \begin{tabular}{cccccccccccccccc}
        \toprule
        & & & \multicolumn{3}{c}{\FastKroneckerRegression} & \multicolumn{3}{c}{\texttt{DJSSW19}} \\
        rank & HOOI & ALS & 1024 & 4096 & 16384 & 1024 & 4096 & 16384 \\
        \midrule
1,1,1,1 & 0.648 & 0.648 & 0.649 & 0.648 & 0.648 & 0.648 & 0.648 & 0.648 \\
4,2,2,1 & 0.569 & 0.570 & 0.574 & 0.571 & 0.570 & 0.573 & 0.571 & 0.570 \\
4,4,2,2 & 0.511 & 0.511 & 0.533 & 0.514 & 0.512 & 0.561 & 0.515 & 0.513 \\
8,2,2,1 & 0.569 & 0.577 & 0.584 & 0.579 & 0.577 & 0.589 & 0.581 & 0.577 \\
8,4,4,1 & 0.448 & 0.452 & 0.491 & 0.459 & 0.453 & 0.488 & 0.458 & 0.454 \\
8,4,4,2 & 0.448 & 0.451 & 0.492 & 0.475 & 0.455 & 0.585 & 0.466 & 0.455 \\
8,8,2,2 & 0.465 & 0.467 & 0.498 & 0.485 & 0.471 & 0.556 & 0.480 & 0.470 \\
8,8,4,4 & 0.350 & 0.351 & -- & 0.371 & 0.356 & 0.679 & 0.476 & 0.409 \\
        \bottomrule
    \end{tabular}
\end{table}

\begin{table}[H]
    \caption{Average iteration time of ALS with sketching-based Kronecker regression
    for cardiac MRI tensor with different multilinear ranks (seconds).}
    \label{tab:mri_time}
    \centering
    \begin{tabular}{cccccccccccccccc}
        \toprule
        & & & \multicolumn{3}{c}{\FastKroneckerRegression} & \multicolumn{3}{c}{\texttt{DJSSW19}} \\
        rank & HOOI & ALS & 1024 & 4096 & 16384 & 1024 & 4096 & 16384 \\
        \midrule
1,1,1,1 & 1.187 & 1.307 & 1.328 & 1.334 & 1.321 & 1.397 & 1.313 & 1.346 \\
4,2,2,1 & 1.429 & 1.368 & 1.345 & 1.326 & 1.349 & 1.395 & 1.350 & 1.383 \\
4,4,2,2 & 1.458 & 1.463 & 1.539 & 1.511 & 1.536 & 1.413 & 1.497 & 1.719 \\
8,2,2,1 & 2.401 & 1.339 & 1.421 & 1.415 & 1.347 & 1.368 & 1.300 & 1.425 \\
8,4,4,1 & 1.664 & 1.435 & 1.575 & 1.562 & 1.676 & 1.573 & 1.693 & 2.737 \\
8,4,4,2 & 1.745 & 1.614 & 1.782 & 1.754 & 2.137 & 1.820 & 2.667 & 6.496 \\
8,8,2,2 & 1.741 & 1.466 & 1.621 & 1.810 & 2.079 & 1.751 & 2.525 & 6.133 \\
8,8,4,4 & 1.784 & 1.835 & 2.131 & 2.745 & 5.210 & 9.199 & 35.977 & 128.538 \\
        \bottomrule
    \end{tabular}
\end{table}

We also investigate how sensitive the convergence rate of sketching-based ALS is to the choice of the error parameter $\varepsilon$.
First, we reduce the number of sampled rows by $\alpha=0.001$ to compensate
for the large constant coefficient in Line~8 in \Cref{alg:fast_kronecker_regression};
otherwise, we do not see any quality degradation even for $\varepsilon = 0.99$.
Then in \Cref{tab:mri_convergence_1} and \Cref{tab:mri_convergence_2}, we compare the RRE
at each step of ALS (without sampling) and when using \FastKroneckerRegression as a subroutine
for decreasing values of $\varepsilon$.

\begin{table}[H]
    \caption{Relative reconstruction errors for cardiac MRI tensor with multilinear rank $(4,4,2,2)$
    during ALS with and without using \FastKroneckerRegression as a subroutine.}
    \label{tab:mri_convergence_1}
    \centering
    \begin{tabular}{cccccccccccccccc}
        \toprule
        & & \multicolumn{5}{c}{\FastKroneckerRegression}\\
        Step & ALS & $\varepsilon = 0.8$ & $\varepsilon = 0.4$ & $\varepsilon = 0.2$ & $\varepsilon = 0.1$ & $\varepsilon = 0.05$ \\
        \midrule
1 & 0.55883 & 0.56270 & 0.56044 & 0.55957 & 0.55942 & 0.55899 \\
2 & 0.51292 & 0.51609 & 0.51511 & 0.51443 & 0.51377 & 0.51316 \\
3 & 0.51096 & 0.51466 & 0.51206 & 0.51167 & 0.51139 & 0.51120 \\
4 & 0.51081 & 0.51338 & 0.51287 & 0.51171 & 0.51127 & 0.51102 \\
5 & 0.51079 & 0.51361 & 0.51286 & 0.51150 & 0.51126 & 0.51105 \\
        \bottomrule
    \end{tabular}
\end{table}

\begin{table}[H]
    \caption{Relative reconstruction errors for cardiac MRI tensor with multilinear rank $(8,8,4,4)$
    during ALS with and without using \FastKroneckerRegression as a subroutine.}
    \label{tab:mri_convergence_2}
    \centering
    \begin{tabular}{cccccccccccccccc}
        \toprule
        & & \multicolumn{5}{c}{\FastKroneckerRegression}\\
        Step & ALS & $\varepsilon = 0.8$ & $\varepsilon = 0.4$ & $\varepsilon = 0.2$ & $\varepsilon = 0.1$ & $\varepsilon = 0.05$ \\
        \midrule
1 & 0.44961 & 0.45165 & 0.45084 & 0.45021 & 0.44987 & 0.44975 \\
2 & 0.36573 & 0.36707 & 0.36650 & 0.36612 & 0.36609 & 0.36579 \\
3 & 0.35488 & 0.35549 & 0.35571 & 0.35508 & 0.35516 & 0.35504 \\
4 & 0.35162 & 0.35293 & 0.35238 & 0.35201 & 0.35184 & 0.35177 \\
5 & 0.35081 & 0.35193 & 0.35149 & 0.35124 & 0.35100 & 0.35092 \\
        \bottomrule
    \end{tabular}
\end{table}

\paragraph{Hyperspectral.}
This dataset is a $1024 \times 1344 \times 33$ tensor of
time-lapse hyperspectral radiance images
capturing a 1-hour interval of a nature scene undergoing illumination changes~\cite{nascimento2016spatial}.
These hyperspectral images and the COIL-100 dataset have both
been used recently as benchmark tasks for low-rank tensor decomposition~\cite{ma2021fast,battaglino2018practical,zhou2014decomposition}.

\begin{table}[H]
    \caption{Relative reconstruction errors for hyperspectral tensor with different multilinear ranks.}
    \label{tab:hyperspectral_rre}
    \centering
    \begin{tabular}{cccccccccccccccc}
        \toprule
        & & & \multicolumn{3}{c}{\FastKroneckerRegression} & \multicolumn{3}{c}{\texttt{DJSSW19}} \\
        rank & HOOI & ALS & 1024 & 4096 & 16384 & 1024 & 4096 & 16384 \\
        \midrule
1,1,1 & 0.271 & 0.271 & 0.271 & 0.271 & 0.271 & 0.271 & 0.271 & 0.271 \\
2,2,2 & 0.235 & 0.235 & 0.236 & 0.236 & 0.235 & 0.236 & 0.236 & 0.235 \\
4,4,4 & 0.203 & 0.208 & 0.213 & 0.211 & 0.211 & 0.213 & 0.208 & 0.208 \\
8,8,4 & 0.169 & 0.170 & 0.189 & 0.176 & 0.171 & 0.201 & 0.175 & 0.171 \\
8,8,8 & 0.169 & 0.169 & 0.213 & 0.177 & 0.171 & 0.261 & 0.180 & 0.171 \\
16,16,4 & 0.133 & 0.134 & -- & 0.155 & 0.139 & 0.465 & 0.156 & 0.139 \\
        \bottomrule
    \end{tabular}
\end{table}

\begin{table}[H]
        \caption{Average iteration time of ALS with sketching-based Kronecker regression
    for the hyperspcetral image tensor with different multilinear ranks (seconds).}
    \label{tab:hyperspectral_time}
    \centering
    \begin{tabular}{cccccccccccccccc}
        \toprule
        & & & \multicolumn{3}{c}{\FastKroneckerRegression} & \multicolumn{3}{c}{\texttt{DJSSW19}} \\
        rank & HOOI & ALS & 1024 & 4096 & 16384 & 1024 & 4096 & 16384 \\
        \midrule
1,1,1 & 1.873 & 2.377 & 2.222 & 2.241 & 2.312 & 2.224 & 2.241 & 2.203 \\
2,2,2 & 2.019 & 2.491 & 2.449 & 2.413 & 2.476 & 2.450 & 2.470 & 2.452 \\
4,4,4 & 2.255 & 2.965 & 2.798 & 2.838 & 2.922 & 2.747 & 2.902 & 3.022 \\
8,8,4 & 2.845 & 3.282 & 3.150 & 3.305 & 3.481 & 3.324 & 4.240 & 7.783 \\
8,8,8 & 2.888 & 4.043 & 3.700 & 4.168 & 5.682 & 5.559 & 11.774 & 34.878 \\
16,16,4 & 3.997 & 3.880 & 3.969 & 4.802 & 7.781 & 11.104 & 38.066 & 129.878 \\
        \bottomrule
    \end{tabular}
\end{table}

\paragraph{COIL-100.}
This dataset is a $7200 \times 120 \times 120 \times 3$ tensor
that contains 7200 colored images of 100 objects (72 images per object).
These objects have a wide variety of geometric characteristics and reflective properties.
To construct this dataset,
these objects were placed on a rotating table
and pictures were taken at pose intervals of 5 degrees,
hence 72 images per object~\cite{nene1996columbia}.

\begin{table}[H]
     \caption{Relative reconstruction errors for the COIL-100 tensor with different multilinear ranks.}
    \label{tab:coil_rre}
    \centering
    \begin{tabular}{cccccccccccccccc}
        \toprule
        & & & \multicolumn{3}{c}{\FastKroneckerRegression} & \multicolumn{3}{c}{\texttt{DJSSW19}} \\
        rank & HOOI & ALS & 1024 & 4096 & 16384 & 1024 & 4096 & 16384 \\
        \midrule
1,1,1,1 & 0.528 & 0.528 & 0.528 & 0.528 & 0.528 & 0.528 & 0.528 & 0.528 \\
4,2,2,1 & 0.460 & 0.460 & 0.463 & 0.461 & 0.460 & 0.464 & 0.461 & 0.461 \\
8,2,2,1 & 0.460 & 0.460 & 0.472 & 0.462 & 0.461 & 0.466 & 0.462 & 0.461 \\
8,4,4,1 & 0.414 & 0.414 & 0.447 & 0.421 & 0.416 & 0.443 & 0.420 & 0.416 \\
8,4,4,2 & 0.379 & 0.386 & 0.454 & 0.400 & 0.388 & 0.438 & 0.399 & 0.388 \\
16,4,4,2 & 0.349 & 0.349 & 0.499 & 0.377 & 0.355 & 0.517 & 0.376 & 0.356 \\
        \bottomrule
    \end{tabular}
\end{table}

\begin{table}[H]
    \caption{Average iteration time of ALS with sketching-based Kronecker regression
    for the COIL-100 tensor with different multilinear ranks (seconds).}
    \label{tab:coil_time}
    \centering
    \begin{tabular}{cccccccccccccccc}
        \toprule
        & & & \multicolumn{3}{c}{\FastKroneckerRegression} & \multicolumn{3}{c}{\texttt{DJSSW19}} \\
        rank & HOOI & ALS & 1024 & 4096 & 16384 & 1024 & 4096 & 16384 \\
        \midrule
1,1,1,1 & 2.455 & 10.975 & 10.128 & 10.147 & 10.138 & 10.131 & 10.264 & 10.195 \\
4,2,2,1 & 6.100 & 11.718 & 11.122 & 11.153 & 11.164 & 11.104 & 10.941 & 11.080 \\
8,2,2,1 & 12.092 & 11.727 & 11.137 & 11.120 & 11.164 & 11.111 & 11.069 & 11.120 \\
8,4,4,1 & 10.877 & 13.873 & 13.096 & 12.954 & 13.051 & 12.924 & 13.067 & 13.986 \\
8,4,4,2 & 10.906 & 19.614 & 17.862 & 17.784 & 17.984 & 17.957 & 18.575 & 22.107 \\
16,4,4,2 & 19.225 & 19.921 & 18.367 & 18.299 & 18.525 & 19.922 & 25.639 & 48.483 \\
        \bottomrule
    \end{tabular}
\end{table}

\end{document}